\documentclass[12pt]{article}%
\usepackage{amssymb}
\usepackage{amsfonts}
\usepackage[nohead]{geometry}
\usepackage[singlespacing]{setspace}
\usepackage[bottom]{footmisc}
\usepackage{indentfirst}
\usepackage{array}
\usepackage{graphicx}
\usepackage{rotating}
\usepackage{amsmath, amsthm, amssymb}
\usepackage{mathrsfs}
\usepackage{booktabs}
\usepackage[flushleft]{threeparttable}
\usepackage{float}
\usepackage{verbatim}
\usepackage{epstopdf}
\usepackage{pdflscape}

\allowdisplaybreaks[4]

\newtheorem{Theorem}{Theorem}[section]

\newtheorem{lemma}{Lemma A.\!\!}
\newtheorem{proof of lemma}{Proof of Lemma A.\!\!}
\newtheorem{Corollary}[Theorem]{Corollary}

\usepackage{multirow}
\usepackage{algorithm}
\usepackage{algorithmic}
\usepackage{bm}
\usepackage{subcaption}
\usepackage{enumitem}

\newcommand{\ind}[1]{I(#1)}

\usepackage[usenames,dvipsnames]{color}

\usepackage[round]{natbib}
\usepackage[hyperfootnotes=false]{hyperref}
\hypersetup{
	colorlinks,
	citecolor=Blue,
	linkcolor=Blue,
	urlcolor=Blue}

\title{A multiplier bootstrap score test for Tobit specification}

\author{Xuewen Li \\ School of Business, North Minzu University, Yinchuan, 750021, China
	\and Xiaojun Song \\ Guanghua School of Management, Peking University, Beijing, 100871, China
	\and Xiaogang Wang\thanks{Correspondence author. Email: \texttt{wangxg9102@163.com}.} \\ School of Mathematics and Information Science, North Minzu University, \\ Yinchuan, 750021, China
	\and Feipeng Zhang \\ School of Economics and Finance, Xi'an Jiaotong University, Xi'an, 710049, China}
	
\begin{document}

\begin{spacing}{1.2}

\maketitle

\begin{abstract}
In this paper, we propose a simple score-based specification test for the Tobit regression model, implemented through a computationally efficient multiplier bootstrap. We  test correct specification through the conditional moment restriction that the generalized Tobit residual has zero conditional mean, and construct a feasible score-based empirical process using the integrated conditional moments approach, conditioning on the estimated index so that no smoothing or tuning parameter is required. We derive a first-order uniform representation of the feasible process, showing how its limiting behavior is affected by maximum-likelihood estimation of the structural parameter, and obtain the asymptotic properties of the corresponding Kolmogorov--Smirnov statistic under the null, fixed alternatives, and a sequence of local alternatives. Since the limiting null distribution is non-pivotal, $p$-values are obtained from a multiplier bootstrap that mimics the first-order representation and therefore captures the parameter estimation uncertainty. 
Monte Carlo experiments, including an unknown-scale design, an index-visible/index-invisible power study, and comparisons with existing specification tests, confirm accurate size and competitive power across a range of designs and covariate dimensions. Two empirical applications further demonstrate the practical utility of our approach.
\vspace{2mm} \newline
\end{abstract}

\thispagestyle{empty}


\normalsize
\section{Introduction}
\label{s1.1}

Censored dependent variable models, i.e., regression models in which the range of the dependent variable is constrained to a subset, have broad applications in cross-sectional studies of economic behavior. For example, both women's working hours (43.2\% of observations are zero) and the number of extramarital affairs (75.0\% of observations are zero) exhibit significant left-censoring characteristics. These zero values are not the result of missing at random; rather, they reflect substantive behavioral boundaries. 
For censored data, ordinary least squares regression yields systematically biased estimates, whereas logistic regression reduces outcomes to binary form and discards all intensity information. Meanwhile, assuming a unified underlying mechanism, the two-stage regression model is both theoretically less efficient and less parsimonious than the Tobit model.

Therefore, when the data contain many boundary values, the Tobit regression model can consistently estimate the overall impact of explanatory variables on the underlying continuous behavior, balancing estimation consistency with predictive reasonableness. The unique advantage of the Tobit regression model lies in its specification of a continuous latent variable (such as ``willingness to work'' or ``extramarital affairs inclination''), which simultaneously determines both the participation decision (whether the value is greater than zero) and the intensity of the behavior (specific duration), using the same set of parameters to uniformly explain both the participation effect and the intensity effect. Owing to its accurate characterization of censored data, the Tobit model remains highly regarded in economics and statistics.

It is well known that violations of the assumptions undermine the robustness of Tobit regression estimators, leading to inconsistent parameter estimates and systematically biased predictions \cite{Nelson1981,Powell1984,Powell1986}.  It is therefore important to have diagnostic tools that can detect misspecification of the fitted Tobit model and to apply them before drawing inferential conclusions. Recent research on specification testing for Tobit regression models has primarily focused on two lines of development: local and global testing methods.

Local methods typically rely on local smoothing to identify misspecification of the regression function. 
Details of the local method can be found in \cite{Hart2013SSBM}. \cite{Song2012} proposed a specification test for the functional form of Tobit regressions using local smoothing methods, establishing the asymptotic properties of the test under both null and alternative hypotheses. \cite{Koul2014JMA} enhanced the specification test for Tobit regression models using local smoothing methods, establishing that the test statistic follows an asymptotic normal distribution under the null hypothesis, and demonstrated the test's consistency against fixed alternatives and derived its asymptotic power against local nonparametric alternatives. The effectiveness of the local method is severely compromised by high-dimensional covariates, leading to the curse of dimensionality.

Global methods test the null hypothesis by transforming conditional expectations into moment conditions and constructing test statistics as weighted averages of residuals. The residual-marked empirical process approach dates back to \citet{Stute1997}, whose construction of integrated conditional moments we follow here. \cite{Song2011JSPI} proposed a computationally efficient distribution-free test for the Tobit mean regression model, based on the Khmaladze-type transformation of a partial sum process of calibrated residuals; however, the transformation is designed for a scalar covariate, so the test is not directly applicable to high-dimensional covariates. \cite{Wang2007SPL} proposed a simple nonparametric test for checking the nonlinearity in the Tobit regression model using least absolute deviation estimation. \cite{Sun2021JKSS} constructed a Cram\'{e}r--von Mises-type specification test statistic to check the adequacy of Tobit regression models, deriving the asymptotic properties of the test under both null and alternative hypotheses. \cite{Acerenza2025arXiv} investigated the identification assumptions in Tobit regression models and established sharp testable implications capable of detecting all possible observable violations of the identifying conditions. \cite{Dominguez2020} established a general framework for model specification testing through conditional moment restrictions, deriving the asymptotic properties of the proposed tests. The global method, which uses moment conditions based on residuals, 
is free of user-chosen tuning parameters and achieves a faster convergence rate. However, the asymptotic behavior and the implementation of the test statistic depend on the consistent estimator $\hat\theta_n$ of the true nuisance parameter $\theta_0$, giving rise to the so-called ``parameter estimation uncertainty'' (PEU, hereafter). 

Although the multiplier bootstrap method, which mimics the first-order uniform representation of the empirical process to address the PEU issue, has been applied to specification tests, its application to score-based empirical processes remains relatively limited. \cite{Sun2018CSDA} developed an integrated testing procedure for right-censored parametric models by incorporating a projection-based empirical process. \cite{Song2019JOE} introduced nonparametric diagnostic tools to evaluate the asymptotic validity of treatment effect estimators that rely on correct propensity score specification, and developed corresponding specification tests. \cite{Sun2021Sinica} developed a projection-based testing framework for partially linear models, established the test's consistency, demonstrated its capability to detect Pitman-local alternatives, and derived the asymptotic distributions of the test statistic under the null hypothesis, local alternatives, and global alternative hypotheses. \cite{escanciano2026specification} developed a nonparametric testing framework to assess model specification correctness under conditional moment restrictions, constructing test statistics that exhibit asymptotic invariance to the estimation methods used for nuisance parameters. Recently, empirical process-type specification tests implemented with multiplier bootstrap procedures have attracted increasing attention because they alleviate the curse of dimensionality, capture the uncertainty associated with preliminary parameter estimates, and offer strong power. To the best of our knowledge, research on specification testing for Tobit regression using score-based empirical process and the multiplier bootstrap method to mimic its behavior remains quite limited.

Our work is also related to the classical literature on conditional moment and score tests \citep{Newey1985,Tauchen1985,Bierens1990,Donald2003}, to the residual-marked empirical process tests \citep{Stute1997,Stute1998,Delgado2001,khmaladze1981}, and to the multiplier (wild) bootstrap literature \citep{Chernozhukov2013,Chernozhukov2017}. Relative to these papers, the contribution here is threefold: we construct a score-based (generalized-residual) process for the Tobit model, reduce the conditioning to a one-dimensional index so that no smoothing or bandwidth selection is required, and develop a multiplier bootstrap that reproduces the full first-order representation, including the parameter-estimation-uncertainty correction, with a complete asymptotic theory under the null, fixed and local alternatives. Compared with local-smoothing tests, the proposed procedure involves no bandwidth selection and is not subject to the curse of dimensionality; compared with existing integrated-moment tests, it conditions on the estimated index, so that the statistic is a one-dimensional empirical process whose multiplier bootstrap provably reproduces the parameter-estimation-uncertainty correction. We stress that the restriction being tested is a necessary but not sufficient condition for correct specification, so that non-rejection is evidence of no index-visible misspecification rather than proof that the Tobit model is correctly specified. 

Specifically, we develop a score-based empirical process from a conditional moment restriction on the generalized residual and establish its asymptotic properties, as well as those of the corresponding test statistic, under the null, the alternative, and a sequence of local alternatives. We construct a multiplier bootstrap procedure that mimics the first-order uniform representation of the feasible process to compute $p$-values for the test statistic, thereby successfully addressing the presence of PEU. Numerical simulations and empirical applications demonstrate that the proposed score-based specification test captures the PEU, maintains accurate size, and delivers competitive power, including in designs with moderately many covariates.

The rest of this paper is organized as follows. Section 2 introduces the testing framework and the score-based empirical process, and Section 3 derives the asymptotic properties of the corresponding test statistic under the null, fixed and local alternatives, together with the multiplier bootstrap and its unknown-scale extension. Section 4 reports an extensive simulation study, Section 5 illustrates the empirical relevance of the tests through two real-data case studies, and Section 6 concludes. All technical proofs are collected in Appendix A, and the replication code and data are provided as supplementary material.

\section{The Score-based Specification Test}

\subsection{Testing Framework}

To analyze the relationship between the limited dependent variables $Y_i = \max(0,Y_i^*)$ and the covariates $\bm{X}_i$, it is popular for applied researchers to use the linear Tobit regression model:
\begin{align}\label{tobit_regression}
Y_i^* &= \bm{X}_i ^\top \bm\beta + e_i, 
\end{align}
where $Y_i^*$ is a latent scalar response, $\bm{X}_i$ is a $p$-dimensional covariate, and $e_i$ is a random error satisfying $e_i \sim N(0, \sigma^2)$, the normal distribution with mean 0 and variance $\sigma^2$. The latent response $Y_i^*$ can only be observed when it is larger than 0. In \eqref{tobit_regression}, $\bm\beta$ and $\sigma^2$ are model parameters. Henceforth, we concentrate on estimating the structural parameter $\bm\beta$. For theoretical development, we fix the scale by unit-variance normalization $\sigma=1$, which sets the scale of the index $\bm X_i^\top\bm\beta$; this normalization is without loss of generality only when the scale is known. The two-parameter version (with both $\bm\beta$ and $\sigma$ unknown) is developed in Section \ref{subsec:two-parameter}, where we also explain precisely what the $\sigma=1$ test does and does not test, and is used in the empirical applications of Section \ref{s:EmpApp}.

However, in empirical analysis, the Tobit regression model is generally adopted for convenience rather than on theoretical grounds. Statistical inference based on the assumed but potentially misspecified model can lead to misleading conclusions. 
It is therefore important to have diagnostic tools that can detect violations of the model's distributional and functional-form assumptions. As such, this paper aims to test whether the assumed linear functional form $\bm{X}_i^\top \bm{\beta}$ in the Tobit regression model \eqref{tobit_regression} is correctly specified. 

To proceed, let $W_i =(Y_i,\bm{X}_i^\top)^\top$ and $d_i = I(Y_i>0)$. Define the inverse Mills ratio $\lambda(u) =  \phi(u)/\Phi(u)$, where $\phi(\cdot)$ and $\Phi(\cdot)$ denote the probability density function and the cumulative distribution function of the standard normal random variable, respectively. Under the null hypothesis of correct specification with parameter $\bm\beta_0$ (and $\sigma=1$), $d_i = I(Y_i^*>0) = I(e_i > -\bm{X}_i^\top \bm\beta_0)$, and since $e_i$ is standard normal and independent of $\bm X_i$,
\begin{align}\label{eqn:d-mean}
	\mathbb E[d_i|\bm X_i] = \mathbb P(e_i > -\bm X_i^\top\bm\beta_0|\bm X_i) = 1-\Phi(-\bm X_i^\top\bm\beta_0) = \Phi(\bm X_i^\top\bm\beta_0),
\end{align}
so that, equivalently, $1-\mathbb E[d_i|\bm X_i] = \Phi(-\bm X_i^\top\bm\beta_0)$. This property is used repeatedly below; in particular, it is the key to the definition of the quantities $G(t,\bm\beta_0)$ and $H(\bm\beta_0)$ defined in Section 2.2 and of their sample analogs. 

Let 
\[u(\bm{W}_i,\bm\beta_0) = d_i (Y_i-\bm{X}_i^\top \bm\beta_0) - (1-d_i)\lambda(-\bm{X}_i^\top \bm\beta_0)\]
denote the score function. The motivation for this quantity is twofold. First, it is the scalar component of the score of the Tobit log-likelihood: the individual log-likelihood contribution is
\[\ell_i(\bm\beta) = d_i\left[-\frac12\left(Y_i-\bm X_i^\top\bm\beta\right)^2\right] + (1-d_i)\ln\Phi(-\bm X_i^\top\bm\beta),\]
and differentiating with respect to the linear predictor gives $\partial\ell_i(\bm\beta)/\partial(\bm X_i^\top\bm\beta) = u(\bm W_i,\bm\beta)$, so that the score vector is $\nabla_{\bm\beta}\ell_i(\bm\beta) = \bm X_i\,u(\bm W_i,\bm\beta)$; hence the name score function. Second, and more intuitively, $u(\bm W_i,\bm\beta_0)$ is the natural generalized residual of the Tobit model: for an uncensored observation ($d_i=1$) it reduces to the usual regression residual $Y_i-\bm X_i^\top\bm\beta_0$, whereas for a censored observation ($d_i=0$), for which only $Y_i^*\le0$ is observed, it equals $-\lambda(-\bm X_i^\top\bm\beta_0)$: this is precisely the conditional mean of the latent error given censoring,
\[\mathbb E\left[e_i\,\big|\,e_i\le -\bm X_i^\top\bm\beta_0,\ \bm X_i\right] = -\frac{\phi(-\bm X_i^\top\bm\beta_0)}{\Phi(-\bm X_i^\top\bm\beta_0)} = -\lambda(-\bm X_i^\top\bm\beta_0),\]
so that each censored observation is assigned the expected value of its unobserved error, the role that the residual plays in uncensored regressions. 

Under the correct specification of the linear functional form $\bm{X}_i^\top \bm{\beta}_0$ for some $\bm{\beta}_0\in \mathcal{B}$, where $\mathcal{B}$ is the parameter space, assumed henceforth to be a compact subset of $\mathbb{R}^p$, the generalized residual $u(\bm{W}_i,\bm\beta_0)$ 
satisfies the following conditional moment restriction almost surely (a.s.): 
\begin{align}\label{eqn:SpecTestP}
H_0: \mathbb E [u(\bm{W}_i,\bm\beta_0)|\bm{X}_i] = 0 \quad \text{ a.s. for some}\  \bm\beta_0 \in \mathcal{B},
\end{align}
which is the null hypothesis of interest in this paper. 
Indeed, correct specification of $\bm{X}_i^\top \bm{\beta}_0$ implies the conditional moment restriction \eqref{eqn:SpecTestP}, because
\begin{align*}
   & \mathbb E [u(\bm{W}_i,\bm\beta_0)|\bm X_i]\notag\\
    =&\mathbb E  \left[d_i (Y_i - \bm{X}_i^\top \beta_0)  - (1-d_i)\lambda(-\bm{X}_i^\top \bm\beta_0) |\bm X_i\right] \notag\\
    =& \mathbb E[I(e_i > - \bm{X}_i^\top \bm\beta_0) e_i|\bm X_i] - [1-\mathbb{E} (d_i|\bm X_i)]
     \frac{\phi(- \bm{X}_i^\top \bm\beta_0)}{\Phi(- \bm{X}_i^\top \bm\beta_0)} \notag\\
    =& \int_{- \bm X_i^\top \bm\beta_0}^{\infty} s \phi(s) \text{d}s -  (1-\Phi(\bm{X}_i^\top \bm\beta_0)) \frac{\phi(- \bm X_i^\top \bm\beta_0)}{1-\Phi(\bm{X}_i^\top \bm\beta_0)}\notag\\ 
    =& \phi(- \bm{X}_i^\top \bm\beta_0) -  
     \phi(- \bm{X}_i^\top \bm\beta_0) = 0 \quad \text{a.s.}
\end{align*}
The alternative hypothesis $H_1$ is the negation of the null hypothesis $H_0$ in \eqref{eqn:SpecTestP}, namely,
\begin{align}\label{alter}
	H_1 : \mathbb P\{\mathbb E[u(\bm{W}_i,\bm\beta)|\bm{X}_i]\neq 0\}>0 \text{ for all}\  \bm\beta \in \mathcal{B}.
\end{align}

Since conditioning on a coarser variable preserves the restriction, in the same spirit of \cite{stute2002model}, the conditional moment restriction in \eqref{eqn:SpecTestP} implies 
\[\mathbb E [u(\bm{W}_i,\bm\beta_0)|\bm{X}_i^\top \bm\beta_0] = 0 \quad \text{ a.s. for some}\  \bm\beta_0 \in \mathcal{B},\]
reducing the $p$-dimensional $\bm{X}_i$ conditioning variable to the one-dimensional $\bm{X}_i^\top \bm\beta_0$ conditioning variable. Our score-based specification test aims to assess the above dimension-reducing conditional moment restriction. The asymptotic distributional results of Section 3 are derived under correct specification with parameter $\bm\beta_0$, which implies $H_0$; the hypotheses $H_0$ and $H_1$ describe the restriction that the test assesses and the class of alternatives it can detect, as summarized in Remark~1. 

As mentioned in Section \ref{s1.1}, the local test method requires smoothing the data, which often compromises the accuracy of the fit and is sensitive to the selection of the smoothing parameter. Inspired by the integrated conditional moment approach introduced in \cite{BIERENS1982105} and developed by many others, including \cite{Escanciano2006ET}, \cite{Dominguez2020}, and \cite{escanciano2026specification}, we avoid the local smoothing techniques by transforming the dimension-reducing conditional moment restriction 
into an infinite number of unconditional moment restrictions and checking all of them 
simultaneously, i.e.,
\begin{align*}
	 \mathbb E[u(\bm{W}_i,\bm\beta_0)|\bm{X}_i^\top \bm\beta_0] = 0 \text{ a.s.} \quad \Leftrightarrow \quad
		 \mathbb E[u(\bm{W}_i,\bm\beta_0)I(\bm{X}_i^\top \bm\beta_0 \leq t)] = 0\text{ for all } t\in\mathbb R,
\end{align*}
where $I(\cdot)$ is the indicator function. Our specification test is then constructed using a feasible sample version of the score-type dimension-reducing unconditional moment restrictions, given in the next subsection.

{\bf Remark 1:} Because our test conditions on the one-dimensional index $\bm{X}_i^\top \bm\beta_0$ rather than the full variable $\bm{X}_i$, it is consistent against the alternatives satisfying $\mathbb E[u(\bm{W}_i,\bm\beta) I(\bm{X}_i^\top \bm\beta \leq t)]\neq 0  \text{ for all}\  \bm\beta \in \mathcal{B} \text{ and for some } t\in \mathbb{R}$, but it has no power against misspecified models whose conditional-mean deviation $g_{\bm\beta}(\bm X)=\mathbb E[u(\bm W,\bm\beta)|\bm X]$ is orthogonal to the index, that is, $\mathbb E[g_{\bm\beta}(\bm X)|\bm X^\top\bm\beta]=0$ a.s.\ for all $\bm\beta\in\mathcal B$. For instance, if $\bm\beta_0=(1,0)^\top$ and $X_1$ is independent of $X_2$, a deviation depending only on $X_2$ and with zero mean is undetectable at any rate. Section \ref{sec:indexinvisible} illustrates this limitation numerically. 




\subsection{The Score-based Specification Test}

If the true parameter $\bm\beta_0$ is known, it is natural to test the null hypothesis $H_0$ in \eqref{eqn:SpecTestP} based on the (infeasible) score-based empirical process
\begin{align}\label{eqn:infeasible-process}
	U_n(t,\bm\beta_0) = \frac{1}{\sqrt{n}} \sum_{i=1}^n u(\bm{W}_i,\bm\beta_0) I(\bm{X}_i^\top \bm\beta_0 \leq t),
\end{align}
whose asymptotic null behavior is readily derived. However, $\bm\beta_0$ is usually unknown in practice and is typically estimated from the available sample via the maximum likelihood method. 
Let $\hat{\bm\beta}$ denote the maximum likelihood estimator for $\bm\beta_0$. When $\bm\beta_0$ is replaced by the $\sqrt n$-consistent estimator $\hat{\bm\beta}$, the validity of \eqref{eqn:SpecTestP} should be measured by the distance of the feasible score-based empirical process
\begin{align}\label{eqn_feasible_stats}
	U_n(t,\hat{\bm\beta}) = \frac{1}{\sqrt{n}} \sum_{i=1}^n u(\bm{W}_i,\hat{\bm\beta}) I(\bm{X}_i^\top \hat{\bm\beta} \leq t)
\end{align}
from zero, whose asymptotic null behavior is, however, generally different from that of $U_n(t,\bm\beta_0)$ due to using the estimator $\hat{\bm\beta}$.

With the help of $U_n(t,\hat{\bm\beta})$, we then reject the null hypothesis $H_0$ for ``large'' values of suitable test statistics, which gauge the distance of $U_n(t,\hat{\bm\beta})$ in \eqref{eqn_feasible_stats} from zero. For example, we can consider the following Kolmogorov--Smirnov (KS)-type test statistic:
\begin{align}\label{eqn:KSn-stat}
	KS_n = \max_{1\leq j \leq n} |U_n(\bm{X}_j^\top \hat{\bm\beta},\hat{\bm\beta})|.
\end{align}
The feasible score-based empirical process $U_n(t,\hat{\bm\beta})$ in \eqref{eqn_feasible_stats} and the corresponding KS statistic \eqref{eqn:KSn-stat} are the main objects of our interest. The only inference problem is the estimation of the asymptotic null distribution of $KS_n$, which, as explained above, must account for the fact that $\hat{\bm\beta}$ is random and introduces an additional layer of estimation uncertainty. Note that the Cram\'{e}r--von Mises (CvM)-type test statistic 
$$CvM_n =\frac{1}{n}\sum_{j=1}^n [U_n(\bm{X}_j^\top \hat{\bm\beta},\hat{\bm\beta})]^2$$ can also be considered, whose asymptotic null distribution can be analyzed similarly; we focus on the $KS_n$ statistic for simplicity.


Given the Assumptions in Appendix A and under correct specification with parameter $\bm\beta_0$, where the two population forms of $G(t,\bm\beta_0)$ coincide, we have the following uniform decomposition:
\begin{align}\label{eqn_Unjieshi}
	U_n(t,\hat{\bm\beta}) = U_n(t,\bm\beta_0) 
		- G^\top(t,\bm\beta_0)\sqrt{n} (\hat{\bm\beta}-\bm\beta_0) +o_p(1),
\end{align}
uniformly in $t\in \mathbb{R}$, where 
\begin{align*}
G(t,\bm\beta_0) = \mathbb E\left[(d-(1-d)\dot{\lambda}(-\bm{X}^\top \bm\beta_0)) \bm{X}  I(\bm{X}^\top \bm\beta_0 \leq t)\right],
\end{align*}
with $\dot{\lambda}(x) =  -x\lambda(x) -\lambda^2(x)$, and, since $\hat{\bm\beta}$ is the score-based maximum likelihood estimator, the influence expansion
\begin{align*}
\sqrt{n} (\hat{\bm\beta}-\bm\beta_0)=H^{-1}(\bm{\beta}_0)\frac{1}{\sqrt{n}} \sum_{i=1}^n u(\bm{W}_i,\bm{\beta}_0)\bm{X}_i +o_p(1)
\end{align*}
holds with 
$$H(\bm\beta_0)=\mathbb E\left[(d-(1-d)\dot{\lambda}(-\bm{X}^\top \bm\beta_0))\bm{X}\bm{X}^\top\right].$$
The definitions of $G(t,\bm\beta_0)$ and $H(\bm\beta_0)$ are motivated by the key property \eqref{eqn:d-mean}:  under correct specification, the conditional mean of the random weight $d-(1-d)\dot\lambda(-\bm X^\top\bm\beta_0)$ is
\begin{align*}
	\ell(\bm X)=\mathbb E\left[d-(1-d)\dot\lambda(-\bm X^\top\bm\beta_0)\big|\bm X\right]=\Phi(\bm X^\top\bm\beta_0) - \Phi(-\bm X^\top\bm\beta_0)\dot\lambda(-\bm X^\top\bm\beta_0),
\end{align*}
so that $G(t,\bm\beta_0)=\mathbb E[\ell(\bm X)\bm X I(\bm X^\top\bm\beta_0\le t)]$ and $H(\bm\beta_0)=\mathbb E[\ell(\bm X)\bm X\bm X^\top]$. Thus $H(\bm\beta_0)$ is the Fisher information matrix of the Tobit model, and the two equivalent expressions of these quantities (with the random $d$ or with its conditional mean $\Phi$) can be used interchangeably in the population versions under  correct specification. In addition, a direct calculation (Lemma A.\ref{lem:varid} in Appendix A) gives the conditional variance identity
\begin{align}\label{eqn:var-identity}
\mathbb E\left[u^2(\bm W,\bm\beta_0)\,\big|\,\bm X\right]=\ell(\bm X)\quad\text{a.s.},
\end{align}
so that $\ell(\bm X)$ is simultaneously the information weight and the conditional variance of the generalized residual; identity \eqref{eqn:var-identity} is the key to the covariance comparison between the feasible and the infeasible processes (Corollary \ref{coro.kernel}).
Substituting the influence expansion into \eqref{eqn_Unjieshi} yields the following first-order uniform representation 
\begin{align}\label{eqn:first-order-rep}
U_n(t,\hat{\bm\beta})= \frac{1}{\sqrt{n}}\sum_{i=1}^n u(\bm W_i,\bm\beta_0)\Big[I(\bm X_i^\top \bm\beta_0 \leq t) - G^\top(t,\bm\beta_0) H^{-1}(\bm\beta_0) \bm X_i\Big] + o_p(1),
\end{align}
uniformly in $t$. Hence the asymptotic null distribution of the $KS_n$ test statistic in \eqref{eqn:KSn-stat} is not the same as that of the infeasible process $U_n(t,\bm\beta_0)$ with known parameters: the parameter estimation contributes the correction term
\[G^\top(t,\bm\beta_0)\,H^{-1}(\bm\beta_0)\,\frac{1}{\sqrt{n}}\sum_{i=1}^n u(\bm W_i,\bm\beta_0)\bm X_i,\]
which is of order $O_p(1)$ and does not vanish asymptotically. Note that this term is commonly known as PEU in the literature on empirical process-based specification tests; ignoring this leads to invalid testing procedures.


To deal with PEU in a computationally efficient way, we can use the multiplier bootstrap to mimic the first-order uniform representation in \eqref{eqn:first-order-rep} directly: its bootstrap analog replaces the unobservable quantities $\bm\beta_0$, $G(t,\bm\beta_0)$, and $H(\bm\beta_0)$ by their sample counterparts $\hat{\bm\beta}$, $G_n(t,\hat{\bm\beta})$, and $H_n(\hat{\bm\beta})$, where
\begin{align*}
	& G_n(t,\hat{\bm\beta}) = \frac1n \sum_{i=1}^n \left[d_i-(1-d_i)\dot{\lambda}(-\bm{X}_i^\top \hat{\bm\beta})\right] \bm{X}_i  I(\bm{X}_i^\top \hat{\bm\beta} \leq t), \\
	& H_n(\hat{\bm\beta}) = \frac1n \sum_{i=1}^n \left[d_i-(1-d_i)\dot{\lambda}(-\bm{X}_i^\top \hat{\bm\beta})\right]\bm{X}_i \bm{X}_i^\top.
\end{align*}
The sample versions above use the random weights $d_i$ and are consistent for  the random-weight limits $G^{w}(t,\bm\beta_0)$ and $H^{w}(\bm\beta_0)$, which coincide with $G(t,\bm\beta_0)$ and $H(\bm\beta_0)$ under correct specification (Lemma A.1). Alternatively, since $\mathbb E[d_i|\bm X_i]=\Phi(\bm X_i^\top\bm\beta_0)$ under  correct specification, one may replace the random weights by their conditional means and use
\begin{align}\label{eqn:Gn-expected}
	& G_n(t,\hat{\bm\beta}) = \frac1n \sum_{i=1}^n \left[\Phi(\bm{X}_i^\top \hat{\bm\beta})-\Phi(-\bm{X}_i^\top \hat{\bm\beta}) \dot{\lambda}(-\bm{X}_i^\top \hat{\bm\beta})\right] \bm{X}_i  I(\bm{X}_i^\top \hat{\bm\beta} \leq t), \\
	& H_n(\hat{\bm\beta}) = \frac1n \sum_{i=1}^n \left[\Phi(\bm{X}_i^\top \hat{\bm\beta})-\Phi(-\bm{X}_i^\top \hat{\bm\beta}) \dot{\lambda}(-\bm{X}_i^\top \hat{\bm\beta})\right]\bm{X}_i \bm{X}_i^\top,\notag
\end{align}
which is the expected-information estimator. The two sample versions are asymptotically equivalent under  correct specification: the conditional-expectation version is a sample mean of functions of $\bm X_i$ and is consistent for $G(t,\bm\beta_0)$ and $H(\bm\beta_0)$ without any distributional assumption, whereas the random-weight version has the same limit only under  correct specification, where $\mathbb E[d_i|\bm X_i]=\Phi(\bm X_i^\top\bm\beta_0)$; in the proofs we use the random-weight version, which is also consistent under the alternative (with the corresponding random-weight limits). 

Consequently, the multiplier bootstrap process to mimic the behavior of $U_n(t,\hat{\bm\beta})$ is 
\begin{align}\label{eqn:multiplier.process}
	U_n^{\ast}(t,\hat{\bm\beta})  = \frac{1}{\sqrt{n}} \sum_{i=1}^n V_i u(\bm{W}_i,\hat{\bm\beta})\left[I(\bm{X}_i^\top \hat{\bm\beta} \leq t)-  G_n^\top(t,\hat{\bm\beta}) H_n^{-1}(\hat{\bm\beta}) \bm{X}_i\right],
\end{align}
with $\{V_i\}_{i=1}^n$ a sequence of i.i.d.\ random variables with mean zero, variance one, satisfying Assumption (A7) and independent of the original sample. Built on $U_n^{\ast}(t,\hat{\bm\beta})$ in \eqref{eqn:multiplier.process}, the multiplier bootstrap KS test statistic is
\begin{align}\label{eqn:bootstrap-Kn}
	KS_n^{\ast} = \max_{1\leq j \leq n} |U_n^{\ast}(\bm{X}_j^\top \hat{\bm\beta},\hat{\bm\beta})|.
\end{align}
The correction term $G_n^\top(t,\hat{\bm\beta})H_n^{-1}(\hat{\bm\beta})\bm X_i$ inside the bootstrap summands is essential: it reproduces, within each bootstrap draw, the contribution
\[-G^\top(t,\bm\beta_0)H^{-1}(\bm\beta_0)\,n^{-1/2}\sum_{i=1}^n V_i u(\bm W_i,\bm\beta_0)\bm X_i\]
of the parameter estimation to the first-order representation \eqref{eqn:first-order-rep}. The $p$-value of the test is obtained as the proportion of bootstrap draws exceeding the observed statistic, as detailed in Section \ref{subsection_bootstrap_pro}.


\section{Asymptotic Properties}
In this section, we examine the asymptotic properties of the feasible score-based empirical process $U_n(t,\hat{\bm\beta})$ as well as those of the $KS_n$ test statistic under the null, the alternative, and a sequence of local alternatives. Throughout, $U_n\Rightarrow U_\infty$ in $\ell^\infty(\mathbb R)$ denotes weak convergence of the stochastic process viewed as a random element of the space $\ell^\infty(\mathbb R)$ of bounded functions on $\mathbb R$ equipped with the supremum norm $\|f\|_\infty=\sup_{t\in\mathbb R}|f(t)|$ and its Borel $\sigma$-field, i.e., $\mathbb E[h(U_n)]\to\mathbb E[h(U_\infty)]$ for every bounded continuous functional $h$ on $\ell^\infty(\mathbb R)$, see, e.g., Chapter 1.5 in \cite{Vaart1996Weak}. Consequently, the continuous mapping theorem, see, e.g., Theorem 1.3.6 in \cite{Vaart1996Weak}, applies to functionals such as the Kolmogorov--Smirnov supremum $f\mapsto\sup_{t\in\mathbb R}|f(t)|$.

\subsection{Asymptotic Null Distribution}
This section formally establishes the asymptotic null distribution of $U_n(t,\hat{\bm\beta})$ and of $KS_n$. We first provide the weak convergence of $U_n(t,\hat{\bm\beta})$ to a mean-zero Gaussian process under  correct specification of the Tobit model with parameter $\bm\beta_0$, which implies the null restriction $H_0$ in \eqref{eqn:SpecTestP}.

\begin{Theorem}\label{Thm.weakconvergence.stat}
	 Suppose that the model is correctly specified with parameter $\bm\beta_0$ and that Assumptions (A1)--(A6) in Appendix A hold. Then we have
	\[U_n(t,\hat{\bm\beta}) \Rightarrow U_\infty(t,\bm\beta_0) \text{ in } \ell^\infty(\mathbb R),\] 
	where $U_\infty(t,\bm\beta_0)$ is a centered Gaussian process with covariance function $\mathbb E[\Sigma(t_1) \Sigma(t_2)]$, where
\begin{align*}
	&\Sigma(t) = u(\bm{W},\bm{\beta}_0)\left[I(\bm{X}^\top \bm{\beta}_0 \leq t) - G^\top (t,\bm{\beta}_0) H^{-1}(\bm{\beta}_0)\bm{X}\right],\\
	& G(t,\bm\beta_0) = \mathbb E\{[d-(1-d)\dot{\lambda}(-\bm{X}^\top \bm{\beta}_0)] \bm{X}  I(\bm{X}^\top \bm{\beta}_0 \leq t)\}
    = \mathbb E\{\ell(\bm X)\bm X I(\bm X^\top\bm\beta_0 \leq t)\},\\
	& H(\bm\beta_0) = \mathbb E\{[d-(1-d) \dot{\lambda}(-\bm{X}^\top \bm\beta_0)]\bm{X} \bm{X}^\top\}
    = \mathbb E\{\ell(\bm X)\bm X \bm X^\top\},
\end{align*}
with $\ell(\bm X)=\mathbb E[d-(1-d)\dot\lambda(-\bm X^\top\bm\beta_0)|\bm X]=\Phi(\bm X^\top\bm\beta_0)-\Phi(-\bm X^\top\bm\beta_0)\dot\lambda(-\bm X^\top\bm\beta_0)$ by \eqref{eqn:d-mean}.
\end{Theorem}

The following corollary is an immediate consequence of Theorem \ref{Thm.weakconvergence.stat}.

\begin{Corollary}\label{coro.kernel}
Under the assumptions of Theorem \ref{Thm.weakconvergence.stat}, the covariance function of $U_\infty$ admits the closed form
\begin{align}\label{eqn:closed-kernel}
\mathbb E[\Sigma(t_1)\Sigma(t_2)]
&=\mathbb E\left[u^2(\bm W,\bm\beta_0)I(\bm X^\top\bm\beta_0\le t_1)I(\bm X^\top\bm\beta_0\le t_2)\right]\notag\\
&\quad-G^\top(t_1,\bm\beta_0)H^{-1}(\bm\beta_0)G(t_2,\bm\beta_0).
\end{align}
Consequently, pointwise in $t$,
\begin{align}\label{eqn:var-reduction}
\mathbb E[\Sigma(t)^2]\le\mathbb E\left[u^2(\bm W,\bm\beta_0)I(\bm X^\top\bm\beta_0\le t)\right],
\end{align}
with equality if and only if $G(t,\bm\beta_0)=0$; the two kernels coincide for all $(t_1,t_2)$ if and only if $G(\cdot,\bm\beta_0)\equiv0$.
\end{Corollary}

Note that under $H_0$, $\mathbb E[\Sigma(t)]=0$, so $U_\infty$ is a centered Gaussian process. The correction term in $\Sigma(t)$ is precisely the contribution of the estimation of $\bm\beta_0$, and by Corollary \ref{coro.kernel} it shrinks the limiting variance: the feasible kernel in \eqref{eqn:closed-kernel} is the infeasible kernel $\mathbb E[u^2I(t_1)I(t_2)]$ minus the positive semidefinite term $G^\top(t_1,\bm\beta_0)H^{-1}(\bm\beta_0)G(t_2,\bm\beta_0)$. This variance reduction is reproduced by the corrected multiplier bootstrap of Section \ref{subsection_bootstrap_pro}. Since $U_\infty$ depends on $\bm\beta_0$ and on the unknown law of $\bm W$, the limiting null distribution of the continuous functional based on it (including our $KS_n$ statistic) is non-pivotal, which motivates the multiplier bootstrap of Section \ref{subsection_bootstrap_pro}.

With the help of Theorem \ref{Thm.weakconvergence.stat} and the continuous mapping theorem, we have the convergence in distribution of the $KS_n$ test statistic under  correct specification with parameter $\bm\beta_0$.

\begin{Corollary}\label{coro.weakconv.stat}
	 Under correct specification with parameter $\bm\beta_0$ and Assumptions (A1)--(A6) in Appendix A, we have
	\[KS_n \overset{\mathcal{D}}{\rightarrow} \sup_{t\in\mathbb R} |U_\infty(t,\bm\beta_0)|,\]
	where $U_\infty(t,\bm\beta_0)$ is the centered Gaussian process as defined in Theorem \ref{Thm.weakconvergence.stat}.
\end{Corollary}

In \eqref{eqn:KSn-stat} the empirical process $U_n(t,\hat{\bm\beta})$ is evaluated at the sample points $t=\bm X_j^\top\hat{\bm\beta}$ for $j=1,\ldots,n$. Since $U_n(\cdot,\hat{\bm\beta})$ is a right-continuous step function with jumps at exactly these points, the maximum in \eqref{eqn:KSn-stat} coincides exactly with the supremum of $|U_n(t,\hat{\bm\beta})|$ over $t\in\mathbb R$; the corresponding supremum of the limiting process is well behaved by the continuity of the distribution of $\bm X^\top\bm\beta_0$ in Assumption (A1). The form \eqref{eqn:KSn-stat} therefore greatly simplifies the computation of the KS statistic and of its multiplier-bootstrap version defined later.

\subsection{Consistency and Local Power}

In this section, we first consider the asymptotic properties of $U_n(t,\hat{\bm\beta})$ and $KS_n$ under the fixed alternative hypothesis $H_1$ in \eqref{alter}.

\begin{Theorem}\label{thm.consistency.stat}
	Suppose that $\hat{\bm\beta}$ converges in probability to a unique pseudo-true value $\bm\beta^*\in\mathcal B\setminus\{\bm 0\}$ satisfying $\mathbb E[u(\bm W,\bm\beta^*)\bm X]=0$ under the alternative hypothesis $H_1$ in \eqref{alter}. Under Assumptions (A1), (A2), (A3), and (A6) in Appendix A,
\[\sup_{t\in\mathbb R} \left| \frac{1}{\sqrt{n}} U_n(t,\hat{\bm\beta}) - \mu(t,\bm\beta^*)\right| = o_p(1),\]
where
\[\mu(t,\bm\beta^*) = \mathbb E[u(\bm W,\bm\beta^*)I(\bm X^\top\bm\beta^*\le t)]\]
is the score moment. 
Consequently, if the index restriction $\mathbb E[u(\bm W,\bm\beta^*)|\bm X^\top\bm\beta^*]=0$ fails a.s.\ (equivalently, $\mu(\cdot,\bm\beta^*)\not\equiv0$), then $\sup_{t\in\mathbb R}|\mu(t,\bm\beta^*)|>0$ and
\[ \frac{KS_n}{\sqrt n}\overset{p}{\to}\sup_{t\in\mathbb R}|\mu(t,\bm\beta^*)|>0,\]
so that $\lim_{n\to\infty}\mathbb P(KS_n>c_\tau)=1$ for any fixed critical value $c_\tau$; the test is consistent against every such alternative.
\end{Theorem}

Theorem \ref{thm.consistency.stat} does not assert consistency for alternatives in $H_1$ whose index projection vanishes at $\bm\beta^*$. Those are precisely the undetectable alternatives of Remark~1 and are discussed in Section~\ref{sec:indexinvisible}.
By Theorem \ref{thm.consistency.stat}, under a fixed alternative in $H_1$ the drift of $n^{-1/2}U_n(t,\hat{\bm\beta})$ is the score moment $\mu(t,\bm\beta^*)=\mathbb E[u(\bm W,\bm\beta^*)I(\bm X^\top\bm\beta^*\le t)]$ evaluated at the pseudo-true value $\bm\beta^*$: at this rate  the estimation error is annihilated by the $n^{-1/2}$ normalization (its contribution to $\bar U_n$ is $o_p(1)$ by the uniform continuity of the mean map and the centered stochastic equicontinuity), and no PEU correction term survives in the limit. Since $\mu(\cdot,\bm\beta^*)\equiv0$ if and only if the index conditional moment restriction $\mathbb E[u(\bm W,\bm\beta^*)|\bm X^\top\bm\beta^*]=0$ holds a.s., the test is consistent exactly against the fixed alternatives in $H_1$ whose deviation is visible through the index at the pseudo-true value, that is, those with $\mu(\cdot,\bm\beta^*)\not\equiv0$. It is important to distinguish this from the full null $H_0$: because $H_0$ requires $\mathbb E[u(\bm W,\bm\beta_0)|\bm X]=0$ a.s.\ for some $\bm\beta_0$, an alternative may fail $H_0$ at every parameter value while its index projection vanishes at $\bm\beta^*$; such an alternative belongs to $H_1$ but leaves the process asymptotically centered, so the test has no power against it. 
Thus, relative to $H_1$, the test detects the ``index-visible'' fixed alternatives, and over that set no misspecification is masked by the parameter estimation. This is in sharp contrast with local alternatives, where the PEU does not vanish, and the drift contains a correction term.

We next investigate the asymptotic behavior of the score-based specification test under a sequence of local alternative hypotheses converging 
at the parametric $n^{-1/2}$ rate:
\begin{align}\label{eqn:local.alt}
	H_{1n} : \mathbb E\left[u(\bm{W}_i,\bm\beta_0) |\bm{X}_i\right]= n^{-1/2}a(\bm{X}_i^\top \bm{\beta}_0) \quad \text{a.s. for some}\  \bm\beta_0 \in \mathcal{B},
\end{align}
where the nonzero deviation function $a(\cdot)$ satisfies $\mathbb E|a(\bm X_i^\top \bm\beta_0)|<\infty$ and $\mathbb E a^2(\bm X_i^\top \bm\beta_0)<\infty$, and is continuous with $\mathbb E[\sup_{\bm\beta\in\mathcal B}|a(\bm X_i^\top\bm\beta)|]<\infty$  (these hold, e.g., for polynomials of degree at most four under (A3)). Throughout, the local sequence is understood as a contiguous perturbation of the correctly specified model with parameter $\bm\beta_0$, so that the conditional law of $(d_i,Y_i)$ given $\bm X_i$ remains at its model value up to the $O(n^{-1/2})$ mean shift of $u$ specified by $H_{1n}$; all conditional moments entering the null quantities therefore remain at their model values up to $O(n^{-1/2})$. Since conditioning on a coarser variable preserves the restriction, $H_{1n}$ implies the index version $\mathbb E[u(\bm{W}_i,\bm\beta_0)|\bm X_i^\top\bm\beta_0]=n^{-1/2}a(\bm X_i^\top\bm\beta_0)$ a.s.

\begin{Theorem}\label{thm.localpower.stat}
	Under the sequence of local alternatives $H_{1n}$ in \eqref{eqn:local.alt}, Assumptions (A1)--(A6) in Appendix A, 
\[U_n(t,\hat{\bm\beta}) \Rightarrow U_\infty(t,\bm\beta_0) + D (t,\bm\beta_0) \text{ in } \ell^\infty(\mathbb R),\]
where $U_\infty(t,\bm\beta_0)$ is the centered Gaussian process as defined in Theorem \ref{Thm.weakconvergence.stat} and $D (t,\bm\beta_0) $ is a deterministic drift function defined by
\[D (t,\bm\beta_0) = \mathbb E[a(\bm{X}^\top \bm{\beta}_0)I(\bm{X}^\top \bm{\beta}_0 \leq t)]
- G^\top(t,\bm\beta_0)H^{-1}(\bm\beta_0)\mathbb E[a(\bm{X}^\top \bm{\beta}_0)\bm X].\]
Hence, 
\[ KS_n \overset{\mathcal{D}}{\to} \sup_{t\in\mathbb R} \left|U_\infty(t,\bm\beta_0) + D (t,\bm\beta_0)\right|. \]
\end{Theorem}

Theorem \ref{thm.localpower.stat} and Theorem \ref{thm.bootstrap.stat} in the next subsection together imply that at the $\sqrt n$ rate the limiting power of the test is
\[\mathbb P\Bigl(\sup_{t\in\mathbb R}|U_\infty(t,\bm\beta_0)+D(t,\bm\beta_0)|>q_{1-\tau}\Bigr),\]
where $q_{1-\tau}$ is the $(1-\tau)$-quantile of $\sup_{t\in\mathbb R}|U_\infty(t,\bm\beta_0)|$, and the limit is nontrivial whenever $D(\cdot,\bm\beta_0)\not\equiv0$. Indeed, for the centered Gaussian limit $U_\infty$ and the symmetric convex acceptance region $\{\sup_t|f(t)|\le q_{1-\tau}\}$, the map $D\mapsto\mathbb P(\sup_t|U_\infty(t)+D(t)|\le q_{1-\tau})$ is even and log-concave by the log-concavity of Gaussian measures (Borell's theorem; see \cite{Bogachev1998}), hence maximized at $D\equiv0$. The limiting power is therefore at least the nominal level, with equality only when $D\equiv0$. The drift contains the PEU correction,
\[D(t,\bm\beta_0)=\mathbb E[a(\bm X^\top\bm\beta_0)I(\bm X^\top\bm\beta_0\le t)]-G^\top(t,\bm\beta_0)H^{-1}(\bm\beta_0)\mathbb E[a(\bm X^\top\bm\beta_0)\bm X].\]
Note that the correction term arises because under $H_{1n}$ the maximum likelihood estimator $\hat{\bm\beta}$ drifts by $H^{-1}(\bm\beta_0)\mathbb E[a(\bm X^\top\bm\beta_0)\bm X]$. The PEU does not vanish at the parametric rate, unlike under the fixed alternatives. Writing
\[D(t,\bm\beta_0)=\mathbb E\bigl[a(\bm X^\top\bm\beta_0)\bigl(I(\bm X^\top\bm\beta_0\le t)-\bm X^\top H^{-1}(\bm\beta_0)G(t,\bm\beta_0)\bigr)\bigr],\]
the local alternatives with $D(\cdot,\bm\beta_0)\equiv0$ are exactly those whose deviation function satisfies
\[a(\bm X^\top\bm\beta_0)=\mathbb E[\ell(\bm X)\bm X^\top c\,|\,\bm X^\top\bm\beta_0]\quad\text{a.s. for some }c\in\mathbb R^p.\]
Such deviations lie in the score space of the parametric Tobit model. They are precisely the index projections of the score directions $\ell(\bm X)\bm X^\top c$, i.e., the part of a local perturbation that the parametric fit absorbs. For example, a $\sqrt n$-shift of the true coefficient, $\bm\beta\rightsquigarrow\bm\beta_0+n^{-1/2}c$, leaves the model correctly specified. The induced conditional deviation is $\ell(\bm X)\bm X^\top c$, whose index projection is $\mathbb E[\ell(\bm X)\bm X^\top c\,|\,\bm X^\top\bm\beta_0]$, and the drift depends on the deviation only through this projection, so the shift is absorbed by the estimator and the limiting rejection probability equals the nominal level. For every other local alternative, $D\not\equiv0$, the limiting power is strictly larger than the nominal level by the same argument, and it converges to one as the drift grows in the sense that $\sup_t|U_\infty(t)+\lambda D(t)|\to\infty$ almost surely as $\lambda\to\infty$ for $D\not\equiv0$. This is the usual price of capturing the PEU at the parametric rate. 

\subsection{The Multiplier Bootstrap Procedure}
\label{subsection_bootstrap_pro}

Since the limiting null distribution of $KS_n$ in Corollary \ref{coro.weakconv.stat}, i.e., $\sup_{t\in\mathbb R} |U_\infty(t,\bm\beta_0)|$, is non-pivotal, we compute its $p$-values via a multiplier bootstrap procedure for computational speed and simplicity. Specifically, the multiplier bootstrap statistic is
\begin{align}\label{eqn.bootstrap.Kn}
	KS_n^{\ast} = \max_{1\leq j \leq n} |U_n^{\ast}(\bm{X}_j^\top \hat{\bm\beta},\hat{\bm\beta})|,
\end{align}
with the multiplier bootstrap score-based empirical process $U_n^{\ast}(t,\hat{\bm\beta})$ defined in \eqref{eqn:multiplier.process}, which mimics the first-order uniform representation \eqref{eqn:first-order-rep} of the feasible process $U_n(t,\hat{\bm\beta})$.

The multiplier bootstrap procedure of our proposed score-based specification test is summarized as follows.

\begin{itemize}
	\item[{\bf Step 1:}] Calculate the test statistic $KS_n$ in \eqref{eqn:KSn-stat}.
	\item[{\bf Step 2:}] Generate independent multipliers $V_i\sim N(0,1)$ for $i=1,2,\cdots,n$, and calculate 
the bootstrap test statistic $KS_n^{\ast}$ in \eqref{eqn.bootstrap.Kn}.
	\item[{\bf Step 3:}] Repeat Step 2 $B$ times, drawing fresh multipliers each time, to obtain $B$ test statistics 
$KS_n^{\ast(1)},KS_n^{\ast(2)},\cdots,
\allowbreak KS_n^{\ast(B)}$.
	\item[{\bf Step 4:}] Calculate the $p$-value of rejecting the null $H_0$ as 
$p=(B+1)^{-1}\bigl(1+\sum_{b=1}^{B}I(KS_n^{\ast(b)}>KS_n)\bigr)$.
\end{itemize}


The following theorem ensures the validity of the multiplier bootstrap procedure. We say that a bootstrap process $Z_n^*$ converges weakly to a tight limit $Z$ \emph{in probability conditionally on the data}, written $Z_n^*\Rightarrow_{p} Z$, if
\[\sup_{h\in\mathrm{BL}_1(\ell^\infty(\mathbb R))}\bigl|\mathbb E^*[h(Z_n^*)]-\mathbb E[h(Z)]\bigr|\overset{p}{\to}0,\]
where $\mathbb E^*$ denotes the conditional expectation given the sample and $\mathrm{BL}_1$ the set of bounded Lipschitz functionals of norm at most one; equivalently, every subsequence has a further subsequence along which the conditional law converges weakly almost surely; see \citet[Section 2.9]{Vaart1996Weak} and \citet[Section 3.2]{Chernozhukov2013}. All conditional statements below are understood in this sense.

\begin{Theorem}\label{thm.bootstrap.stat}
	Suppose that Assumptions (A1)--(A7) in Appendix A hold. Under  correct specification with parameter $\bm\beta_0$, or under the sequence of local alternatives $H_{1n}$ in \eqref{eqn:local.alt}, we have
	\[U_n^{\ast}(t,\hat{\bm\beta}) \Rightarrow_{p} U_\infty(t,\bm\beta_0) \text{ in } \ell^\infty(\mathbb R),\]
where $U_\infty(t,\bm\beta_0)$ is the centered Gaussian process as defined in Theorem \ref{Thm.weakconvergence.stat}. Furthermore, under  correct specification, the $KS_n^{\ast}$ statistic has the same limiting null distribution $\sup_{t\in\mathbb R} |U_\infty(t,\bm\beta_0)|$ as the $KS_n$ statistic does in Corollary \ref{coro.weakconv.stat}, conditionally on the original data and in probability.
	
	Under the alternative hypothesis $H_1$ in \eqref{alter}, let $\bm\beta^*=\mathrm{plim}\,\hat{\bm\beta}\in\mathcal B\setminus\{\bm 0\}$ and assume that $H(\bm\beta^*)=\mathbb E\{[d-(1-d)\dot{\lambda}(-\bm X^\top\bm\beta^*)]\bm X\bm X^\top\}$ is nonsingular. Then, conditionally on the original data and in probability,
	\[U_n^{\ast}(t,\hat{\bm\beta}) \Rightarrow_p U_\infty^{\ast}(t,\bm\beta^*) \text{ in } \ell^\infty(\mathbb R),\]
	where $U_\infty^{\ast}(t,\bm\beta^*)$ is a centered Gaussian process with covariance function $\mathbb E[\Sigma^*(t_1;\bm\beta^*)\Sigma^*(t_2;\bm\beta^*)]$, with
	\[\Sigma^*(t;\bm\beta^*)=u(\bm W,\bm\beta^*)\left[I(\bm X^\top\bm\beta^*\le t)-G^\top(t,\bm\beta^*)H^{-1}(\bm\beta^*)\bm X\right]\]
	and $G(t,\bm\beta^*)=\mathbb E\{[d-(1-d)\dot{\lambda}(-\bm X^\top\bm\beta^*)]\bm X I(\bm X^\top\bm\beta^*\le t)\}$. Consequently, $KS_n^{\ast}\overset{\mathcal D}{\to}\sup_{t\in\mathbb R}|U_\infty^{\ast}(t,\bm\beta^*)|$ in probability conditional on the original data, and $KS_n^{\ast}=O_p(1)$. Since $KS_n\to\infty$ in probability under alternatives in $H_1$ with $\mu(\cdot,\bm\beta^*)\not\equiv0$ (Theorem \ref{thm.consistency.stat}), the bootstrap $p$-value converges to zero in probability, so the bootstrap-based test is consistent against every such alternative.
\end{Theorem}

It is important to remark that in the bootstrap process \eqref{eqn:multiplier.process} the multipliers multiply the whole summand, so that the correction term is reproduced within each bootstrap draw through the term $-G_n^\top(t,\hat{\bm\beta})H_n^{-1}(\hat{\bm\beta})\,n^{-1/2}\sum_i V_i u(\bm W_i,\hat{\bm\beta})\bm X_i$, which is of order $O_p(1)$ conditional on the data and matches the PEU contribution in the bootstrap analogue of \eqref{eqn:first-order-rep}. Omitting this term would change the conditional covariance kernel of the bootstrap process, so the correction is essential for valid inference; with it, the proposed bootstrap ``captures the PEU'' and reproduces the full first-order representation of the feasible process. Note that, for the exact maximum likelihood estimator, the first-order condition $\sum_{i=1}^n u(\bm W_i,\hat{\bm\beta})\bm X_i=0$ holds identically, so the statistic $KS_n$ itself requires no adjustment; the entire effect of the parameter estimation is carried by the bootstrap.

\subsection{The Case of Unknown Scale $\sigma$}
\label{subsec:two-parameter}

The asymptotic theory so far normalizes the scale to $\sigma=1$, which is appropriate when the scale is known. In the empirical applications of Section \ref{s:EmpApp}, the scale is unknown, so it is important to spell out what the $\sigma=1$ normalization tests and to specify the version of the score-based specification test used there. Throughout this section we assume that $\Theta$ is compact with $\sigma\ge\sigma_{\min}>0$, that $\mathbb E|z|^4<\infty$, and that $H(\bm\theta_0)$ is positive definite; the Donsker classes of Assumption (A6) are replaced by their two-parameter analogs built from $m(\bm W,\bm\theta)$ and $\bm s(\bm W,\bm\theta)$.

To proceed, collect the parameters in $\bm\theta=(\bm\beta^\top,\sigma)^\top\in\Theta\subset\mathbb R^{p+1}$, and define, for $z=(Y-\bm X^\top\bm\beta)/\sigma$ and $z'=-\bm X^\top\bm\beta/\sigma$, the scalar mark and the joint score
\begin{align}\label{eqn:two-param-mark}
m(\bm W,\bm\theta) &= d z - (1-d)\lambda(z'),\\[2pt]
\bm s(\bm W,\bm\theta) &= \begin{pmatrix} u_\beta(\bm W,\bm\theta)\\ u_\sigma(\bm W,\bm\theta)\end{pmatrix},\qquad
u_\beta(\bm W,\bm\theta)=\frac{\bm X}{\sigma}\,m(\bm W,\bm\theta),\\[2pt]
u_\sigma(\bm W,\bm\theta)&=\frac{1}{\sigma}\left[d(z^2-1) - (1-d)z'\lambda(z')\right],
\end{align}
so that $\bm s$ is the score vector of the Tobit log-likelihood with respect to $(\bm\beta,\sigma)$. Under correct specification with parameter $\bm\theta_0=(\bm\beta_0^\top,\sigma_0)^\top$, $\mathbb E[\bm s(\bm W,\bm\theta_0)|\bm X]=0$ a.s., and in particular $\mathbb E[m(\bm W,\bm\theta_0)|\bm X]=0$ a.s.; the latter reduces to \eqref{eqn:SpecTestP} when $\sigma_0=1$. The two-parameter analog of \eqref{eqn:d-mean} is $\mathbb E[d|\bm X]=\Phi(\bm X^\top\bm\beta_0/\sigma_0)$. The information blocks in \eqref{eqn:two-param-H} below are displayed in the random-weight (observed-information) form, in terms of the indicators $d$ and the standardized residuals $z,z'$. Replacing $d$ by its conditional mean $\Phi(\bm X^\top\bm\beta_0/\sigma_0)$ yields the corresponding conditional-expectation (expected-information) form of each weight, exactly as in the one-parameter case. The two population versions are equal  under correct specification (where $\mathbb E[d|\bm X]=\Phi(\bm X^\top\bm\beta_0/\sigma_0)$), while only the random-weight versions remain consistent under the alternative. 

In the following, the maintained null is the two-parameter score restriction 
$$H_0^{2p}:\mathbb E[\bm s(\bm W,\bm\theta_0)|\bm X]=\bm 0 \text{ a.s. for some } \bm\theta_0\in\Theta,$$ 
with the alternative $H_1^{2p}$ defined as its negation as in \eqref{alter}. By the coarsening argument, we have in particular $\mathbb E[m(\bm W,\bm\theta_0)|\bm X^\top\bm\beta_0]=0$ a.s. As in Remark 1, the index version is a necessary but not sufficient condition for correct specification, and the test detects deviations that are visible through the index. Note that testing $H_0^{2p}$ does \emph{not} impose $\sigma=1$: a correctly specified model with any scale is accepted at its true $(\bm\beta_0,\sigma_0)$. In contrast, the $\sigma=1$ score of Section 2.1 tests the joint hypothesis of a correct functional form and unit scale: if the data are generated by the correctly specified Tobit model $Y^*=\bm X^\top\bm\beta_0+\sigma e$ with $e\sim N(0,1)$ and $\sigma\neq1$, then
\[\mathbb E[u(\bm W,\bm\beta_0)|\bm X]=\sigma\,\phi\Bigl(\frac{\bm X^\top\bm\beta_0}{\sigma}\Bigr)-\Phi\Bigl(-\frac{\bm X^\top\bm\beta_0}{\sigma}\Bigr)\frac{\phi(\bm X^\top\bm\beta_0)}{\Phi(-\bm X^\top\bm\beta_0)}\neq0\quad\text{in general},\]
which reduces to $(\sigma-1)\phi(\bm X^\top\bm\beta_0)$ when $\bm X^\top\bm\beta_0=0$; see Lemma A.\ref{lem:scale-calc}. The $\sigma=1$ test would therefore reject a perfectly specified model solely because of the scale. The two-parameter version below separates scale from functional-form misspecification and is the version used in the empirical applications.

Now the information matrix $H(\bm\theta)=\mathbb E[-\partial_{\bm\theta}\bm s(\bm W,\bm\theta)]$ has blocks
\begin{align}\label{eqn:two-param-H}
H_{11} &= \frac{1}{\sigma^2}\left[d-(1-d)\dot\lambda(z')\right],\\[2pt]
H_{12} &= \frac{1}{\sigma^2}\left[2dz - (1-d)\left(z'\dot\lambda(z')+\lambda(z')\right)\right],\\[2pt]
H_{22} &= \frac{1}{\sigma^2}\left[d(3z^2-1) - (1-d)\left(z'^2\dot\lambda(z')+2z'\lambda(z')\right)\right],
\end{align}
so that 
$$H(\bm\theta)=\mathbb E\left[\begin{matrix} H_{11}\bm X\bm X^\top & H_{12}\bm X\\ H_{12}\bm X^\top & H_{22}\end{matrix}\right].$$ 
In addition, the derivative of the mark satisfies
\begin{align*}
\partial_{\bm\beta}m(\bm W,\bm\theta) &= -\frac{\bm X}{\sigma}\left[d-(1-d)\dot\lambda(z')\right]
\end{align*}
and
\begin{align*}
\partial_\sigma m(\bm W,\bm\theta) = \frac{1}{\sigma}\left[-dz + (1-d)z'\dot\lambda(z')\right].
\end{align*}
Define the $(p+1)$-vector $\mathbb G(t,\bm\theta)=\mathbb E[-\partial_{\bm\theta}m(\bm W,\bm\theta)I(\bm X^\top\bm\beta\le t)]$ and its sample analog
\begin{align}\label{eqn:Gn-twopar}
\mathbb G_n(t,\hat{\bm\theta})=\frac1n\sum_{i=1}^n[-\partial_{\bm\theta}m(\bm W_i,\hat{\bm\theta})]I(\bm X_i^\top\hat{\bm\beta}\le t),\qquad
H_n(\hat{\bm\theta})=\frac1n\sum_{i=1}^n[-\partial_{\bm\theta}\bm s(\bm W_i,\hat{\bm\theta})],
\end{align}
where the entries of $H_n(\hat{\bm\theta})$ are the sample versions of \eqref{eqn:two-param-H}. The feasible process is
\begin{align}\label{eqn:two-param-process}
U_n(t,\hat{\bm\theta}) = \frac{1}{\sqrt n}\sum_{i=1}^n m(\bm W_i,\hat{\bm\theta})I(\bm X_i^\top\hat{\bm\beta}\le t),
\end{align}
with the KS-type test statistic 
$$KS_n=\max_{1\le j\le n}|U_n(\bm X_j^\top\hat{\bm\beta},\hat{\bm\theta})|.$$ 
In the case of unknown scale $\sigma$, the multiplier bootstrap process mimicking its first-order uniform representation is shown to have the following form:
\begin{align}\label{eqn:two-param-bootstrap}
U_n^{\ast}(t,\hat{\bm\theta}) = \frac{1}{\sqrt n}\sum_{i=1}^n V_i\left[m(\bm W_i,\hat{\bm\theta})I(\bm X_i^\top\hat{\bm\beta}\le t) - \mathbb G_n^\top(t,\hat{\bm\theta})H_n^{-1}(\hat{\bm\theta})\,\bm s(\bm W_i,\hat{\bm\theta})\right].
\end{align}
Note that when $\sigma$ is known and normalized to 1, $m=u$, $\bm s=u\bm X$, and $\mathbb G^\top H^{-1}\bm s_i = u_i\,G^\top H^{-1}\bm X_i$. The process and the bootstrap then reduce to their counterparts in Section 2.2. Under correct specification, the feasible process admits the first-order uniform representation
\begin{align}\label{eqn:twopar-representation}
U_n(t,\hat{\bm\theta}) = \frac{1}{\sqrt n}\sum_{i=1}^n\left[m(\bm W_i,\bm\theta_0)I(\bm X_i^\top\bm\beta_0\le t) - \mathbb G^\top(t,\bm\theta_0)H^{-1}(\bm\theta_0)\bm s(\bm W_i,\bm\theta_0)\right]+o_p(1),
\end{align}
uniformly in $t$, and the limiting covariance kernel of Theorem \ref{Thm.weakconvergence.twopar} simplifies, by the two-parameter identities of Lemma A.\ref{lem:twopar-id}, to
\begin{align}\label{eqn:twopar-kernel}
\mathbb E[\Sigma(t_1)\Sigma(t_2)]
&=\mathbb E\left[m^2(\bm W,\bm\theta_0)I(\bm X^\top\bm\beta_0\le t_1)I(\bm X^\top\bm\beta_0\le t_2)\right]\notag\\
&\quad-\mathbb G^\top(t_1,\bm\theta_0)H^{-1}(\bm\theta_0)\mathbb G(t_2,\bm\theta_0)\notag\\
&\le\mathbb E\left[m^2(\bm W,\bm\theta_0)I(\bm X^\top\bm\beta_0\le t_1)I(\bm X^\top\bm\beta_0\le t_2)\right].
\end{align}

The following theorems summarize the theoretical results for $U_n(t,\hat{\bm\theta}) $ and its bootstrap version $U_n^{\ast}(t,\hat{\bm\theta})$, which are designed for the case of unknown $\sigma$.

\begin{Theorem}\label{Thm.weakconvergence.twopar}
	Suppose that the model is correctly specified with parameter $\bm\theta_0$ and that Assumptions (A1)--(A6) in Appendix A hold, with the Donsker classes in (A6) replaced by their two-parameter analogs built from $m(\bm W,\bm\theta)$ and $\bm s(\bm W,\bm\theta)$, together with $\mathbb E|z|^4<\infty$. Then we have
	\[U_n(t,\hat{\bm\theta}) \Rightarrow U_\infty(t,\bm\theta_0)\text{ in } \ell^\infty(\mathbb R),\]
	where $U_\infty(t,\bm\theta_0)$ is a centered Gaussian process with covariance function $\mathbb E[\Sigma(t_1)\Sigma(t_2)]$, with
	\[\Sigma(t)=m(\bm W,\bm\theta_0)I(\bm X^\top\bm\beta_0\le t)-\mathbb G^\top(t,\bm\theta_0)H^{-1}(\bm\theta_0)\bm s(\bm W,\bm\theta_0).\]
	Consequently, $KS_n\overset{\mathcal D}{\to}\sup_{t\in\mathbb R}|U_\infty(t,\bm\theta_0)|$.
\end{Theorem}

\begin{Theorem}\label{thm.twopar.consistency}
	Suppose that $\hat{\bm\theta}$ converges in probability to a unique pseudo-true value $\bm\theta^*=(\bm\beta^{*\top},\sigma^*)^\top\in\Theta$ with $\bm\beta^*\neq\bm 0$ under $H_1^{2p}$, satisfying $\mathbb E[\bm s(\bm W,\bm\theta^*)]=0$. Under Assumptions (A1), (A2), (A3), and (A6) in Appendix A (with the two-parameter analogs as above),
	\[\sup_{t\in\mathbb R}\left|\frac{1}{\sqrt n}U_n(t,\hat{\bm\theta}) - \mu(t,\bm\theta^*)\right| = o_p(1),\]
	where $\mu(t,\bm\theta^*)=\mathbb E[m(\bm W,\bm\theta^*)I(\bm X^\top\bm\beta^*\le t)]$ is the mark moment. 
Consequently, if the index restriction $\mathbb E[m(\bm W,\bm\theta^*)|\bm X^\top\bm\beta^*]=0$ fails a.s.\ (equivalently, $\mu(\cdot,\bm\theta^*)\not\equiv0$), then $\sup_{t\in\mathbb R}|\mu(t,\bm\theta^*)|>0$ and
\[ \frac{KS_n}{\sqrt n}\overset{p}{\to}\sup_{t\in\mathbb R}|\mu(t,\bm\theta^*)|>0,\]
so that $\lim_{n\to\infty}\mathbb P(KS_n>c_\tau)=1$ for any fixed critical value $c_\tau$; the test is consistent against every such alternative.
\end{Theorem}

\begin{Theorem}\label{thm.twopar.local}
	Consider the sequence of local alternatives (contiguous perturbations of the correctly specified model with parameter $\bm\theta_0$) $\mathbb E[m(\bm W,\bm\theta_0)|\bm X]=n^{-1/2}a(\bm X^\top\bm\beta_0)$ and $\mathbb E[\bm s(\bm W,\bm\theta_0)|\bm X]=n^{-1/2}\bm q(\bm X)$, where $\bm q(\bm X)=\bigl(a(\bm X^\top\bm\beta_0)\,\bm X^\top/\sigma_0,\,b(\bm X^\top\bm\beta_0)\bigr)^\top$ for a function $b$, the function $a$ satisfies the conditions in \eqref{eqn:local.alt}, and $\mathbb E\|\bm q(\bm X)\|<\infty$. Since $u_\beta=\bm X m/\sigma$, the first $p$ components of the condition on $\mathbb E[\bm s|\bm X]$ are automatic given the condition on $\mathbb E[m|\bm X]$; the substantive additional restriction is that the scale-score deviation $\mathbb E[u_\sigma(\bm W,\bm\theta_0)|\bm X]$ is an index function $b(\bm X^\top\bm\beta_0)$. The theorem extends 
to a general function $b(\bm X)$ with the same integrability and continuity conditions. Then, under the extended assumptions,
	\[U_n(t,\hat{\bm\theta}) \Rightarrow U_\infty(t,\bm\theta_0) + D(t,\bm\theta_0) \text{ in } \ell^\infty(\mathbb R),\]
	with the deterministic drift function
	\[D(t,\bm\theta_0) = \mathbb E[a(\bm X^\top\bm\beta_0)I(\bm X^\top\bm\beta_0\le t)] - \mathbb G^\top(t,\bm\theta_0)H^{-1}(\bm\theta_0)\,\mathbb E[\bm q(\bm X)],\]
	and $KS_n\overset{\mathcal D}{\to}\sup_{t\in\mathbb R}|U_\infty(t,\bm\theta_0)+D(t,\bm\theta_0)|$. This is the two-parameter analog of Theorem \ref{thm.localpower.stat}, with the correction now computed from the joint score.
\end{Theorem}

\begin{Theorem}\label{thm.twopar.bootstrap}
	Under the extended assumptions and under  correct specification with parameter $\bm\theta_0$ or the sequence of local alternatives of Theorem \ref{thm.twopar.local}, $U_n^{\ast}(t,\hat{\bm\theta})$ in \eqref{eqn:two-param-bootstrap} converges weakly in $\ell^\infty(\mathbb R)$ to $U_\infty(t,\bm\theta_0)$ conditional on the original data and in probability. Under  correct specification, $KS_n^{\ast}$ has the same limiting null distribution as $KS_n$, conditional on the data.
	
	Under $H_1^{2p}$, with $\bm\theta^*=(\bm\beta^{*\top},\sigma^*)^\top=\mathrm{plim}\,\hat{\bm\theta}$, $\bm\beta^*\neq\bm 0$, and $H(\bm\theta^*)$ nonsingular, the same argument yields $U_n^{\ast}(t,\hat{\bm\theta})\Rightarrow_p U_\infty^{\ast}(t,\bm\theta^*)$ in $\ell^\infty(\mathbb R)$ conditional on the original data and in probability, where $U_\infty^{\ast}(t,\bm\theta^*)$ is a centered Gaussian process with covariance kernel $\mathbb E[\Sigma^*(t_1;\bm\theta^*)\Sigma^*(t_2;\bm\theta^*)]$ and $\Sigma^*(t;\bm\theta^*)=m(\bm W,\bm\theta^*)I(\bm X^\top\bm\beta^*\le t)-\mathbb G^\top(t,\bm\theta^*)H^{-1}(\bm\theta^*)\bm s(\bm W,\bm\theta^*)$; consequently, $KS_n^{\ast}=O_p(1)$ while $KS_n\to\infty$ in probability, so the test is consistent against every fixed alternative in $H_1^{2p}$ for which the index projection of the mark moment at the pseudo-true value is nonzero.
\end{Theorem}

The two-parameter results parallel their known-$\sigma$ counterparts term by term: with the joint score $\bm s$, the $(p+1)$-vector $\mathbb G(t,\bm\theta)$ and the $(p+1)\times(p+1)$ information matrix $H(\bm\theta)$ replacing $u\bm X$, $G(t,\bm\beta)$ and $H(\bm\beta)$, the first-order uniform representation, the PEU correction $-\mathbb G^\top(t,\bm\theta)H^{-1}(\bm\theta)\bm s(\bm W,\bm\theta)$ and the multiplier bootstrap take the same form. In particular, when the scale is known and $\sigma=1$, deleting the scale component gives $\bm s=u\bm X$, $\mathbb G=G$, $H=H(\bm\beta)$ and $m=u$, so Theorems \ref{Thm.weakconvergence.twopar}--\ref{thm.twopar.bootstrap} reduce exactly to Theorems \ref{Thm.weakconvergence.stat}--\ref{thm.bootstrap.stat}. When $\sigma$ is unknown, however, the two limits are not related by simply setting $\sigma=1$: both the mark and the score change, so the known-$\sigma$ critical values are not generally valid for the unknown-$\sigma$ statistic (nor conversely), and the two-parameter version must be used whenever the scale is estimated. The analog of the variance identity is
\begin{align*}
\mathbb E[\Sigma(t)^2]&=\mathbb E[m^2(\bm W,\bm\theta_0)I(\bm X^\top\bm\beta_0\le t)]-\mathbb G^\top(t,\bm\theta_0)H^{-1}(\bm\theta_0)\mathbb G(t,\bm\theta_0)\\
&\le\mathbb E[m^2(\bm W,\bm\theta_0)I(\bm X^\top\bm\beta_0\le t)],
\end{align*}
so the correction can only shrink the limiting variance relative to the process without the correction term. The reduction is now driven by both the regression and the scale score, and the bootstrap correction remains essential. 

Under 
$H_1^{2p}$, Theorem \ref{thm.twopar.consistency} gives the mark moment $\mu(t,\bm\theta^*)=\mathbb E[m(\bm W,\bm\theta^*)I(\bm X^\top\bm\beta^*\le t)]$, so the two-parameter test is consistent against every fixed alternative whose index projection at the pseudo-true value is nonzero. Theorem \ref{thm.twopar.bootstrap} shows that the bootstrap statistic is $O_p(1)$ while $KS_n\to\infty$ in probability, so the bootstrap-based test is also consistent. Under the local alternatives of Theorem \ref{thm.twopar.local}, the drift is projected with the joint score,
\[D(t,\bm\theta_0)=\mathbb E[a(\bm X^\top\bm\beta_0)I(\bm X^\top\bm\beta_0\le t)]-\mathbb G^\top(t,\bm\theta_0)H^{-1}(\bm\theta_0)\mathbb E[\bm q(\bm X)],\]
and the undetectable local alternatives are those whose deviation is absorbed by the joint score. Compared with the known-$\sigma$ case, the correction now also contains the scale-score direction. Besides the local perturbations absorbed by the regression score, those absorbed by a $\sqrt n$-shift of $\sigma$ are undetectable at the parametric rate. The bootstrap remains valid under the local alternatives, so the limiting local power is nontrivial whenever $D\not\equiv0$.

\section{Numerical Studies}\label{sec:simulations}
We evaluate the finite-sample performance of the proposed test in three complementary designs (Settings~I and II below, the index study of Section \ref{sec:indexinvisible}, and an unknown-scale design). For Settings~I and II, and for the index study, we also compare the proposed test with two existing specification tests under a common protocol: the same sample sizes, the same 1000 replications, and $B=1000$ bootstrap replications for every method.
The significance levels are set to 0.01, 0.05, and 0.1 for all simulations.
For Settings~I and II, we use sample sizes $n=100$ and $n=200$ and repeat the tests 1000 times, with $B=1000$ bootstrap replications;  The unknown-scale design uses the sample sizes and replication counts given in its subsection. The proposed test uses the multiplier bootstrap of Section \ref{subsection_bootstrap_pro}. Section \ref{sec:size-power} reports size and power together with the comparison with existing tests, and Section \ref{sec:indexinvisible} the index-visible/index-invisible study.

\subsection{Size, Power, and Comparison with Existing Tests}\label{sec:size-power}\label{sec:benchmarks}

We compare the proposed test with two existing procedures that test correct specification of the Tobit conditional mean. For $r(\bm W,\bm\theta)=Y-g(\bm X,\bm\theta)$ with $g(\bm x,\bm\theta)=\Phi(\eta)(\bm x^\top\bm\beta+\sigma\lambda(\eta))$ and $\eta=\bm x^\top\bm\beta/\sigma$, the conditional-mean restriction $\mathbb E[r(\bm W,\bm\theta_0)|\bm X]=0$ a.s.\ tested by the competitors and the mark restriction $\mathbb E[m(\bm W,\bm\theta_0)|\bm X]=0$ a.s.\ tested here are both implied by correct specification of the Tobit model, and the two restrictions coincide when the censoring probability equals $\Phi(\bm X^\top\bm\beta_0/\sigma_0)$; the three tests therefore assess the implications of correct specification of the same Tobit model and differ in the statistic, the conditioning variable, and the calibration.
The first competitor is the cumulative-residual Cram\'er--von Mises test of \cite{Sun2021JKSS}. They base their test on the estimated empirical process $W_n(\bm x)=n^{-1/2}\sum_i\hat\xi_iI(\bm X_i\le\bm x)$, where $\hat\xi_i=Y_i-g(\bm X_i,\hat{\bm\theta}_n)$, and on the Cram\'er--von Mises statistic $T_n=\int W_n(\bm x)^2\text{d}F_n(\bm x)=n^{-2}\sum_i[\sum_j\hat\xi_jI(\bm X_j\le\bm X_i)]^2$, with critical values obtained from the wild bootstrap $Y_i^b=\max\{g(\bm X_i,\hat{\bm\theta}_n)+(Y_i-g(\bm X_i,\hat{\bm\theta}_n))V_i,\,0\}$; see their equations (2.2)--(2.3) and Section~3.2. Since our design and test are based on the normal Tobit likelihood, we implement the test statistic using the maximum likelihood estimator and the normal conditional mean, and update the parameter on each bootstrap sample with one Newton step; this is the efficient version of their test in our setting.
The second competitor is the kernel-smoothing lack-of-fit test of \cite{Koul2014JMA}, a weighted adaptation of \cite{Zheng1996} to Tobit regression. With residuals $r_i=Y_i-g(\bm X_i,\hat{\bm\theta})$, a product kernel $K$ and bandwidth $h$, the statistic is $I_n=[n(n-1)h^p]^{-1}\sum_{i\neq j}K((\bm X_i-\bm X_j)/h)r_ir_j$, and the standardized version $J_n=nh^{p/2}I_n/\hat\sigma_n$ is asymptotically standard normal under $H_0$, where $\hat\sigma_n^2=2[n(n-1)h^p]^{-1}\sum_{i\neq j}K((\bm X_i-\bm X_j)/h)^2r_i^2r_j^2$ is a consistent estimator of the asymptotic variance. We use the product standard-normal kernel and the bandwidth $h=n^{-1/7}$, as chosen by \cite{Koul2014JMA} and \cite{Sun2021JKSS}. Because the normal approximation for this degenerate $U$-statistic is inaccurate in samples of this size, we calibrate the test by a wild bootstrap ($r_i^b=r_iV_i$); the parameter-estimation effect is asymptotically negligible for this second-order statistic. Unlike the kernel test, the proposed test requires no bandwidth selection, and unlike full-covariate methods such as the cumulative-residual test, conditioning on the index avoids the dimensionality penalty of high-dimensional indicators.

{\bf Setting I:} The data generating process is 
\begin{align*}
	Y^{*} &= \beta_0 + \beta_1 X_{1} + \beta_2 X_{2} + \gamma h(\bm{X}) + e,\\[2pt]
	Y &= Y^* I(Y^*>0).
\end{align*}

We set  $e\sim N(0,1)$ and $\gamma=0,0.1,0.2,0.3,0.4,0.5,0.6$. The true regression parameters and covariates are chosen to be as follows:
\begin{itemize}
	\item [$(i).$] $(\beta_0,\beta_1,\beta_2) = (-0.5,-1.5,2),X_1, X_2 \sim U(0,2\pi), h(\bm{X})=|\sin(X_{1}) + 0.5 X_{2}^2|.$
	\item [$(ii).$] $(\beta_0,\beta_1,\beta_2) = (1,1,1),X_1, X_2 \sim N(0,1), h(\bm{X})=X_{1}^2 + X_{2}^2.$
\end{itemize}

\begin{table}[htbp]
\centering
\caption{Empirical rejection rates (in \%) of the proposed test under Setting I, scenario $(i)$; columns report the significance levels $\tau=1\%,5\%,10\%$. Based on $M=1000$ replications with $B=1000$ bootstrap replications; Monte Carlo standard errors are at most $\sqrt{0.5\cdot0.5/1000}\approx1.6$ percentage points.}
\label{tab:set1a}
\begin{tabular}{c|ccc|ccc}
\hline
$\gamma$ & \multicolumn{3}{c|}{$n=100$} & \multicolumn{3}{c}{$n=200$}\\
 & 1\% & 5\% & 10\% & 1\% & 5\% & 10\%\\
\hline
0.0 & 0.2 & 3.9 & 9.6 & 0.7 & 4.7 & 8.8\\
0.1 & 0.6 & 5.7 & 10.2 & 0.9 & 5.8 & 12.8\\
0.2 & 0.5 & 7.1 & 14.7 & 3.0 & 12.8 & 22.2\\
0.3 & 2.2 & 13.2 & 24.0 & 8.6 & 27.5 & 42.2\\
0.4 & 5.0 & 19.7 & 35.3 & 21.2 & 49.7 & 65.4\\
0.5 & 7.8 & 30.5 & 49.8 & 44.1 & 73.6 & 85.5\\
0.6 & 14.7 & 44.1 & 64.4 & 58.5 & 90.1 & 95.9\\
\hline
\end{tabular}
\end{table}

\begin{table}[htbp]
\centering
\caption{As in Table \ref{tab:set1a} for Setting I, scenario $(ii)$.}
\label{tab:set1b}
\begin{tabular}{c|ccc|ccc}
\hline
$\gamma$ & \multicolumn{3}{c|}{$n=100$} & \multicolumn{3}{c}{$n=200$}\\
 & 1\% & 5\% & 10\% & 1\% & 5\% & 10\%\\
\hline
0.0 & 1.0 & 4.4 & 10.4 & 0.6 & 5.6 & 9.4\\
0.1 & 1.0 & 8.6 & 16.7 & 3.0 & 12.7 & 21.9\\
0.2 & 3.4 & 15.2 & 28.5 & 12.3 & 34.6 & 48.2\\
0.3 & 9.0 & 30.6 & 47.9 & 34.0 & 69.6 & 83.8\\
0.4 & 17.3 & 46.8 & 63.8 & 58.7 & 87.0 & 95.3\\
0.5 & 29.0 & 62.4 & 79.0 & 74.9 & 94.8 & 99.1\\
0.6 & 39.3 & 73.7 & 87.1 & 84.3 & 97.2 & 99.6\\
\hline
\end{tabular}
\end{table}

Tables \ref{tab:set1a} and \ref{tab:set1b} report the empirical rejection rates of the proposed test under the two scenarios of Setting I.

The proposed test maintains empirical sizes close to the nominal levels in Setting I: across the two scenarios and both sample sizes its rejection rates under the null ($\gamma=0$) lie in the ranges $[0.2\%,1.0\%]$, $[3.9\%,5.6\%]$ and $[8.8\%,10.4\%]$ at the 1\%, 5\% and 10\% levels, respectively. As $\gamma$ increases, the power of the proposed test increases overall; the minor non-monotonicities at small $\gamma$ (e.g., at the $1\%$ level with $n=100$ in scenario $(i)$) are within Monte Carlo error. Under Scenario $(ii)$, which conforms to the standard Tobit model assumptions, the power is higher than under Scenario $(i)$; e.g., at the 5\% level and $n=200$ the rejection rate reaches 97.2\% at $\gamma=0.6$ in Scenario $(ii)$, versus 90.1\% in Scenario $(i)$. Increasing the sample size from 100 to 200 yields substantially higher power, while the proposed test maintains an accurate size throughout.

{\bf Setting I comparison.} Table \ref{tab:benchS1} reports the rejection rates of the three tests under Setting~I at both sample sizes, and Figure \ref{fig:benchS1} displays the corresponding 5\% rejection-rate curves. The proposed test is the only procedure whose empirical size is close to the nominal level in both scenarios and at both sample sizes (at the 5\% level, 3.9\% and 4.7\% in scenario $(i)$ and 4.4\% and 5.6\% in scenario $(ii)$ for $n=100$ and $n=200$), whereas the cumulative-residual test of \cite{Sun2021JKSS} is conservative at $n=200$ (2.6\% in scenario $(i)$ and 1.0\% in scenario $(ii)$) and the kernel test of \cite{Koul2014JMA} severely under-rejects at the 5\% level in every design (0.0\%--0.6\%), so that its nominal level is not meaningful.

\begin{table}[htbp]
\centering
\small
\setlength{\tabcolsep}{4pt}
\caption{Empirical rejection rates (\%) of the three tests under Setting~I, scenarios $(i)$ and $(ii)$, $n=100$ and $n=200$ ($M=1000$ replications and $B=1000$ bootstrap replications for every method). For $\gamma=0$ the 1\%, 5\% and 10\% rejection rates are reported (size); for $\gamma>0$ the rejection rate at the 5\% level is reported (power). The proposed test uses the multiplier bootstrap; both benchmarks are calibrated by a wild bootstrap.}
\label{tab:benchS1}
\begin{tabular}{l|ccc|ccc|ccc|ccc}
\hline
 & \multicolumn{6}{c|}{$n=100$} & \multicolumn{6}{c}{$n=200$}\\
 & \multicolumn{3}{c|}{size ($\gamma=0$)} & \multicolumn{3}{c|}{power (5\%)} & \multicolumn{3}{c|}{size ($\gamma=0$)} & \multicolumn{3}{c}{power (5\%)}\\
Method & 1\% & 5\% & 10\% & 0.2 & 0.4 & 0.6 & 1\% & 5\% & 10\% & 0.2 & 0.4 & 0.6\\
\hline
\multicolumn{13}{c}{Scenario $(i)$}\\
Proposed & 0.2 & 3.9 & 9.6 & 7.1 & 19.7 & 44.1 & 0.7 & 4.7 & 8.8 & 12.8 & 49.7 & 90.1\\
Sun et al.\ (2021a) & 0.6 & 4.3 & 11.3 & 4.7 & 9.1 & 12.0 & 0.2 & 2.6 & 6.0 & 6.2 & 18.6 & 36.2\\
Koul et al.\ (2014) & 0.0 & 0.0 & 0.5 & 1.9 & 22.0 & 67.7 & 0.0 & 0.1 & 1.4 & 8.0 & 72.1 & 99.8\\
\hline
\multicolumn{13}{c}{Scenario $(ii)$}\\
Proposed & 1.0 & 4.4 & 10.4 & 15.2 & 46.8 & 73.7 & 0.6 & 5.6 & 9.4 & 34.6 & 87.0 & 97.2\\
Sun et al.\ (2021a) & 1.0 & 5.5 & 12.2 & 24.5 & 51.8 & 66.0 & 0.2 & 1.0 & 2.5 & 25.9 & 81.4 & 91.4\\
Koul et al.\ (2014) & 0.0 & 0.6 & 1.9 & 26.7 & 94.5 & 100.0 & 0.1 & 0.6 & 2.2 & 59.7 & 99.9 & 100.0\\
\hline
\end{tabular}
\end{table}

\begin{figure}[htbp]
	\centering
	\begin{subfigure}[b]{0.45\textwidth}
		\centering
		\includegraphics[width=\textwidth]{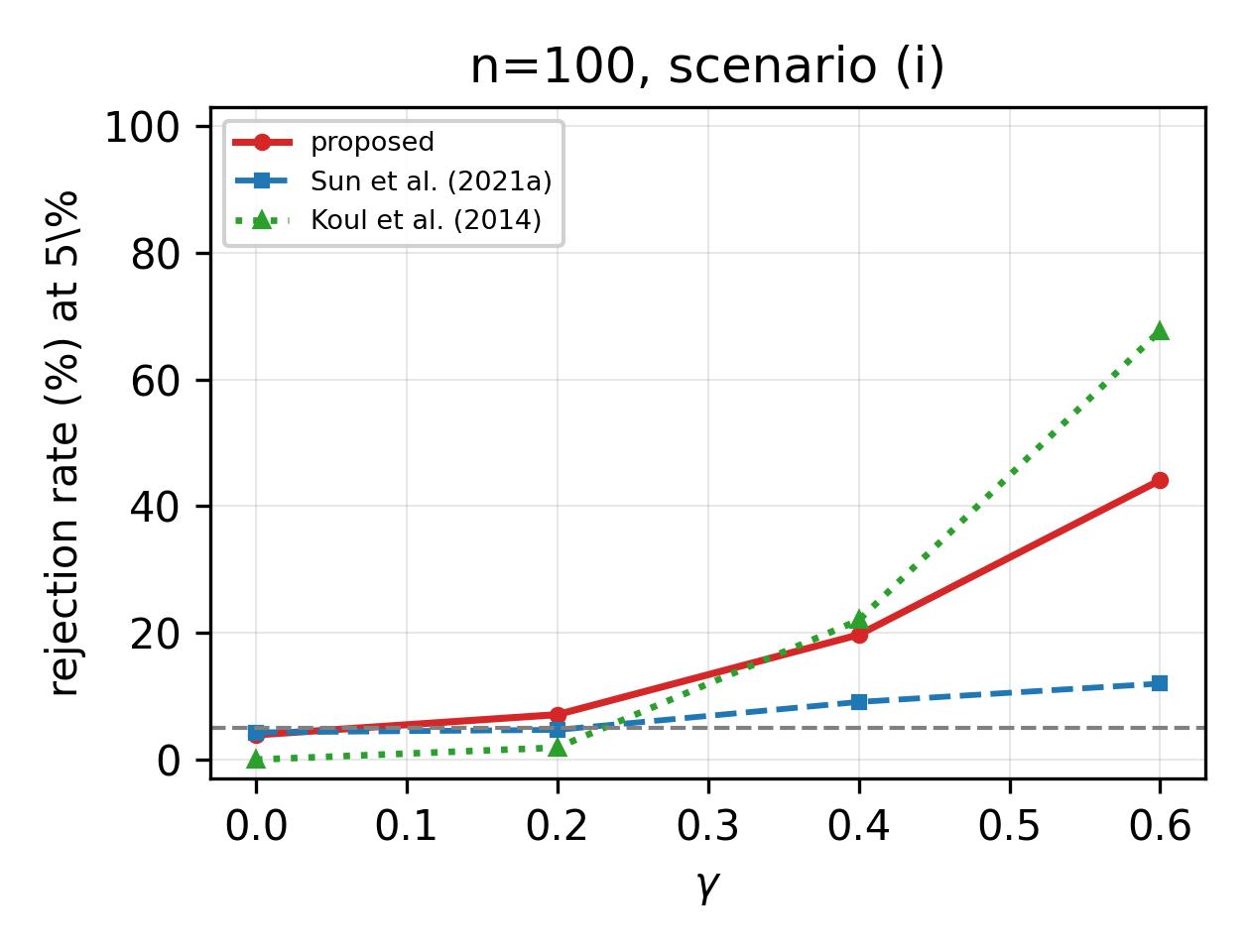}
		\caption{$n$=100, scenario $(i)$}
	\end{subfigure}
	\hfill
	\begin{subfigure}[b]{0.45\textwidth}
		\centering
		\includegraphics[width=\textwidth]{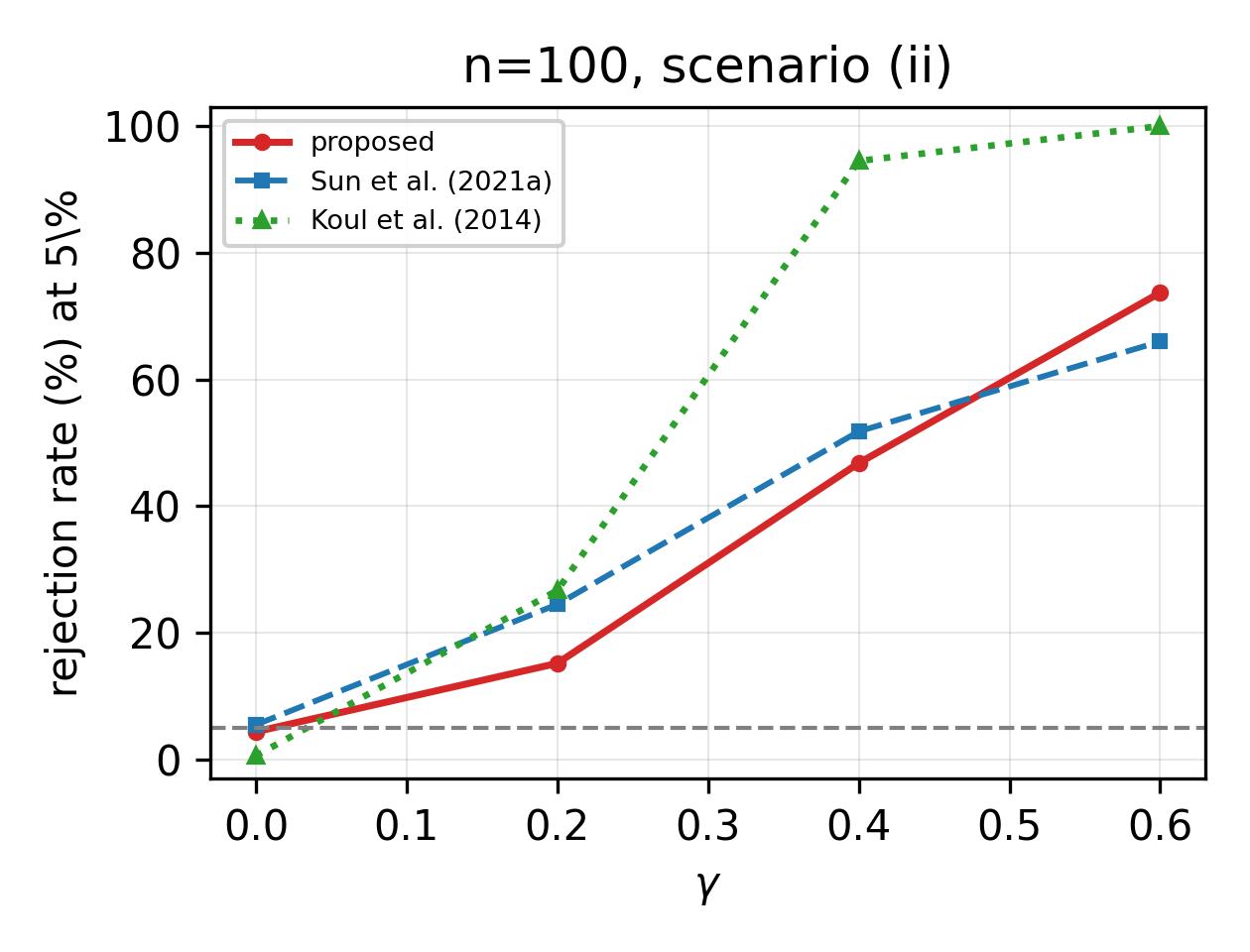}
		\caption{$n$=100, scenario $(ii)$}
	\end{subfigure}
	\vspace{1em}
	\begin{subfigure}[b]{0.45\textwidth}
		\centering
		\includegraphics[width=\textwidth]{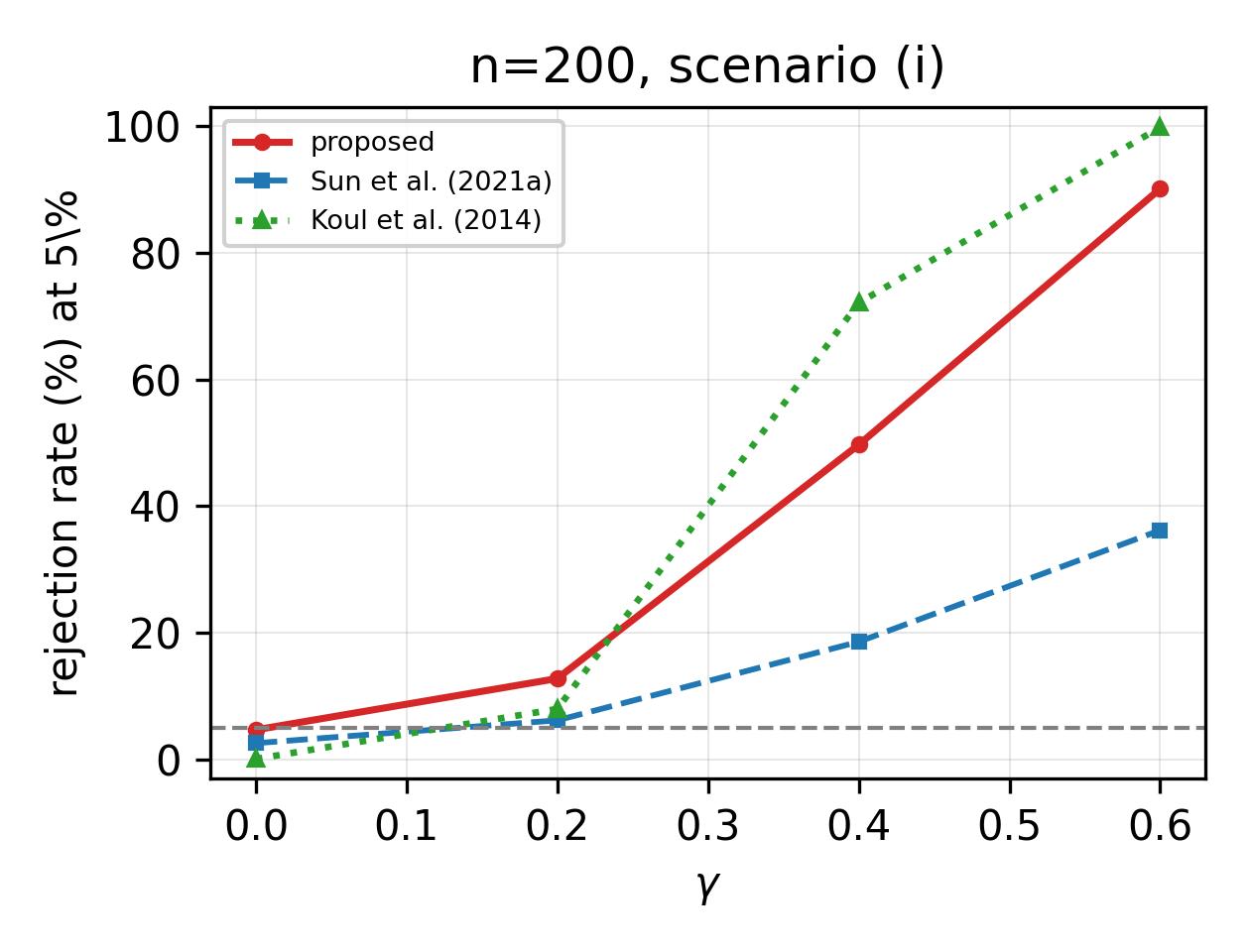}
		\caption{$n$=200, scenario $(i)$}
	\end{subfigure}
	\hfill
	\begin{subfigure}[b]{0.45\textwidth}
		\centering
		\includegraphics[width=\textwidth]{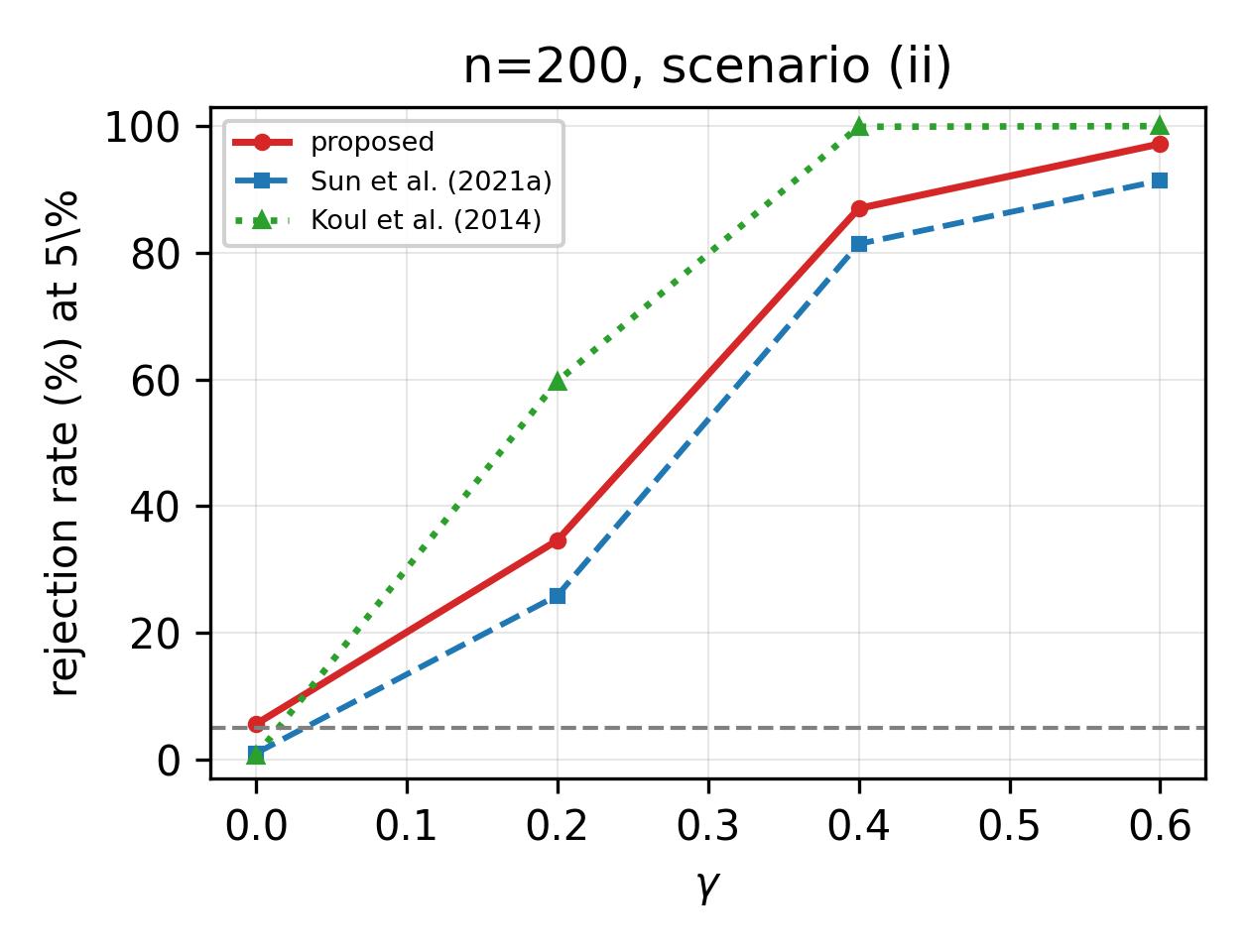}
		\caption{$n$=200, scenario $(ii)$}
	\end{subfigure}
	\caption{Empirical rejection rates of the three tests at the 5\% level under Setting~I, scenarios $(i)$ and $(ii)$ and $n=100$, $200$ ($M=1000$, $B=1000$); the gray dashed line marks the nominal 5\% level.}
	\label{fig:benchS1}
\end{figure}

Among the procedures whose level is reliable, the proposed test delivers competitive power: at $n=200$ it dominates the cumulative-residual test throughout, for example 90.1\% versus 36.2\% in scenario $(i)$ and 97.2\% versus 91.4\% in scenario $(ii)$ at $\gamma=0.6$; the cumulative-residual test is somewhat more powerful at $n=100$ for small deviations in scenario $(ii)$ (24.5\% versus 15.2\% at $\gamma=0.2$), where its size is close to nominal, but it is overtaken by the proposed test at $n=200$. The kernel test rejects more often under the alternatives, but its rejection rates are those of a test whose actual size under the null is essentially zero (0.1\% at the 5\% level, $n=200$, scenario $(i)$), and it requires a bandwidth. This comparison illustrates the trade-off discussed in Remark~1: conditioning on the index buys robustness and dimension reduction at the cost of some power against index-orthogonal deviations, whereas the full-covariate competitors cannot be used reliably once the dimension grows, as the next comparison shows.

{\bf Setting II:} To evaluate the performance of the proposed test with moderately high-dimensional covariates, we consider the following Tobit model with $p$-dimensional covariates.
\begin{align*}
	Y^{*} &= \beta_0 + \beta_1 X_{1} + \beta_2 X_{2} +\cdots + \beta_p X_{p} + \gamma g(\bm{X}) + e,\\[2pt]
	Y &= Y^* I(Y^*>0).
\end{align*}
where $g(\bm{X})=0.15\sum_{i=1}^p X_i^2$, $X_i (i = 1, 2, ..., p)$ all follow the standard normal distribution and are mutually independent, $e\sim N(0,1)$, and $\gamma=0,0.5,1.0,1.5,2.0$. All true parameters are assigned a value of 1.

\begin{table}[htbp]
\centering
\caption{Empirical rejection rates (in \%) of the proposed test under Setting II with $n=100$ and covariate dimensions $p=8$ and $p=12$; see Table \ref{tab:set1a} for notation.}
\label{tab:set2}
\begin{tabular}{c|ccc|ccc}
\hline
$\gamma$ & \multicolumn{3}{c|}{$n=100$, $p=8$} & \multicolumn{3}{c}{$n=100$, $p=12$}\\
 & 1\% & 5\% & 10\% & 1\% & 5\% & 10\%\\
\hline
0.0 & 1.6 & 6.9 & 12.8 & 1.5 & 7.9 & 17.6\\
0.5 & 1.0 & 8.9 & 16.0 & 1.1 & 7.5 & 15.5\\
1.0 & 2.0 & 10.6 & 19.7 & 2.9 & 12.3 & 23.3\\
1.5 & 4.5 & 17.5 & 28.3 & 4.3 & 18.3 & 30.5\\
2.0 & 6.3 & 24.8 & 36.8 & 6.7 & 21.7 & 35.9\\
\hline
\end{tabular}
\end{table}

\begin{table}[htbp]
\centering
\caption{As in Table \ref{tab:set2} for $n=200$, $p=8,12$.}
\label{tab:set2b}
\begin{tabular}{c|ccc|ccc}
\hline
$\gamma$ & \multicolumn{3}{c|}{$n=200$, $p=8$} & \multicolumn{3}{c}{$n=200$, $p=12$}\\
 & 1\% & 5\% & 10\% & 1\% & 5\% & 10\%\\
\hline
0.0 & 1.1 & 6.7 & 12.0 & 1.4 & 7.3 & 12.9\\
0.5 & 1.9 & 7.5 & 15.1 & 2.3 & 8.9 & 15.7\\
1.0 & 4.9 & 19.7 & 31.8 & 6.7 & 18.6 & 30.6\\
1.5 & 13.9 & 36.5 & 51.9 & 13.1 & 35.2 & 49.7\\
2.0 & 21.6 & 52.9 & 66.9 & 19.5 & 47.4 & 62.8\\
\hline
\end{tabular}
\end{table}

The simulation results are summarized in Tables \ref{tab:set2} and \ref{tab:set2b}.

Tables \ref{tab:set2} and \ref{tab:set2b} report the results for Setting II. Under the null ($\gamma=0$) the proposed test's rejection rates are $[1.1\%,1.6\%]$, $[6.7\%,7.9\%]$ and $[12.0\%,17.6\%]$ at the 1\%, 5\% and 10\% levels, respectively, across $p=8$ and $p=12$; the test is thus accurate at the $1\%$ level and moderately oversized at the $5\%$ and $10\%$ levels, especially when $p=12$. This oversizing is consistent with the fact that the PEU correction involves the inverse of the $(p+1)\times(p+1)$ information matrix, whose estimation error grows with $p$, and with the $O_p(n^{-1/4}(\log n)^{1/2})$ remainder of the uniform representation given in Lemma A.\ref{lemma.score.A2}. The size is expected to improve as $n$ increases relative to $p$. The power of the proposed test increases with $\gamma$ in all configurations and with the sample size; at the 5\% level with $n=200$ it reaches 52.9\% ($p=8$) and 47.4\% ($p=12$) at $\gamma=2.0$. The power loss from $p=8$ to $p=12$ is modest within this range of dimensions, although the simulations should be interpreted with the fixed-$p$ theory of Section 3 in mind.

{\bf Setting II comparison.} Table \ref{tab:benchS2} reports the corresponding results for $p=8$ and $p=12$, and Figure \ref{fig:benchS2} displays the 5\% rejection-rate curves. The cumulative-residual test no longer controls size in these designs: at $p=8$ its empirical null rejection rates are 42.8\%/82.5\%/96.4\% ($n=100$) and 94.7\%/99.5\%/99.8\% ($n=200$) at the 1\%/5\%/10\% levels, which we attribute to the sparsity of the high-dimensional orthant indicator $I(\bm X_j\le\bm X_i)$; at $p=12$ it changes character, losing power entirely at $n=100$ (rejection rates remain at 4\%--5\% for all $\gamma$) and over-rejecting under the null at $n=200$ (20.6\% at the 5\% level) with low, nearly flat power.

\begin{table}[htbp]
\centering
\small
\setlength{\tabcolsep}{4pt}
\caption{Empirical rejection rates (\%) of the proposed test and the kernel test of \cite{Koul2014JMA} under Setting~II, $p=8$ and $p=12$, $n=100$ and $n=200$ ($M=1000$, $B=1000$); notation as in Table \ref{tab:benchS1}. The cumulative-residual test of \cite{Sun2021JKSS} does not control size in these designs (at $p=8$, $n=200$ its empirical null rejection rates are 94.7\%, 99.5\% and 99.8\% at the 1\%, 5\% and 10\% levels) and is omitted.}
\label{tab:benchS2}
\begin{tabular}{l|ccc|cccc|ccc|cccc}
\hline
 & \multicolumn{7}{c|}{$n=100$} & \multicolumn{7}{c}{$n=200$}\\
 & \multicolumn{3}{c|}{size ($\gamma=0$)} & \multicolumn{4}{c|}{power (5\%)} & \multicolumn{3}{c|}{size ($\gamma=0$)} & \multicolumn{4}{c}{power (5\%)}\\
Method & 1\% & 5\% & 10\% & 0.5 & 1.0 & 1.5 & 2.0 & 1\% & 5\% & 10\% & 0.5 & 1.0 & 1.5 & 2.0\\
\hline
\multicolumn{15}{c}{$p=8$}\\
Proposed & 1.6 & 6.9 & 12.8 & 8.9 & 10.6 & 17.5 & 24.8 & 1.1 & 6.7 & 12.0 & 7.5 & 19.7 & 36.5 & 52.9\\
Koul et al.\ (2014) & 0.0 & 2.4 & 10.9 & 6.2 & 18.6 & 45.1 & 71.4 & 0.0 & 2.6 & 9.5 & 9.4 & 38.6 & 74.1 & 94.3\\
\hline
\multicolumn{15}{c}{$p=12$}\\
Proposed & 1.5 & 7.9 & 17.6 & 7.5 & 12.3 & 18.3 & 21.7 & 1.4 & 7.3 & 12.9 & 8.9 & 18.6 & 35.2 & 47.4\\
Koul et al.\ (2014) & 0.0 & 1.0 & 13.4 & 5.0 & 13.3 & 26.1 & 38.2 & 0.0 & 2.4 & 13.8 & 5.0 & 20.0 & 40.6 & 53.6\\
\hline
\end{tabular}
\end{table}

\begin{figure}[htbp]
	\centering
	\begin{subfigure}[b]{0.45\textwidth}
		\centering
		\includegraphics[width=\textwidth]{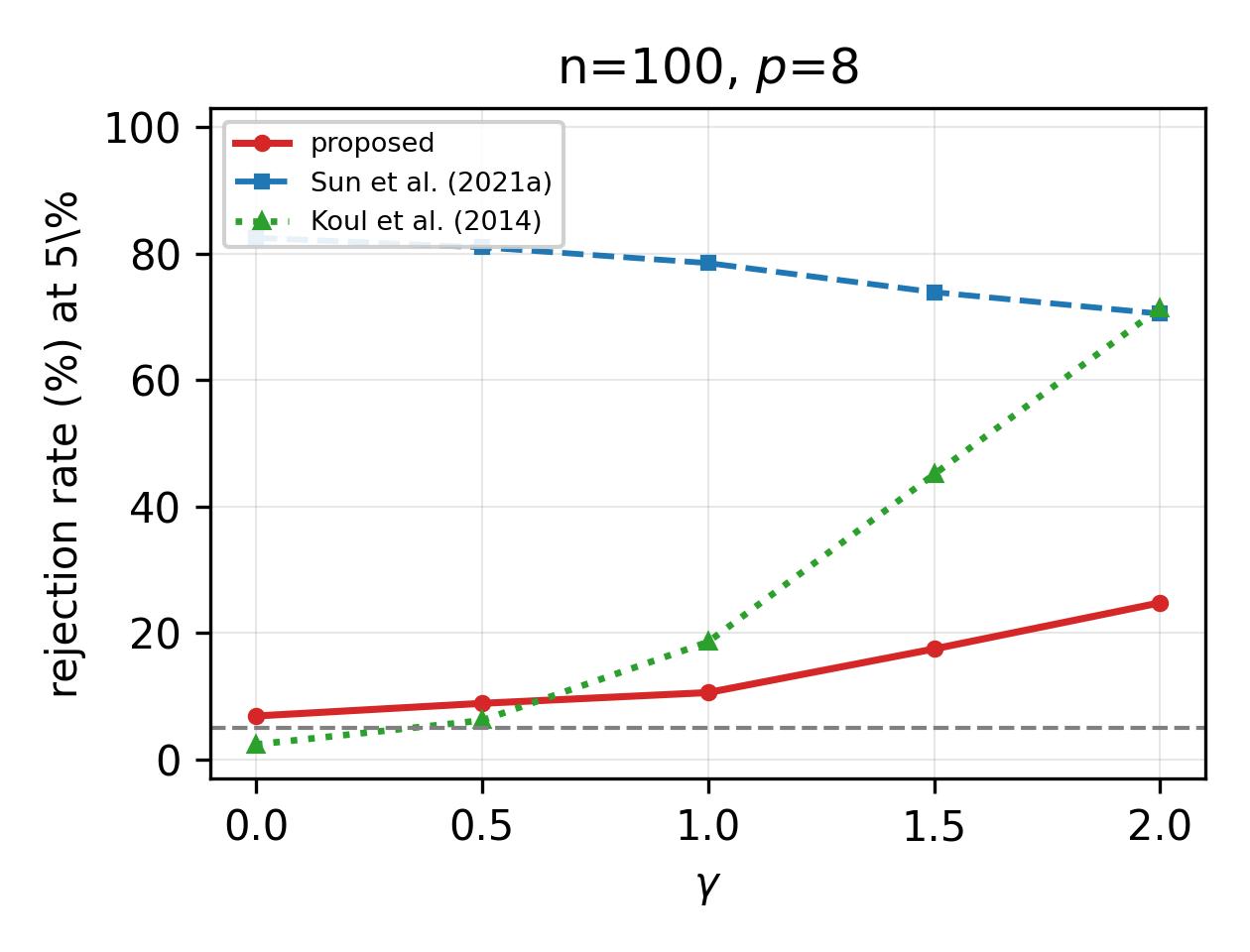}
		\caption{$n$=100, $p$=8}
	\end{subfigure}
	\hfill
	\begin{subfigure}[b]{0.45\textwidth}
		\centering
		\includegraphics[width=\textwidth]{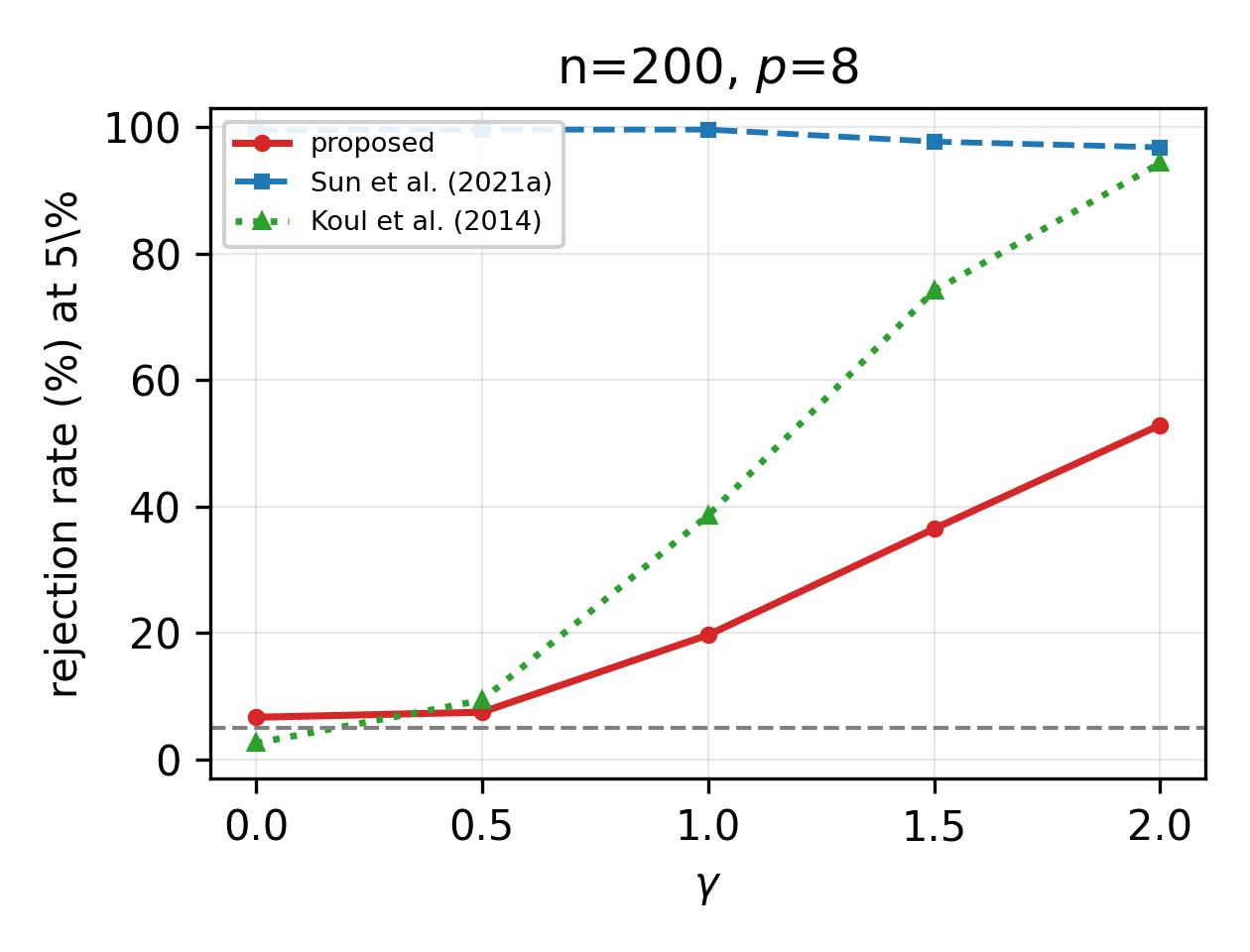}
		\caption{$n$=200, $p$=8}
	\end{subfigure}
	\vspace{1em}
	\begin{subfigure}[b]{0.45\textwidth}
		\centering
		\includegraphics[width=\textwidth]{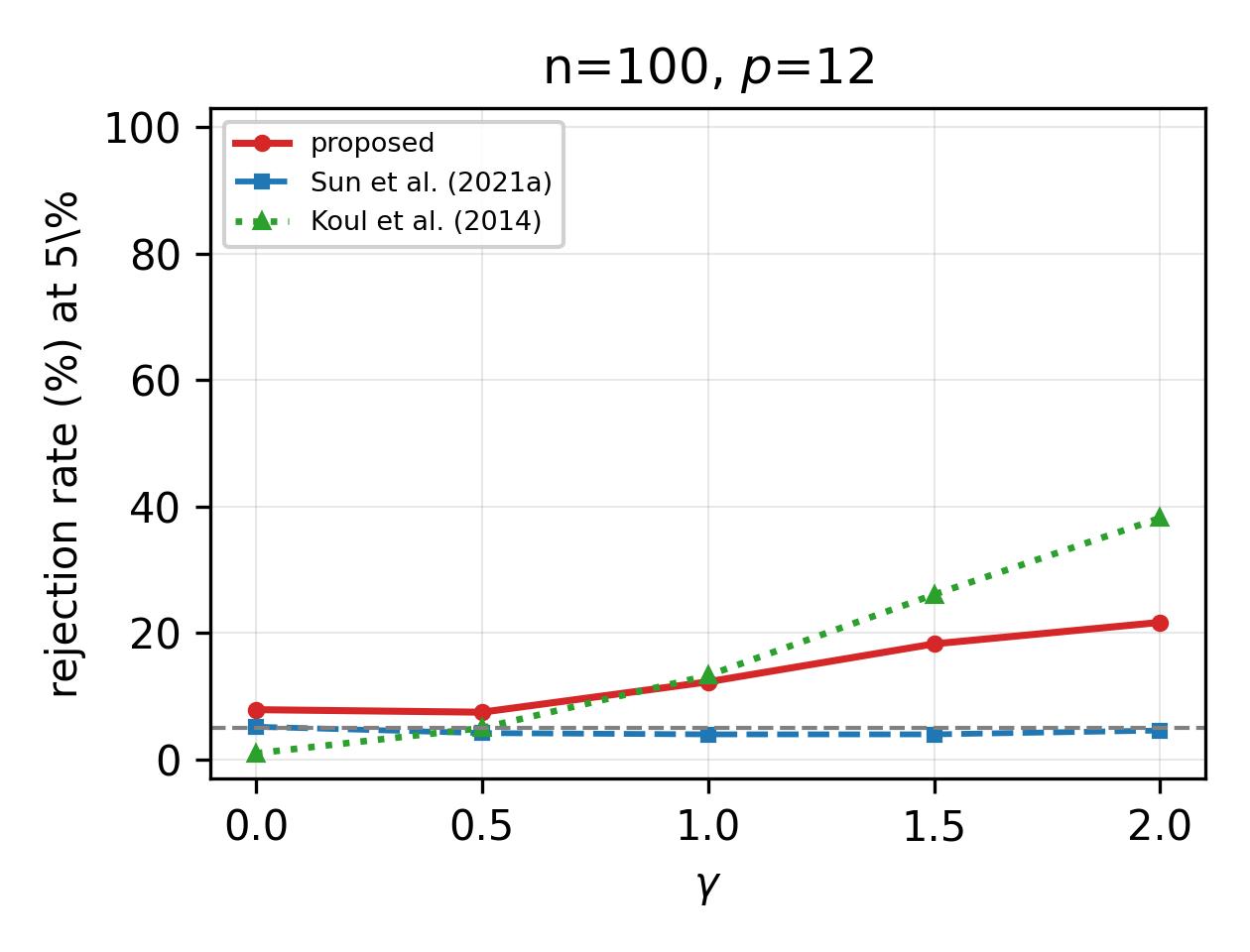}
		\caption{$n$=100, $p$=12}
	\end{subfigure}
	\hfill
	\begin{subfigure}[b]{0.45\textwidth}
		\centering
		\includegraphics[width=\textwidth]{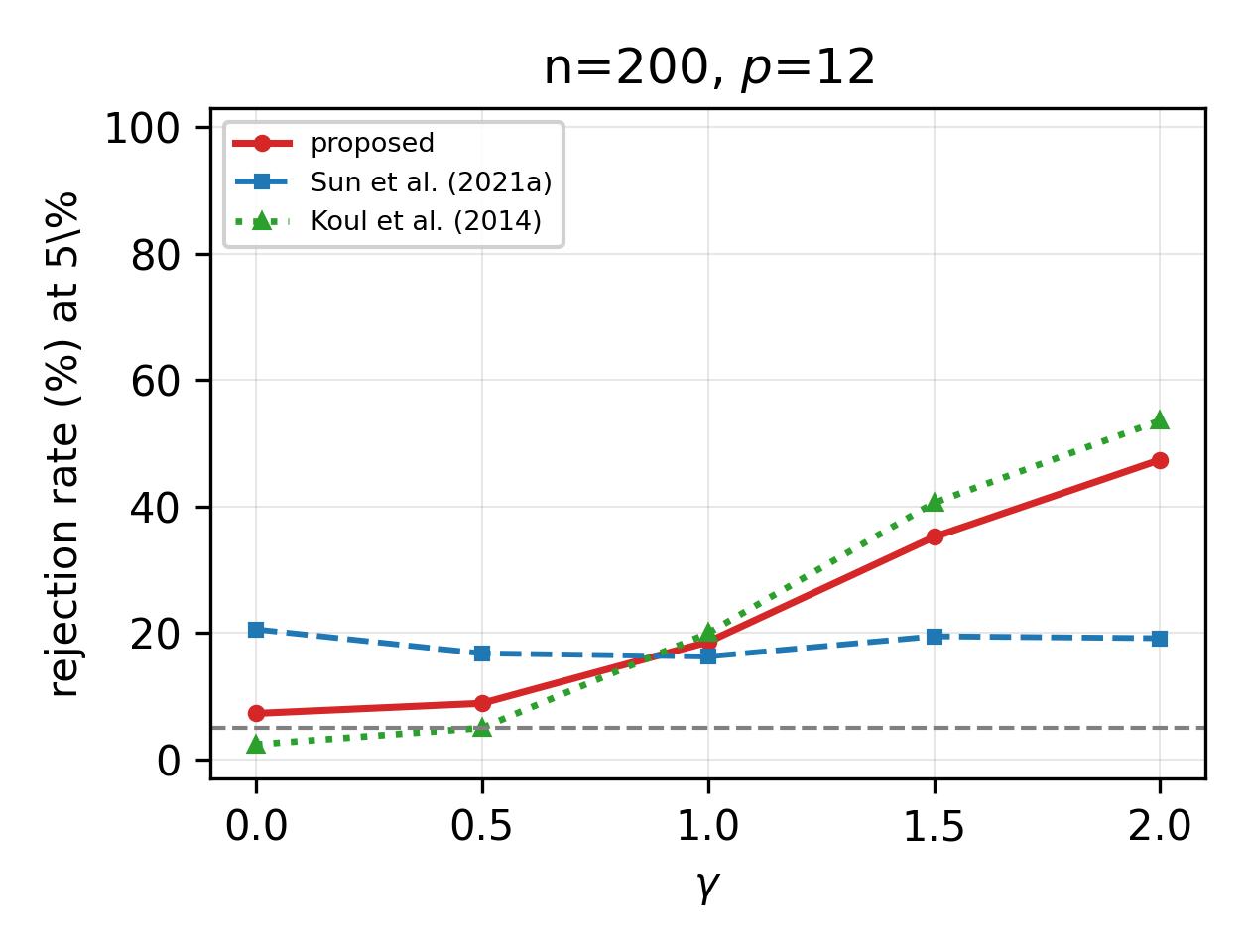}
		\caption{$n$=200, $p$=12}
	\end{subfigure}
	\caption{Empirical rejection rates of the three tests at the 5\% level under Setting~II, $p=8$, $12$ and $n=100$, $200$ ($M=1000$, $B=1000$); the gray dashed line marks the nominal 5\% level.}
	\label{fig:benchS2}
\end{figure}

The kernel test retains power (94.3\% at $\gamma=2$, $p=8$, $n=200$), but its size is again distorted (2.6\% at the 5\% level, and 13.8\% at the 10\% level for $p=12$), and its performance depends on the bandwidth, which must be selected. In contrast, the proposed test remains the procedure whose size stays closest to the nominal level across dimensions (6.7\%--7.9\% at the 5\% level, versus 2.6\% for the kernel test and 99.5\% for the cumulative-residual test at $p=8$, $n=200$) and delivers power that increases with $\gamma$ and with the sample size (52.9\% at $\gamma=2$, $p=8$, $n=200$). In moderate dimensions, the proposed test therefore offers reliable inference without any tuning parameter, whereas the competitors either break down (Sun et al.) or trade in level distortion and bandwidth sensitivity (Koul et al.). Finally, after an initial sort of the index values, each bootstrap replication of the proposed test requires only $O(n)$ operations (for fixed $p$), whereas each replication of the two competing statistics evaluates a $U$-statistic of order $O(n^2)$.

{\bf Unknown scale:}\label{sec:twopar-sim} The theory of Section \ref{subsec:two-parameter} and the empirical applications use the two-parameter version of the test with both $\bm\beta$ and $\sigma$ unknown. To assess its finite-sample performance, we repeat the exercise for the DGP
\[
Y^*=1+X_1+X_2+\gamma(X_1^2+X_2^2)+\sigma e,\qquad
Y=\max(0,Y^*),\qquad e\sim N(0,1),
\]
with $X_1,X_2\sim N(0,1)$, $\sigma=2$ and $\gamma=0,0.2,0.4,0.6$. The mark, the joint score and the information matrix are those of Section \ref{subsec:two-parameter}, and the estimator is the two-parameter maximum likelihood estimator; all other settings are as in Setting~I, with $M=500$ replications and $B=500$ bootstrap replications. Table \ref{tab:twopar} reports the results. Under the null ($\gamma=0$) the proposed test has rejection rates $0.8\%$, $5.0\%$ and $11.4\%$ at the three levels for $n=100$, and $1.4\%$, $5.8\%$ and $12.0\%$ for $n=200$; the slight oversizing at the 10\% level is of the same order as in Setting~II. The power increases rapidly with $\gamma$: at the 5\% level and $n=200$ it rises from 5.8\% at $\gamma=0$ to $17.4\%$, $53.8\%$ and $88.0\%$ at $\gamma=0.2,0.4,0.6$.

\begin{table}[htbp]
\centering
\caption{Empirical rejection rates (in \%) of the proposed test for the two-parameter (unknown-$\sigma$) model, $M=500$ replications, $B=500$. Notation as in Table \ref{tab:set1a}.}
\label{tab:twopar}
\begin{tabular}{c|ccc|ccc}
\hline
$\gamma$ & \multicolumn{3}{c|}{$n=100$} & \multicolumn{3}{c}{$n=200$}\\
 & 1\% & 5\% & 10\% & 1\% & 5\% & 10\%\\
\hline
0.0 & 0.8 & 5.0 & 11.4 & 1.4 & 5.8 & 12.0\\
0.2 & 3.0 & 10.0 & 18.8 & 3.6 & 17.4 & 27.6\\
0.4 & 9.8 & 28.0 & 39.0 & 26.6 & 53.8 & 68.6\\
0.6 & 24.4 & 52.0 & 65.4 & 60.6 & 88.0 & 95.4\\
\hline
\end{tabular}
\end{table}

\subsection{Index-Visible versus Index-Invisible Misspecification}\label{sec:indexinvisible}

By Remark 1, the test is consistent against deviations of the conditional mean that are visible through the index, but it has no power against deviations orthogonal to the index. We illustrate this limitation with the design $\bm X=(X_1,X_2)$ and an intercept, $\bm\beta_0=(0,1,0)^\top$ (so that the index is $X_1$), and the mean-shift alternatives
\[
\text{(visible): } Y^*=X_1+\gamma(X_1^2-1)+e,\qquad
\text{(invisible): } Y^*=X_1+\gamma(X_2^2-1)+e,
\]
with $X_1,X_2,e$ i.i.d.\ $N(0,1)$, $n=200$,  $M=1000$ and $B=1000$; we also apply the two benchmark tests of Section \ref{sec:size-power} to both designs under the same protocol. In the visible design, the perturbation is a function of the index and the test has power that grows quickly with $\gamma$; in the invisible design, the perturbation is, to first order, orthogonal to the index, so the score deviation $\mathbb E[u(\bm W,\bm\beta_0)|X_1]$ vanishes to leading order, and the test should remain close to its nominal size. Table \ref{tab:indexinvisible} confirms the contrast and completes the comparison with existing approaches. In the visible design all three tests detect the deviation: at the 5\% level the proposed test rejects $70.0\%$ and $89.3\%$ of the time at $\gamma=0.5$ and $\gamma=1.0$, the cumulative-residual test $54.8\%$ and $92.1\%$, and the kernel test $89.7\%$ and $100.0\%$. In the invisible design the tests separate exactly as Remark~1 predicts: the proposed test remains close to its nominal size at small perturbations ($6.0\%$ at $\gamma=0.5$) and reacts only through the second-order effect of the mean shift on the conditional score mean ($15.1\%$, $33.6\%$ and $49.4\%$ at $\gamma=1.0,1.5,2.0$), whereas the full-covariate benchmarks detect the index-orthogonal deviation with high power ($57.3\%$ and $97.9\%$ at $\gamma=0.5$ for the cumulative-residual and kernel tests, respectively). This is precisely the trade-off of Remark~1: conditioning on the index delivers reliable size without any tuning, together with the dimension reduction documented in Section \ref{sec:size-power}, at the cost of power against deviations that are invisible through the index, whereas the full-covariate tests have the complementary profile and, as shown in Section \ref{sec:size-power}, cannot be used reliably once the dimension grows. Both benchmarks are also systematically conservative under the null at the 5\% level in this design ($0.6\%$--$0.8\%$ at $\gamma=0$, versus $4.8\%$--$5.2\%$ for the proposed test). The experiment therefore supports interpreting the proposed test as a diagnostic for the index projection of the score conditional moment.

\begin{table}[htbp]
\centering
\small
\setlength{\tabcolsep}{4pt}
\caption{Empirical rejection rates (\%) of the three tests in the index-visible/index-invisible design, $n=200$ ($M=1000$ replications and $B=1000$ bootstrap replications for every method); notation as in Table \ref{tab:benchS1}. For $\gamma=0$ the 1\%, 5\% and 10\% rejection rates are reported (size); for $\gamma>0$ the rejection rate at the 5\% level is reported (power).}
\label{tab:indexinvisible}
\begin{tabular}{l|ccc|cccc|ccc|cccc}
\hline
 & \multicolumn{7}{c|}{Visible: $\gamma(X_1^2-1)$} & \multicolumn{7}{c}{Invisible: $\gamma(X_2^2-1)$}\\
 & \multicolumn{3}{c|}{size ($\gamma=0$)} & \multicolumn{4}{c|}{power (5\%)} & \multicolumn{3}{c|}{size ($\gamma=0$)} & \multicolumn{4}{c}{power (5\%)}\\
Method & 1\% & 5\% & 10\% & 0.5 & 1.0 & 1.5 & 2.0 & 1\% & 5\% & 10\% & 0.5 & 1.0 & 1.5 & 2.0\\
\hline
Proposed & 1.0 & 5.2 & 10.6 & 70.0 & 89.3 & 91.5 & 89.3 & 1.1 & 4.8 & 10.1 & 6.0 & 15.1 & 33.6 & 49.4\\
Sun et al.\ (2021a) & 0.0 & 0.8 & 1.5 & 54.8 & 92.1 & 87.3 & 76.5 & 0.2 & 0.6 & 1.9 & 57.3 & 82.4 & 79.0 & 73.3\\
Koul et al.\ (2014) & 0.2 & 0.8 & 1.8 & 89.7 & 100.0 & 100.0 & 100.0 & 0.0 & 0.6 & 2.0 & 97.9 & 100.0 & 100.0 & 100.0\\
\hline
\end{tabular}
\end{table}

\section{Empirical Applications}
\label{s:EmpApp}

In this section, we apply our proposed score-based specification test to the Michigan Panel Study of Income Dynamics (PSID) dataset and the Extramarital Affairs dataset. 

For these two datasets, the tested null is the index-conditional mark moment implied by $H_0^{2p}$,
\[
\mathbb E\bigl[m(\bm W,\bm\theta_0)\,\big|\,\bm X^\top\bm\beta_0\bigr]=0 \quad \text{a.s.\ for some } \bm\theta_0=(\bm\beta_0^\top,\sigma_0)^\top\in\Theta,
\]
where $m$ and $\bm s$ are the mark and joint score of Section~\ref{subsec:two-parameter}. Note that this restriction is a necessary but not sufficient condition for the specification of the Tobit model $Y^*=\bm X^\top\bm\beta+e$ with $e\sim N(0,\sigma^2)$ independent of $\bm X$. If the model holds for $(\bm\beta_0,\sigma_0)$, then the full score restriction $\mathbb E[\bm s(\bm W,\bm\theta_0)|\bm X]=0$ a.s.\ holds, and hence so does the index projection of the mark moment. Since the test conditions on the one-dimensional index, it detects only deviations that are visible through the index; non-rejection therefore means that the data provide no evidence of \emph{index-visible} misspecification, and it cannot be interpreted as proof that the model is correctly specified (deviations orthogonal to the index are undetectable, and the mark moment constrains only one conditional moment of the censored outcome; see Remark~1 and Section~\ref{sec:indexinvisible}). Since the scale $\sigma$ is unknown in the applications, we implement the two-parameter version of Section \ref{subsec:two-parameter}, i.e., the feasible process \eqref{eqn:two-param-process} based on the joint score \eqref{eqn:two-param-mark}, with the multiplier bootstrap \eqref{eqn:two-param-bootstrap}; the estimation of $(\bm\beta,\sigma)$ is by maximum likelihood, and the bootstrap uses $B=1000$ replications. Table \ref{tab:emp} reports the maximum likelihood estimates for both datasets, computed using the same code as in the tests. The Affairs column matches the estimates in \cite{Greene2012econometric} (Table 19.3) to within $0.002$, providing a check on the implementation. The PSID likelihood is numerically delicate, a phenomenon documented for these data by \cite{McCullough2004}: some estimates quoted in the literature correspond to a suboptimal point, and we report the actual maximizer, obtained by  quasi-Newton (L-BFGS-B) optimization from several starting values with the analytic score (the norm of the average score is below $10^{-6}$), which is the fit used to compute the test statistic below. Figure~\ref{fig:hist} displays the two outcome distributions.

\begin{table}[htbp]
\centering
\small
\caption{Maximum likelihood estimates of the two-parameter Tobit model for the two empirical applications, from the fitted models used in the tests below; non-wife income is measured in thousands of 1975 dollars.}
\label{tab:emp}
\begin{tabular}{lll}
\hline
 & PSID1976 & Affairs\\
\hline
Intercept & 965.72 & 6.8174\\
Non-wife income / marriage rating & $-8.816$ & $-2.2785$\\
Education & 80.64 & $0.0599$\\
Experience & 131.55 & --\\
Experience$^2$ & $-1.864$ & --\\
Age & $-54.41$ & $-0.1799$\\
Children $<6$ / children indicator & $-894.00$ & $1.1236$\\
Children 6--18 & $-16.24$ & --\\
Years married & -- & $0.5090$\\
Religiousness & -- & $-1.6885$\\
Occupation & -- & $0.3071$\\
$\sigma$ & 1122.01 & 8.2577\\
Log-likelihood & $-3819.10$ & $-705.13$\\
\hline
Test statistic $KS_n$ (proposed) & 0.547 & 0.615\\
$p$-value (proposed test) & 0.519 & 0.117\\
$p$-value (Sun et al., 2021a) & 0.83 & 0.97\\
$J_n$ (Koul et al., 2014) & $-0.895$ & 0.628\\
$p$-value (Koul et al., 2014) & 0.82 & 0.25\\
\hline
\end{tabular}
\par\smallskip
{\small Note: The $p$-values in the last three rows are wild-bootstrap $p$-values ($B=1000$) for the cumulative-residual test of \cite{Sun2021JKSS} and the kernel test of \cite{Koul2014JMA}, implemented in their two-parameter form as in Section~\ref{sec:benchmarks}; $J_n$ is the standardized kernel statistic (asymptotic-normal $p$-values are $0.815$ and $0.265$, respectively).}
\end{table}

\begin{figure}[htbp]
	\centering
	\begin{subfigure}[b]{0.48\textwidth}
		\centering
		\includegraphics[width=\textwidth]{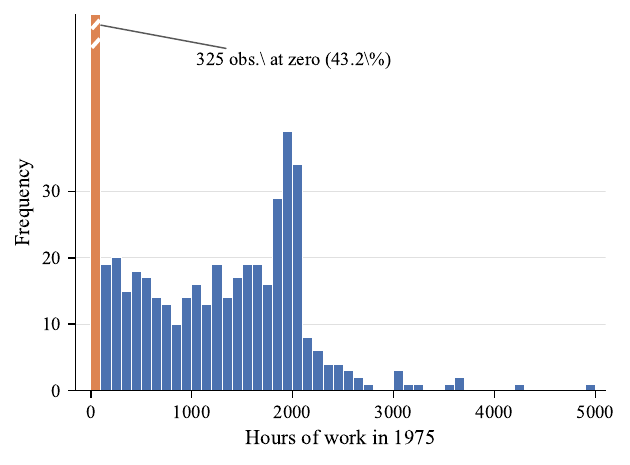}
		\caption{PSID data}
		\label{fig:hist_hours}
	\end{subfigure}
	\hfill
	\begin{subfigure}[b]{0.48\textwidth}
		\centering
		\includegraphics[width=\textwidth]{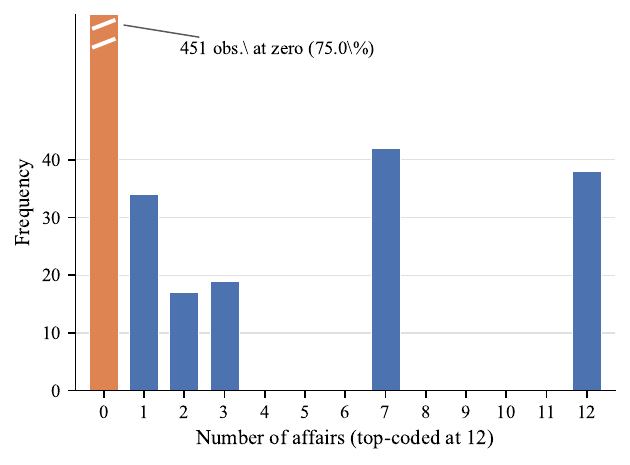}
		\caption{Extramarital Affairs data}
		\label{fig:hist_affairs}
	\end{subfigure}
	\caption{Distributions of the two empirical datasets. 
	The darker bar marks the censoring point in each panel; it is truncated at the white break marks because the mass at zero far exceeds the remaining frequencies (43.2\% of the PSID observations report zero hours of work and 75.0\% of the Affairs observations report zero affairs, versus a maximum bin frequency of 39 and 42, respectively, among positive values). The Affairs variable takes the coded values $0,1,2,3,7,12$ (top-coded at 12).}
	\label{fig:hist}
\end{figure}

\subsection{PSID data}

The PSID dataset, available via the \textit{AER} package (\texttt{PSID1976}), contains 753 married women; 325 women did not work for pay in 1975, so 43.16\% of the observations are censored at zero. \cite{Winkelmann2009AM} employ the Tobit model for estimation and interpretation of these data.
The response variable is annual hours of work, and there are 7 covariates: non-wife income (family income net of the woman's labor earnings), years of schooling, years of experience, years of experience squared, age, number of children younger than 6, and number of children aged 6 to 18; the estimation results for the Tobit model are presented in Table 7.1 of \cite{Winkelmann2009AM}.

\cite{Greene2012econometric}  notes that the Tobit model can not only predict the clustering of data at the censoring point but also the clustering of observations near the censoring point. 
Therefore, as shown in Figure~\ref{fig:hist}(a), the distribution of ``hours of work'' exhibits a concentration of observations near the censoring point, aligning with the assumptions of the Tobit model.

We fit the two-parameter Tobit model by maximum likelihood (Table \ref{tab:emp}) and apply the proposed test with $B=1000$ bootstrap replications: the test statistic is $KS_n=0.547$ and the $p$-value is 0.519, so the index-conditional mark moment is not rejected: the data provide no evidence of index-visible misspecification. The large $p$-value indicates that the index-conditional mark moment is compatible with these data.

\subsection{Extramarital Affairs data}

The Extramarital Affairs data are Fair's 601-observation survey sample used in \cite{Greene2012econometric} (Table 19.3); the data are available in the \textit{AER} package (\texttt{Affairs}).
Given the high frequency of zero values in the dataset (451 observations, about 75.0\%), we employ a Tobit model to account for this censoring structure.
The core response variable is the number of extramarital affairs (top-coded at 12), and there are 7 covariates: marriage satisfaction rating (1--5), age (years), years married, children indicator, religiousness (1--5), years of education, and occupational level (1--7).

From Figure~\ref{fig:hist}(b), the distribution of the number of affairs displays a large mass point exactly at zero (75.0\%), while the positive values are sparse and far from zero, indicating potential misspecification if the Tobit model is applied.

Similarly, for the Extramarital Affairs data, we fit the two-parameter Tobit model by maximum likelihood and apply the proposed test with $B=1000$ bootstrap replications: the test statistic is $KS_n=0.615$ and the $p$-value is 0.117, so there is no evidence of index-visible misspecification at the conventional 5\% level (nor at the 10\% level). This finding is in line with the Tobit estimates of \cite{Greene2012econometric} (Table 19.3), whose standard specification is not overwhelmingly rejected on this sample; as emphasized above, the test targets the index projection of the mark moment, so these findings do not rule out misspecification in directions that are invisible through the index (Remark~1 and Section~\ref{sec:indexinvisible}). We note that results can differ on the larger 6,366-observation Redbook subsample analyzed elsewhere in the literature, which the code provided with this paper can accommodate.

For comparison, we also apply the two benchmark tests of Section~\ref{sec:benchmarks} to both datasets, using their full-covariate statistics, the same fitted Tobit models, and a wild bootstrap with $B=1000$ replications. The results are reported in the last three rows of Table~\ref{tab:emp}. Neither benchmark rejects at the 5\% level: the cumulative-residual test of \cite{Sun2021JKSS} gives wild-bootstrap $p$-values of $0.83$ (PSID) and $0.97$ (Affairs), and the kernel test of \cite{Koul2014JMA} gives $0.82$ and $0.25$, respectively. These conclusions agree with the proposed test. The large $p$-values of the cumulative-residual test reflect the finite-sample conservativeness of its wild bootstrap documented in Section~\ref{sec:benchmarks}, while the kernel test is consistent with the proposed test in ranking the Affairs data as the more suspicious of the two.

\section{Conclusion}
\label{s:conclusion}

This paper proposes a score-based empirical process for specification testing in Tobit regression models. We develop asymptotic theory under the null, the alternative, and local alternatives, for both the known-scale and the unknown-scale versions of the test. Monte Carlo simulations demonstrate that the proposed test controls size well 
and delivers competitive power. An unknown-scale design, an index-visible/index-invisible power study, and comparisons with existing specification tests complete the picture. Two empirical applications further confirm the practical utility of our approach. By avoiding the curse of dimensionality inherent in nonparametric methods and requiring neither bandwidth selection nor tuning parameters, the proposed test effectively addresses the PEU. The multiplier bootstrap reproduces the first-order uniform representation of the feasible process, so that it is valid under correct specification.

\section*{CRediT Authorship Contribution Statement}
Xuewen Li: Writing - review - original draft, Data analysis. 
Xiaojun Song: Conceptualization, Writing - review \& editing, Validation.
Xiaogang Wang: Conceptualization, Writing - review - original draft, Organization, Proof, Data analysis.
Feipeng Zhang: Conceptualization, Writing - review \& editing, Proof, Validation.

\section*{Acknowledgments}
Wang's work is supported by the Ningxia Basic Science Research Center of Mathematics (No. 2026NXSXZX0203) and the Key Program of Ningxia Natural Science Foundation (No. 2023AAC02043).

\section*{Data Availability Statement}

The data that support the findings of this study are available in the R package AER. 
The empirical applications use the \texttt{PSID1976} and \texttt{Affairs} datasets contained in this package.


\bibliography{mybibfile1020}

\newpage
\setcounter{section}{1}
\renewcommand{\thesection}{\Alph{section}}
\setcounter{apptheorem}{0}
\section*{Appendix A: Proofs of the Score-based Specification Test}

We impose the following assumptions.
\begin{itemize}
	\item[(A1)] The observations $(Y_i,\bm{X}_i),i=1,2,\cdots,n$ are independent and identically distributed, $\bm\beta_0\neq0$, and the index family is uniformly continuous on compact subsets of the parameter space away from the origin: for every compact $\mathcal K\subset\mathcal B\setminus\{\bm 0\}$ there exists $C_{\mathcal K}<\infty$ such that
	\[\sup_{\bm\beta\in\mathcal K}\ \sup_{t\in\mathbb R}\ \mathbb P\bigl(|\bm X_1^\top\bm\beta-t|\le\gamma\bigr)\le C_{\mathcal K}\,\gamma\qquad\text{for all }\gamma>0\]
	and, for all $\bm\beta_1,\bm\beta_2\in\mathcal K$ and all $t\in\mathbb R$,
	\[\mathbb P\Bigl(\bigl\{\bm X_1^\top\bm\beta_1\le t<\bm X_1^\top\bm\beta_2\bigr\}\cup\bigl\{\bm X_1^\top\bm\beta_2\le t<\bm X_1^\top\bm\beta_1\bigr\}\Bigr)\le C_{\mathcal K}\,\|\bm\beta_1-\bm\beta_2\|.\]
	For example, both bounds hold if $\bm X_1$ has a continuously distributed component whose conditional density given the remaining components is bounded and whose coefficient is bounded away from zero on $\mathcal K$; see Lemma A.\ref{lem:uniform-continuity}. The two bounds are what is used below: the first gives the uniform continuity of the index distribution, and the second bounds the symmetric difference of the indicator increments, which makes the indicator-marked classes continuous in $L^2(P)$ and yields the rate in Lemma A.\ref{lemma.score.A2}; together they replace the requirement that the distribution of $\bm X_1^\top\bm\beta$ be continuous for every $\bm\beta\in\mathcal B$, which cannot hold at $\bm\beta=0$.
    \item[(A2)] The true parameter ${\bm\beta}_0$ lies in the interior of the compact parameter space $\mathcal{B}$. The matrices $\mathbb E(\bm X_1 \bm X_1^\top)$ and $\mathbb E[\ell(\bm X_1)\bm X_1 \bm X_1^\top]$, where $\ell(\bm X)=\Phi(\bm X^\top\bm\beta_0)-\Phi(-\bm X^\top\bm\beta_0)\dot\lambda(-\bm X^\top\bm\beta_0)>0$, are positive definite.
    \item[(A3)] Under correct specification, the full conditional moment restriction $\mathbb E[u(\bm W_1,\bm\beta_0)|\bm X_1]=0$ a.s.\ holds (it is implied by \eqref{eqn:SpecTestP}); the moment conditions 
    $\mathbb E\|{\bm X}_1\|^4<\infty$, $\mathbb E[\sup_{\bm\beta\in\mathcal B}|u(\bm W_1,\bm\beta)|^4]<\infty$, and $\mathbb E[u^2(\bm W_1,\bm\beta_0)\|\bm X_1\|^2]<\infty$ are instead maintained under both the null and the alternative hypotheses.  
     \item[(A4)] The estimator $\hat{\bm\beta}$ is the score-based maximum likelihood estimator, i.e., it solves 
     $$\sum_{i=1}^n u(\bm W_i,\hat{\bm\beta})\bm X_i=0$$ up to a numerical tolerance that is $o_p(n^{-1/2})$, and admits the influence expansion
     \[\sqrt n(\hat{\bm\beta}-\bm\beta_0)=H^{-1}(\bm\beta_0)\frac{1}{\sqrt n}\sum_{i=1}^n u(\bm W_i,\bm\beta_0)\bm X_i+o_p(1).\]
     Under correct specification, this expansion is standard under (A2), (A5), and the identification of $\bm\beta_0$. It is stated as an assumption only to keep the proofs focused on the empirical-process arguments.
     \item[(A5)] The log-likelihood function is three times continuously differentiable with respect to $\bm\beta$, and its first three derivatives are dominated, uniformly in $\bm\beta\in\mathcal B$, by integrable functions of $(\bm W_1,\bm X_1)$.
     \item[(A6)] The function classes $\{u(\bm W,\bm\beta)I(\bm X^\top\bm\beta\le t):\bm\beta\in\mathcal B,\ t\in\mathbb R\}$ and $\{u(\bm W,\bm\beta)\bm X:\bm\beta\in\mathcal B\}$ are $\mathbb P$-Donsker with square-integrable envelopes. Sufficient primitive conditions are given in Lemma A.\ref{lemma.donsker} below, whose proof uses only (A1)--(A5).
     \item[(A7)] The multipliers $\{V_i\}_{i=1}^n$ are i.i.d.\ with $\mathbb E (V_i)=0$, $\mathbb E (V_i^2)=1$, $\sup_i\mathbb E(|V_i|^{2+\delta})<\infty$ for some $\delta>0$, and are independent of the sample.
\end{itemize}
The normalization $\sigma=1$ in the main text is without loss of generality under  correct specification with known scale; for unknown $\sigma$, see Section \ref{subsec:two-parameter} and Remark A.1.

{\bf Remark A.1:} (Unknown scale.) The two-parameter version of the test with unknown $\sigma$ is developed in Section \ref{subsec:two-parameter}: the joint score, the information matrix, the feasible process \eqref{eqn:two-param-process} and its bootstrap \eqref{eqn:two-param-bootstrap}, together with the null limit, consistency, local power and bootstrap validity theorems. The interpretation of the $\sigma=1$ normalization used in the main text is discussed there: the $\sigma=1$ score tests the joint hypothesis $\{$correct functional form and $\sigma=1\}$, and rejects correctly specified models solely because of a non-unit scale.

{\bf Remark A.2:}  The standard latent Tobit model $Y_i^*=\bm X_i^\top\bm\beta+e_i$ with observed response $Y_i=\max(0,Y_i^*)$ has a zero censoring threshold; it can be extended to a non-zero threshold,  $Y_i=\max(a,Y_i^*)$; see \cite{Carson2007Tobit}. With censoring, $\hat a=\min_{1\le i\le n}Y_i$ equals $a$ exactly whenever at least one observation is censored, an event whose probability tends to one, so the threshold can be treated as known in the asymptotic analysis and the generalized residual and the score are formed from the shifted response $Y_i-\hat a$. A complete treatment of the estimation effect of $\hat a$, including the modified regularity conditions, is left for future research.

\begin{lemma}\label{lemma.score.A1}
	Under Assumptions (A1)--(A5), with the random-weight sample versions
	\[
	H_n(\bm\beta)=\frac1n\sum_{i=1}^n\bigl[d_i-(1-d_i)\dot\lambda(-\bm X_i^\top\bm\beta)\bigr]\bm X_i\bm X_i^\top
	\]
	and
	\[
	G_n(t,\bm\beta)=\frac1n\sum_{i=1}^n\bigl[d_i-(1-d_i)\dot\lambda(-\bm X_i^\top\bm\beta)\bigr]\bm X_i I(\bm X_i^\top\bm\beta\le t),
	\]
	the following convergences hold uniformly in $\bm\beta\in\mathcal B$ and $t\in\mathbb R$:
	\begin{align*}
		  &H_n(\bm\beta) \overset{p}{\rightarrow} H^{w}(\bm\beta),\qquad
		  G_n(t,\bm\beta) \overset{p}{\rightarrow} G^{w}(t,\bm\beta),
	\end{align*}
	where $H^{w}(\bm\beta)=\mathbb E\{[d-(1-d)\dot\lambda(-\bm X^\top\bm\beta)]\bm X\bm X^\top\}$ and $G^{w}(t,\bm\beta)=\mathbb E\{[d-(1-d)\dot\lambda(-\bm X^\top\bm\beta)]\bm X I(\bm X^\top\bm\beta\le t)\}$. Replacing the random weights by their conditional means gives the same conclusion with the limits $H(\bm\beta)=\mathbb E[\ell(\bm X,\bm\beta)\bm X\bm X^\top]$ and $G(t,\bm\beta)=\mathbb E[\ell(\bm X,\bm\beta)\bm X I(\bm X^\top\bm\beta\le t)]$, where $\ell(\bm X,\bm\beta)=\Phi(\bm X^\top\bm\beta)-\Phi(-\bm X^\top\bm\beta)\dot\lambda(-\bm X^\top\bm\beta)$. Under  correct specification with parameter $\bm\beta_0$ the two limits coincide at $\bm\beta=\bm\beta_0$, since $\mathbb E[d|\bm X]=\Phi(\bm X^\top\bm\beta_0)$; in particular, $H(\bm\beta_0)=H^{w}(\bm\beta_0)$ is the information matrix.
\end{lemma}
\noindent{\bf Proof of Lemma A.\ref{lemma.score.A1}}  The random weight $d-(1-d)\dot\lambda(-\bm X^\top\bm\beta)$ is bounded in $[0,1]$ (indeed $-1<\dot\lambda(x)<0$ for all $x$) and continuous in $\bm\beta$; the classes $\{[d-(1-d)\dot\lambda(-\bm X^\top\bm\beta)]\bm X\bm X^\top\}$ and $\{[d-(1-d)\dot\lambda(-\bm X^\top\bm\beta)]\bm X I(\bm X^\top\bm\beta\le t)\}$ are therefore Glivenko--Cantelli with square-integrable envelopes (the latter by the VC property of the indicator class; cf.\ the proof of Lemma A.\ref{lemma.donsker} below), and the uniform law of large numbers \citep[Section 2.4]{Vaart1996Weak} gives the stated convergences. Replacing the random weights by their conditional means leaves a uniformly bounded class and the same argument applies, yielding the limits $H(\bm\beta)$ and $G(t,\bm\beta)$.  Under correct specification, at $\bm\beta=\bm\beta_0$, $\mathbb E[d|\bm X]=\Phi(\bm X^\top\bm\beta_0)$ by \eqref{eqn:d-mean}, so the two population forms coincide and $H(\bm\beta_0)$ equals the information matrix $\mathbb E[(d-(1-d)\dot\lambda(-\bm X^\top\bm\beta_0))\bm X\bm X^\top]$.

\begin{lemma}\label{lemma.score.A2}
	Under correct specification with parameter $\bm\beta_0$ and Assumptions (A1)--(A6),
	\begin{align*}
		\sup_{t\in\mathbb R} \left|U_n(t,\hat{\bm\beta}) - U_n(t,\bm\beta_0) + G^\top(t,\bm\beta_0)\sqrt{n}(\hat{\bm\beta} - \bm\beta_0)\right| = o_p(1).
	\end{align*}
	The decomposition is proved under correct specification, where the indicator-jump part is centered because $\mathbb E[u(\bm W,\bm\beta_0)|\bm X]=0$ and the random-weight and expected-weight forms of $G$ coincide; it extends to the local alternatives $H_{1n}$ of Section 3.2 (the corresponding mean term is $o(1)$), while under a fixed alternative the analogous expansion, centered at the pseudo-true value $\bm\beta^*$, has indicator-jump remainder  $O_p(n^{1/4})=o_p(n^{1/2})$ uniformly in $t$ (the mean-map part is $O_p(n^{1/4})$ and dominates the empirical-process part by the symmetric-difference bound in (A1)), which is all that the $n^{-1/2}$-normalized limit of Theorem 3.4 requires.
\end{lemma}
\noindent{\bf Proof of Lemma A.\ref{lemma.score.A2}} Write
\begin{align*}
U_n(t,\hat{\bm\beta}) = U_n(t,\bm\beta_0) + \mathcal R_{1n}(t) + \mathcal R_{2n}(t),
\end{align*}
with
\begin{align*}
\mathcal R_{1n}(t) &= \frac{1}{\sqrt{n}}\sum_{i=1}^n \left[u(\bm W_i,\hat{\bm\beta})-u(\bm W_i,\bm\beta_0)\right]I(\bm X_i^\top\hat{\bm\beta}\le t),\\
\mathcal R_{2n}(t) &= \frac{1}{\sqrt{n}}\sum_{i=1}^n u(\bm W_i,\bm\beta_0)\left[I(\bm X_i^\top\hat{\bm\beta}\le t)-I(\bm X_i^\top\bm\beta_0\le t)\right].
\end{align*}
{\em Step 1 (smooth part).} By the mean value theorem, with mean-value points $\bar{\bm\beta}_i$ on the segments between $\hat{\bm\beta}$ and $\bm\beta_0$,
\begin{align*}
\mathcal R_{1n}(t) = \left[\frac{1}{n}\sum_{i=1}^n \partial_{\bm\beta}u(\bm W_i,\bar{\bm\beta}_i) I(\bm X_i^\top\hat{\bm\beta}\le t)\right]^\top \sqrt{n}(\hat{\bm\beta}-\bm\beta_0),
\end{align*}
where $\partial_{\bm\beta}u(\bm W,\bm\beta)=-(d-(1-d)\dot\lambda(-\bm X^\top\bm\beta))\bm X$. The functions $\partial_{\bm\beta}u(\bm W,\bm\beta)I(\bm X^\top\bm\beta\le t)$, indexed by $(\bm\beta,t)\in\mathcal B\times\mathbb R$, form a class of the form treated in Lemma A.\ref{lemma.donsker} (the gradient is again Lipschitz in $\bm\beta$ with square-integrable Lipschitz constant under (A3)), hence $\mathbb P$-Donsker and Glivenko--Cantelli, and by  the uniform law of large numbers for this class together with the centered stochastic equicontinuity of Lemma A.\ref{lemma.donsker}, using $\sup_{1\le i\le n}\|\bar{\bm\beta}_i-\bm\beta_0\|=o_p(1)$ and $\hat{\bm\beta}\overset{p}{\to}\bm\beta_0$,
\begin{align*}
\frac{1}{n}\sum_{i=1}^n \partial_{\bm\beta}u(\bm W_i,\bar{\bm\beta}_i) I(\bm X_i^\top\hat{\bm\beta}\le t) = \frac{1}{n}\sum_{i=1}^n \partial_{\bm\beta}u(\bm W_i,\bm\beta_0) I(\bm X_i^\top\bm\beta_0\le t) + o_p(1)
\end{align*}
uniformly in $t$, and the right-hand side converges in probability, uniformly in $t$, to
\begin{align*}
\mathbb E[\partial_{\bm\beta}u(\bm W,\bm\beta_0)I(\bm X^\top\bm\beta_0\le t)]
&= -\mathbb E[(d-(1-d)\dot\lambda(-\bm X^\top\bm\beta_0))\bm X I(\bm X^\top\bm\beta_0\le t)]\\
&= -G(t,\bm\beta_0),
\end{align*}
using  the identity $\mathbb E[d|\bm X]=\Phi(\bm X^\top\bm\beta_0)$ under correct specification, under which the random-weight form of $G$ equals $G(t,\bm\beta_0)$. Hence $\mathcal R_{1n}(t) = -G^\top(t,\bm\beta_0)\sqrt{n}(\hat{\bm\beta}-\bm\beta_0)+o_p(1)$ uniformly in $t$.

{\em Step 2 (indicator-jump part), under $H_0$.} Let $T_n(t,\bm\beta)=\frac1{\sqrt n}\sum_{i=1}^n u(\bm W_i,\bm\beta_0)I(\bm X_i^\top\bm\beta\le t)$, so that $\mathcal R_{2n}(t)=T_n(t,\hat{\bm\beta})-T_n(t,\bm\beta_0)$. The class $\mathcal G=\{u(\bm W,\bm\beta_0)I(\bm X^\top\bm\beta\le t):\bm\beta\in\mathcal B,\ t\in\mathbb R\}$ is $\mathbb P$-Donsker by Lemma A.\ref{lemma.donsker}, and the map $(\bm\beta,t)\mapsto g_{\bm\beta,t}=u(\bm W,\bm\beta_0)I(\bm X^\top\bm\beta\le t)$ is continuous from $(\mathcal B\times\mathbb R,\rho)$ into $L^2(P)$: indeed, for $\rho((\bm\beta_1,t_1),(\bm\beta_2,t_2))<\delta$,
\begin{align*}
\|g_{\bm\beta_1,t_1}-g_{\bm\beta_2,t_2}\|_2^2 &\le \mathbb E\Bigl[u^2(\bm W,\bm\beta_0)\,\mathbf 1\{I(\bm X^\top\bm\beta_1\le t_1)\neq I(\bm X^\top\bm\beta_2\le t_2)\}\Bigr] \\[2pt]
&\to 0
\end{align*}
as $\delta\downarrow 0$, since the symmetric-difference event implies that one of the indices falls within distance $(\|\bm X\|+1)\delta$ of its threshold, and its probability is $O(\delta)$ uniformly in $t$ by the two bounds in (A1). Under $H_0$ the class $\mathcal G$ is centered: $\mathbb E[u(\bm W,\bm\beta_0)|\bm X]=0$ implies $Pg_{\bm\beta,t}=\mathbb E[I(\bm X^\top\bm\beta\le t)\,\mathbb E[u(\bm W,\bm\beta_0)|\bm X]]=0$ for every $(\bm\beta,t)$, so that $T_n(\cdot)=\mathbb G_n(g_\cdot)$ on this class and the stochastic equicontinuity of Lemma A.\ref{lemma.donsker} applies to $T_n$ itself. By the stochastic equicontinuity of the centered process $\mathbb G_n$ over the Donsker class $\mathcal G$ \citep[Section 1.5]{Vaart1996Weak} and since $\rho((\hat{\bm\beta},t),(\bm\beta_0,t))=\|\hat{\bm\beta}-\bm\beta_0\|=O_p(n^{-1/2})=o_p(1)$, we obtain, for every $\epsilon>0$ and $\eta>0$,
\begin{align*}
\mathbb P\!\left(\sup_{t\in\mathbb R}|\mathcal R_{2n}(t)|>\epsilon\right)
&\le \mathbb P\!\left(\sup_{\rho(\theta_1,\theta_2)<\delta}|T_n(\theta_1)-T_n(\theta_2)|>\epsilon\right)\\
&\quad + \mathbb P(\|\hat{\bm\beta}-\bm\beta_0\|>\delta) \le \eta + o(1)
\end{align*}
for $\delta>0$ small enough and then $n$ large enough; hence $\sup_t|\mathcal R_{2n}(t)|=o_p(1)$. Combining the two steps completes the proof. We note for later use that Step 2 uses the null only through the centering $\mathbb E[u(\bm W,\bm\beta_0)|\bm X]=0$ (under the local alternatives $H_{1n}$ the remaining mean term $\mathbb E[a(\bm X^\top\bm\beta_0)(I_{\bm\beta_1,t_1}-I_{\bm\beta_2,t_2})]$ is $o(1)$, and under a fixed alternative the remainder, centered at the pseudo-true value, is  $O_p(n^{1/4})$ uniformly in $t$), and that the approximation error of Lemma A.\ref{lemma.score.A2} is of order $O_p(n^{-1/4}(\log n)^{1/2})$ uniformly in $t$, the rate being driven by the indicator-jump term (the empirical-process modulus over $L^2$-balls of radius $O(n^{-1/4})$).

\begin{lemma}\label{lemma.score.A3}
	Under Assumptions (A1)--(A5), we have
	\begin{align*}
\frac{1}{\sqrt{n}} \sum_{i=1}^n u(\bm{W}_i,\hat{\bm\beta}) \bm{X}_i
	= \frac{1}{\sqrt{n}} \sum_{i=1}^n u(\bm{W}_i,\bm\beta_0) 
    \bm{X}_i 
		- H^{w}(\bm\beta_0) \sqrt{n} (\hat{\bm\beta} - \bm\beta_0) + o_p(1),
	\end{align*}
where $H^{w}(\bm\beta_0)=\mathbb E[(d-(1-d)\dot\lambda(-\bm X^\top\bm\beta_0))\bm X\bm X^\top]$ is the random-weight limit of Lemma A.\ref{lemma.score.A1}; under  correct specification, $H^{w}(\bm\beta_0)=H(\bm\beta_0)=\mathbb E[\ell(\bm X)\bm X\bm X^\top]$.
\end{lemma}

{\bf Proof of Lemma A.\ref{lemma.score.A3}} 
By a mean value expansion, with mean-value points $\bar{\bm\beta}_i$ on the segments between $\hat{\bm\beta}$ and $\bm\beta_0$,
\begin{align*}
	&\frac{1}{\sqrt{n}} \sum_{i=1}^n u(\bm{W}_i,\hat{\bm\beta}) \bm{X}_i\\
		&= \frac{1}{\sqrt{n}} \sum_{i=1}^n u(\bm{W}_i,\bm\beta_0) \bm{X}_i 
			+ \frac{1}{\sqrt{n}}\sum_{i=1}^n (u(\bm{W}_i,\hat{\bm\beta})-u(\bm{W}_i,\bm\beta_0))  \bm{X}_i \\
	&= \frac{1}{\sqrt{n}} \sum_{i=1}^n u(\bm{W}_i,\bm\beta_0) \bm{X}_i 
		- \left[\frac{1}{n}\sum_{i=1}^n [d_i-(1-d_i)\dot{\lambda}(-\bm{X}_i^\top \bar{\bm\beta}_i)]\bm{X}_i \bm{X}_i^\top\right] \sqrt{n}(\hat{\bm\beta} - \bm\beta_0)\\
	&= \frac{1}{\sqrt{n}} \sum_{i=1}^n u(\bm{W}_i,\bm\beta_0) \bm{X}_i 
		- H^{w}(\bm\beta_0) \sqrt{n}(\hat{\bm\beta} - \bm\beta_0) + o_p(1),
\end{align*}
where the second equality uses $u(\bm W_i,\hat{\bm\beta})-u(\bm W_i,\bm\beta_0)=-\bm X_i^\top[d_i-(1-d_i)\dot\lambda(-\bm X_i^\top\bar{\bm\beta}_i)](\hat{\bm\beta}-\bm\beta_0)$, and the last equality follows from Lemma A.\ref{lemma.score.A1} applied to the class $\{[d-(1-d)\dot\lambda(-\bm X^\top\bm\beta)]\bm X\bm X^\top\}$ (whose limit is $H^{w}(\bm\beta_0)$), together with $\hat{\bm\beta}\overset{p}{\to}\bm\beta_0$ and the continuity of the limit at $\bm\beta_0$. Under correct specification, in particular, if $\hat{\bm\beta}$ is the maximum likelihood estimator, the first-order conditions make the left-hand side identically zero, and the lemma is then the influence expansion of Assumption (A4) rearranged (under correct specification, where $H^{w}(\bm\beta_0)=H(\bm\beta_0)$).

\begin{lemma}\label{lemma.donsker}
	Let $\mathcal T=\mathbb R$, and for $\bm\theta=(\bm\beta^\top,t)^\top \in \mathcal{B}\times \mathcal{T}$, define the process
	\[T_n(\bm\theta)=T_n(t,\bm\beta)=\frac{1}{\sqrt{n}}\sum_{i=1}^n u(\bm W_i,\bm\beta) I(\bm X_i^\top \bm\beta \leq t),\qquad f_{\bm\theta}(\bm W)=u(\bm W,\bm\beta) I(\bm X^\top \bm\beta \leq t),\]
	and the metric
	\[\rho(\bm\theta_1,\bm\theta_2) =\|\bm\beta_1 - \bm\beta_2\| + |t_1-t_2|.\]
	Under Assumptions (A1)--(A5), the class $\mathcal F=\{f_{\bm\theta}:\bm\theta\in\mathcal B\times\mathcal T\}$ is $\mathbb P$-Donsker with square-integrable envelope and the map $\bm\theta\mapsto f_{\bm\theta}$ is continuous from $(\mathcal B\times\mathcal T,\rho)$ into $L^2(P)$; consequently, the centered process
	\[\mathbb G_n(f_{\bm\theta})=\frac1{\sqrt n}\sum_{i=1}^n\bigl[f_{\bm\theta}(\bm W_i)-Pf_{\bm\theta}\bigr],\qquad \bm\theta\in\mathcal B\times\mathcal T,\]
	is stochastically $\rho$-equicontinuous: for any $\epsilon>0$,
	\[\lim_{\delta \downarrow 0} \limsup_{n\rightarrow \infty}  \mathbb P\Bigl(\sup_{\rho(\bm\theta_1,\bm\theta_2)<\delta} |\mathbb G_n(f_{\bm\theta_1}) - \mathbb G_n(f_{\bm\theta_2})|>\epsilon \Bigr) = 0.\]
	If, in addition, $Pf_{\bm\theta}=0$ for all $\bm\theta\in\mathcal B\times\mathcal T$, then $T_n(\bm\theta)=\mathbb G_n(f_{\bm\theta})$ and the same conclusion holds for the uncentered sums $T_n$. This is the case under $H_0$ for the class with $u(\bm W,\bm\beta_0)$ held fixed, because $\mathbb E[u(\bm W,\bm\beta_0)|\bm X]=0$, and for the multiplier processes below, conditionally on the data, because $\mathbb E^*[V_i]=0$. In general the deterministic part $\sqrt n(Pf_{\bm\theta_1}-Pf_{\bm\theta_2})$ does not vanish as $n\to\infty$ and is treated separately, as in the proof of Lemma A.\ref{lemma.score.A2}.
	Moreover, the same conclusion holds with $u$ replaced by any fixed measurable function $m(\bm W,\bm\beta)$ that is Lipschitz in $\bm\beta$ with square-integrable Lipschitz constant, such as $\partial_{\bm\beta}u$ or, when $a$ is Lipschitz (in particular for the polynomials in \eqref{eqn:local.alt}), $a(\bm X^\top\bm\beta)$.
\end{lemma}

\noindent{\bf Proof of Lemma A.\ref{lemma.donsker}} {\em Step 1 (bracketing entropy).} For $\bm\beta_1,\bm\beta_2\in\mathcal B$,
\begin{align*}
|u(\bm W,\bm\beta_1)-u(\bm W,\bm\beta_2)| &\le C(\bm W)\|\bm\beta_1-\bm\beta_2\|,\\[2pt]
C(\bm W)&=\|\bm X\|\bigl(1+\sup_{x}|\dot\lambda(x)|\bigr)\le 2\|\bm X\|,
\end{align*}
since $\partial_{\bm\beta}u(\bm W,\bm\beta)=-(d-(1-d)\dot\lambda(-\bm X^\top\bm\beta))\bm X$ and $-1<\dot\lambda(x)<0$ for all $x$ (the Riccati identity $\dot\lambda(x)=-\lambda(x)(x+\lambda(x))$ together with the classical bounds on the inverse Mills ratio; see, e.g., \citealp{Sampford1953}); under (A3), $\mathbb E C^2(\bm W)<\infty$. The class $\mathcal F_0=\{I(\bm X^\top \bm\beta \leq t):\bm\beta\in \mathcal{B},\ t\in \mathcal{T}\}$ is a VC-subgraph class with VC index at most $p+2$ \citep[Section 2.6]{Vaart1996Weak}, so that $\log N_{[]}(\eta,\mathcal F_0,L_2(P))\lesssim (p+2)\log(1/\eta)$. For the class $\mathcal F_1=\{u(\bm W,\bm\beta):\bm\beta\in\mathcal B\}$, cover $\mathcal B$ by $O(\eta^{-p})$ balls of diameter $\eta$ with centers $\bm\beta_j$; the pairs $[u(\bm W,\bm\beta_j)-C(\bm W)\eta,\,u(\bm W,\bm\beta_j)+C(\bm W)\eta]$ form brackets of $L_2$-size $O(\eta)$, hence $\log N_{[]}(\eta,\mathcal F_1,L_2)\lesssim p\log(1/\eta)$. For the product class $\mathcal F=\{u(\bm W,\bm\beta)I(\bm X^\top\bm\beta\le t)\}$, combine an $\eta$-bracket $[u^L,u^U]$ for $\mathcal F_1$ with an $\eta$-bracket $[I^L,I^U]$ for $\mathcal F_0$. Since the indicators are $\{0,1\}$-valued, the product bracket with lower bound $b^L:=$$\min\{u^LI^L,u^LI^U,u^UI^L,u^UI^U\}$ and upper bound $b^U:=\max\{u^LI^L,u^LI^U,u^UI^L,u^UI^U\}$, has $L_2$-size $\|b^U-b^L\|_2$, which is bounded by
\begin{align*}
\|b^U-b^L\|_2 &\le \|u^U-u^L\|_2 + \bigl\|\max(|u^L|,|u^U|)(I^U-I^L)\bigr\|_2\notag\\
&\le \eta + (\|F\|_4+\eta)\,\|I^U-I^L\|_2^{1/2} \le \eta+(\|F\|_4+\eta)\,\eta^{1/2},
\end{align*}
where the first inequality is the pointwise bound $|b^U-b^L|\le(u^U-u^L)+\max(|u^L|,|u^U|)(I^U-I^L)$, and the second uses $\max(|u^L|,|u^U|)\le F+\eta$, H\"older and $\|I^U-I^L\|_4^4=\|I^U-I^L\|_2^2$ for $\{0,1\}$-valued functions. Hence product brackets of $L_2$-size $O(\eta^{1/2})$ exist: replacing $\eta$ by $\eta^2$ gives $\log N_{[]}(\eta,\mathcal F,L_2)\lesssim(4p+4)\log(1/\eta)$. With envelope $F=\sup_{\bm\beta\in\mathcal B}|u(\bm W,\bm\beta)|\le |Y|+\|\bm X\|\sup_{\bm\beta\in\mathcal B}\|\bm\beta\|+\lambda\bigl(-\|\bm X\|\sup_{\bm\beta\in\mathcal B}\|\bm\beta\|\bigr)$ (since $\lambda$ is decreasing and $|\bm X^\top\bm\beta|\le\|\bm X\|\sup_{\bm\beta\in\mathcal B}\|\bm\beta\|$), which satisfies $\mathbb E F^4<\infty$ under (A3) because $\lambda(-t)\le t+1$ for all $t\ge0$ by the classical bounds on the inverse Mills ratio and $\mathbb E Y^4\lesssim \mathbb E\|\bm X\|^4+1<\infty$, the bracketing entropy integral $\int_0^1\sqrt{\log N_{[]}(\eta,\mathcal F,L_2)}\,d\eta$ is finite, and by \citet[Theorem 2.5.6]{Vaart1996Weak} the class $\mathcal F$ is $\mathbb P$-Donsker. This proves the Donsker part of Assumption (A6) from primitive conditions; the class $\{u(\bm W,\bm\beta)\bm X\}$ is Donsker by the same argument.

{\em Step 2 (continuity into $L^2$ and equicontinuity).} For $\theta_j=(\bm\beta_j,t_j)$, $j=1,2$, with $\rho(\theta_1,\theta_2)<\delta$,
\begin{align*}
\|f_{\theta_1}-f_{\theta_2}\|_2 &\le \|u(\bm W,\bm\beta_1)-u(\bm W,\bm\beta_2)\|_2\\[2pt]
&\quad + \|F\|_2\sqrt{P(I(\bm X^\top\bm\beta_1\le t_1)\neq I(\bm X^\top\bm\beta_2\le t_2))}.
\end{align*}
The first term is bounded by $\|C\|_2\delta$. For the second, the symmetric-difference event implies that one of the indices $\bm X^\top\bm\beta_j$ falls within distance $(\|\bm X\|+1)\delta$ of $t_j$; the probability of this event tends to $0$ as $\delta\downarrow0$ uniformly in $\bm\beta$ and $t$ by Assumption (A1) and Lemma A.\ref{lem:uniform-continuity}. Thus $\omega(\delta):=\sup_{\rho(\theta_1,\theta_2)<\delta}\|f_{\theta_1}-f_{\theta_2}\|_2\downarrow 0$. Since the empirical process indexed by a Donsker class is asymptotically equicontinuous with respect to the $L_2$-norm \citep[Section 1.5]{Vaart1996Weak}, we obtain, for every $\epsilon>0$,
	\[\lim_{\delta\downarrow0}\limsup_{n\to\infty}\mathbb P\Bigl(\sup_{\rho(\theta_1,\theta_2)<\delta}|\mathbb G_n(f_{\theta_1})-\mathbb G_n(f_{\theta_2})|>\epsilon\Bigr)=0,\]
	which is the claimed stochastic $\rho$-equicontinuity of the centered process. When $Pf_{\bm\theta}=0$ for all $\bm\theta$, $\mathbb G_n(f_{\bm\theta})=T_n(\bm\theta)$ and the conclusion transfers to the uncentered sums; otherwise the additional deterministic term $\sqrt n(Pf_{\theta_1}-Pf_{\theta_2})$ is handled separately, as in the proof of Lemma A.\ref{lemma.score.A2}. The extension to the classes listed in the statement is immediate from the same argument. $\square$

\begin{lemma}\label{lem:varid}
	 Under correct specification with parameter $\bm\beta_0$, Assumption (A3), and the notation of Sections 2.1--2.2:
	\begin{enumerate}[label=(\roman*),leftmargin=2.2em]
		\item $\mathbb E[u^2(\bm W,\bm\beta_0)|\bm X]=\ell(\bm X)=\Phi(\bm X^\top\bm\beta_0)-\Phi(-\bm X^\top\bm\beta_0)\dot\lambda(-\bm X^\top\bm\beta_0)$ a.s.;
		\item $\mathbb E[\Sigma(t_1)\Sigma(t_2)]=\mathbb E[u^2(\bm W,\bm\beta_0)I(\bm X^\top\bm\beta_0\le t_1)I(\bm X^\top\bm\beta_0\le t_2)]-G^\top(t_1,\bm\beta_0)H^{-1}(\bm\beta_0)G(t_2,\bm\beta_0)$;
		\item $\mathbb E[\Sigma(t)^2]\le\mathbb E[u^2(\bm W,\bm\beta_0)I(\bm X^\top\bm\beta_0\le t)]$, with equality if and only if $G(t,\bm\beta_0)=0$; the two kernels coincide for all $(t_1,t_2)$ if and only if $G(\cdot,\bm\beta_0)\equiv0$;
		\item $\mathbb E[\Sigma(t)\,u(\bm W,\bm\beta_0)\bm X^\top]=0$ for all $t$, so $\Sigma(t)$ and the score $u(\bm W,\bm\beta_0)\bm X$ are uncorrelated.
	\end{enumerate}
\end{lemma}

\noindent{\bf Proof of Lemma A.\ref{lem:varid}} Write $\eta=\bm X^\top\bm\beta_0$ and $\lambda=\lambda(-\eta)$, $\dot\lambda=\dot\lambda(-\eta)$. Since $d(1-d)=0$,
\[
u^2(\bm W,\bm\beta_0)=d(Y-\eta)^2+(1-d)\lambda^2.
\]
On $\{d=1\}$ we have $Y-\eta=e$ with $e>-\eta$, and the identity $\frac{d}{d\eta}[\Phi(\eta)-\eta\phi(\eta)]=\eta^2\phi(\eta)$ gives $\mathbb E[e^2\ind{e>-\eta}]=\Phi(\eta)-\eta\phi(\eta)$. Hence
\begin{align*}
\mathbb E[u^2|\bm X]&=\Phi(\eta)-\eta\phi(\eta)+\Phi(-\eta)\lambda(-\eta)^2
=\Phi(\eta)-\eta\phi(\eta)+\frac{\phi(\eta)^2}{\Phi(-\eta)}.
\end{align*}
On the other hand, the Riccati identity $\dot\lambda(-\eta)=\eta\lambda(-\eta)-\lambda(-\eta)^2$ yields
$\Phi(-\eta)\dot\lambda(-\eta)=\eta\phi(\eta)-\phi(\eta)^2/\Phi(-\eta)$, so that
$\ell(\bm X)=\Phi(\eta)-\Phi(-\eta)\dot\lambda(-\eta)=\Phi(\eta)-\eta\phi(\eta)+\phi(\eta)^2/\Phi(-\eta)$, which proves (i).
For (ii), expand $\Sigma(t_1)\Sigma(t_2)=u^2[\ind{\cdot}-G^\top(t_1)H^{-1}\bm X][\ind{\cdot}-G^\top(t_2)H^{-1}\bm X]$ and use (i) together with
\begin{align*}
\mathbb E[u^2\ind{\bm X^\top\bm\beta_0\le t}\bm X]
&=\mathbb E[\ell(\bm X)\bm X\ind{\bm X^\top\bm\beta_0\le t}]=G(t,\bm\beta_0),\\
\mathbb E[u^2\bm X\bm X^\top]&=\mathbb E[\ell(\bm X)\bm X\bm X^\top]=H(\bm\beta_0),
\end{align*}
which follow from the two equivalent population forms of $G$ and $H$ under  correct specification (Lemma A.\ref{lemma.score.A1}); the two cross terms coincide because $G^\top(t_1)H^{-1}G(t_2)$ is a scalar, and the quadratic term cancels against one of them. Part (iii) follows by taking $t_1=t_2=t$ and noting that $G^\top H^{-1}G\ge0$ with equality if and only if $G=0$; the ``for all $(t_1,t_2)$'' statement follows by taking $t_1=t_2$. Part (iv) is
$\mathbb E[u^2\ind{\cdot}\bm X^\top]-G^\top(t)H^{-1}\mathbb E[u^2\bm X\bm X^\top]=G^\top(t)-G^\top(t)=0$. $\square$

\begin{lemma}\label{lem:uniform-continuity}
	Assume that $\bm X_1$ has a component whose conditional density given the remaining components is bounded by $\bar f<\infty$, that $\mathbb E\|\bm X_1\|<\infty$, and that the coefficient of this component at $\bm\beta^\dagger\in\mathcal B$ is nonzero, say $\beta^\dagger_1\neq0$. Then there exist $\varepsilon>0$ and $C<\infty$ such that
	\[
	\sup_{\|\bm\beta-\bm\beta^\dagger\|\le\varepsilon}\ \sup_{t\in\mathbb R}\ \mathbb P\bigl(|\bm X_1^\top\bm\beta-t|\le\gamma\bigr)\le C\gamma
	\qquad\text{for all }\gamma>0.
	\]
	Moreover, for all $\bm\beta_1,\bm\beta_2$ with $\|\bm\beta_j-\bm\beta^\dagger\|\le\varepsilon$ and all $t\in\mathbb R$,
	\[\mathbb P\Bigl(\bigl\{\bm X_1^\top\bm\beta_1\le t<\bm X_1^\top\bm\beta_2\bigr\}\cup\bigl\{\bm X_1^\top\bm\beta_2\le t<\bm X_1^\top\bm\beta_1\bigr\}\Bigr)\le C\,\|\bm\beta_1-\bm\beta_2\|.\]
\end{lemma}

\noindent{\bf Proof of Lemma A.\ref{lem:uniform-continuity}} Relabel the components so that $\beta^\dagger_1\neq0$ and write $\bm X_1=(X_1,\bm X_{-1})$ with $\bm\beta=(\beta_1,\bm\beta_{-1})$. Choose $\varepsilon\le|\beta^\dagger_1|/2$, so that $|\beta_1|\ge|\beta^\dagger_1|/2$ for all $\|\bm\beta-\bm\beta^\dagger\|\le\varepsilon$. Conditioning on $\bm X_{-1}$ and using the bounded conditional density,
\[
\mathbb P\bigl(|\bm X_1^\top\bm\beta-t|\le\gamma\bigr)
=\mathbb E\Bigl[\mathbb P\bigl(|X_1\beta_1+\bm X_{-1}^\top\bm\beta_{-1}-t|\le\gamma\,\big|\,\bm X_{-1}\bigr)\Bigr]
\le \frac{2\gamma}{|\beta_1|}\,\bar f\wedge1
\le \frac{4\bar f}{|\beta^\dagger_1|}\,\gamma\wedge1,
\]
uniformly in $t$ and in $\|\bm\beta-\bm\beta^\dagger\|\le\varepsilon$. The same conditioning argument gives the symmetric-difference bound: on that event the index lies within distance $|\bm X_1^\top(\bm\beta_1-\bm\beta_2)|\le\|\bm X_1\|\,\|\bm\beta_1-\bm\beta_2\|$ of the threshold, so its conditional probability given $\bm X_{-1}$ is at most $2\bar f\|\bm X_1\|\,\|\bm\beta_1-\bm\beta_2\|/|\beta_1|$; taking expectations and using $|\beta_1|\ge|\beta^\dagger_1|/2$ and $\mathbb E\|\bm X_1\|<\infty$ gives the claim with a constant proportional to $\bar f\,\mathbb E\|\bm X_1\|/|\beta^\dagger_1|$. $\square$

\begin{lemma}\label{lem:scale-calc}
	Suppose the data are generated by the correctly specified Tobit model $Y^*=\bm X^\top\bm\beta_0+\sigma e$ with $e\sim N(0,1)$, $\sigma>0$, and let $\eta=\bm X^\top\bm\beta_0$. Then
	\[
	\mathbb E[u(\bm W,\bm\beta_0)|\bm X]=\sigma\,\phi\!\Bigl(\frac{\eta}{\sigma}\Bigr)
	-\Phi\!\Bigl(-\frac{\eta}{\sigma}\Bigr)\frac{\phi(\eta)}{\Phi(-\eta)},
	\]
	which reduces to $(\sigma-1)\phi(\eta)$ when $\eta=0$.
\end{lemma}

\noindent{\bf Proof of Lemma A.\ref{lem:scale-calc}} Here $d=\ind{e>-\eta/\sigma}$ and, on $\{d=1\}$, $Y-\eta=\sigma e$. Therefore
\begin{align*}
\mathbb E[u(\bm W,\bm\beta_0)|\bm X]
&=\mathbb E\bigl[\sigma e\,\ind{e>-\eta/\sigma}\,\big|\,\bm X\bigr]-\Phi(-\eta/\sigma)\lambda(-\eta)\\
&=\sigma\phi\!\Bigl(\frac{\eta}{\sigma}\Bigr)-\Phi\!\Bigl(-\frac{\eta}{\sigma}\Bigr)\frac{\phi(\eta)}{\Phi(-\eta)},
\end{align*}
using $\mathbb E[e\ind{e>a}]=\phi(a)$ and $\lambda(-\eta)=\phi(\eta)/\Phi(-\eta)$. At $\eta=0$ this equals $\sigma\phi(0)-\phi(0)=(\sigma-1)\phi(0)$. $\square$

\begin{lemma}\label{lem:twopar-id}
	Under correct specification with parameter $\bm\theta_0=(\bm\beta_0^\top,\sigma_0)^\top$ and the moment conditions of the two-parameter version, with $z=(Y-\bm X^\top\bm\beta_0)/\sigma_0$ and $z'=-\bm X^\top\bm\beta_0/\sigma_0$:
	\begin{enumerate}[label=(\roman*),leftmargin=2.2em]
		\item $\mathbb E[m^2(\bm W,\bm\theta_0)|\bm X]=\mathbb E[d-(1-d)\dot\lambda(z')|\bm X]$ a.s.;
		\item $\mathbb E[m(\bm W,\bm\theta_0)\bm s(\bm W,\bm\theta_0)^\top|\bm X]=-\mathbb E[\partial_{\bm\theta}m(\bm W,\bm\theta_0)|\bm X]^\top$ a.s.;
		\item $\mathbb E[\Sigma(t_1)\Sigma(t_2)]=\mathbb E[m^2(\bm W,\bm\theta_0)I(\bm X^\top\bm\beta_0\le t_1)I(\bm X^\top\bm\beta_0\le t_2)]-\mathbb G^\top(t_1,\bm\theta_0)H^{-1}(\bm\theta_0)\mathbb G(t_2,\bm\theta_0)$;
		\item $\mathbb E[\Sigma(t)^2]\le\mathbb E[m^2(\bm W,\bm\theta_0)I(\bm X^\top\bm\beta_0\le t)]$.
	\end{enumerate}
\end{lemma}

\noindent{\bf Proof of Lemma A.\ref{lem:twopar-id}} Let $\eta=\bm X^\top\bm\beta_0/\sigma_0$ and $\lambda=\lambda(-\eta)$, so that $z'=-\eta$. Since $m=dz-(1-d)\lambda$ and $d(1-d)=0$, the computation of Lemma A.\ref{lem:varid}(i) gives
\[
\mathbb E[m^2|\bm X]=\Phi(\eta)-\eta\phi(\eta)+\frac{\phi(\eta)^2}{\Phi(-\eta)}
=\Phi(\eta)-\Phi(-\eta)\dot\lambda(-\eta)
=\mathbb E[d-(1-d)\dot\lambda(z')|\bm X],
\]
which is (i). For the $\beta$-block of (ii),
\[
\mathbb E[m\,u_\beta|\bm X]=\frac{\bm X}{\sigma_0}\mathbb E[m^2|\bm X]
=-\mathbb E[\partial_\beta m|\bm X],
\qquad
\partial_\beta m=-\frac{\bm X}{\sigma_0}\{d-(1-d)\dot\lambda(z')\},
\]
by (i). For the $\sigma$-block, $\partial_\sigma m=\frac1{\sigma_0}\{-dz+(1-d)z'\dot\lambda(z')\}$, and a direct computation using $\dot\lambda(-\eta)=\eta\lambda-\lambda^2$ gives both $\mathbb E[m\,u_\sigma|\bm X]$ and $-\mathbb E[\partial_\sigma m|\bm X]$ equal to
\[
\frac1{\sigma_0}\Bigl\{\phi(\eta)(1+\eta^2)-\eta\,\frac{\phi(\eta)^2}{\Phi(-\eta)}\Bigr\},
\]
which proves (ii). Part (iii) follows from (ii), the Bartlett identity $\mathbb E[\bm s(\bm W,\bm\theta_0)\bm s(\bm W,\bm\theta_0)^\top]=H(\bm\theta_0)$ and $\mathbb E[m(\bm W,\bm\theta_0)\bm s(\bm W,\bm\theta_0)^\top\ind{\bm X^\top\bm\beta_0\le t}]=\mathbb G^\top(t,\bm\theta_0)$, using the same expansion as in Lemma A.\ref{lem:varid}(ii). Part (iv) is the case $t_1=t_2$. $\square$

\noindent{\bf Proof of Theorem \ref{Thm.weakconvergence.stat}} 
{\em Step 1 (uniform representation).} By Lemma A.\ref{lemma.score.A2}, the feasible process admits the uniform decomposition
\begin{align*}
	U_n(t,\hat{\bm\beta}) = U_n(t,\bm\beta_0) - G^\top(t,\bm\beta_0)\sqrt{n}(\hat{\bm\beta}-\bm\beta_0) + o_p(1),
\end{align*}
uniformly in $t\in\mathbb R$, where the proof of Lemma A.\ref{lemma.score.A2} shows that the smooth (mean value) part contributes $-G^\top(t,\bm\beta_0)\sqrt{n}(\hat{\bm\beta}-\bm\beta_0)$ and the indicator-jump part is $o_p(1)$ uniformly in $t$ by the centered stochastic equicontinuity of Lemma A.\ref{lemma.donsker} together with $\mathbb E[u(\bm W,\bm\beta_0)|\bm X]=0$, which centers the relevant class. Substituting the influence expansion of Assumption (A4), $\sqrt{n}(\hat{\bm\beta}-\bm\beta_0)=H^{-1}(\bm\beta_0)\frac1{\sqrt n}\sum_{i=1}^n u(\bm W_i,\bm\beta_0)\bm X_i+o_p(1)$, we obtain the first-order representation
\begin{align*}
	U_n(t,\hat{\bm\beta}) &=  \frac{1}{\sqrt{n}} \sum_{i=1}^n u(\bm{W}_i,\bm{\beta}_0)\left[ I(\bm{X}_i^\top \bm{\beta}_0 \leq t) 
	-	G^\top(t,\bm{\beta}_0) H^{-1}(\bm{\beta}_0) \bm{X}_i \right]  +o_p(1)\\[2pt]
	&\overset{\triangle}{=} U_{n0}(t,\bm\beta_0) + o_p(1),
\end{align*}
uniformly in $t\in\mathbb R$; the terms $\pm G^\top(t,\bm\beta_0)\sqrt{n}(\hat{\bm\beta}-\bm\beta_0)$ cancel exactly.

{\em Step 2 (weak convergence).} Under $H_0$ (the conditional moment restriction \eqref{eqn:SpecTestP}), $\mathbb E[u(\bm W,\bm\beta_0)|\bm X]=0$ and hence
\begin{align*}
\mathbb E[\Sigma(t)] = \mathbb E[u(\bm W,\bm\beta_0)I(\bm X^\top\bm\beta_0\le t)] - G^\top(t,\bm\beta_0)H^{-1}(\bm\beta_0)\mathbb E[u(\bm W,\bm\beta_0)\bm X] = 0.
\end{align*}
The class $\{\Sigma(t):t\in\mathbb R\}$ decomposes as
\begin{align*}
\Sigma(t) = u(\bm W,\bm\beta_0)I(\bm X^\top\bm\beta_0\le t) - \left(H^{-1}(\bm\beta_0)G(t,\bm\beta_0)\right)^\top u(\bm W,\bm\beta_0)\bm X.
\end{align*}
The first summand is a member of the Donsker class of Lemma A.\ref{lemma.donsker} (with $\bm\beta$ fixed at $\bm\beta_0$); the second is a bounded continuous functional of the fixed square-integrable random vector $u(\bm W,\bm\beta_0)\bm X$ indexed by $t$, since $t\mapsto G(t,\bm\beta_0)$ is continuous and bounded under (A1)--(A3). Hence the class $\{\Sigma(t)\}$ is $\mathbb P$-Donsker with square-integrable envelope, and by the Donsker theorem \citep[Theorem 2.5.6]{Vaart1996Weak}, $U_{n0}(t,\bm\beta_0)\Rightarrow U_\infty(t,\bm\beta_0)$ in $\ell^\infty(\mathbb R)$, where $U_\infty$ is the centered Gaussian process with covariance function $\mathbb E[\Sigma(t_1)\Sigma(t_2)]$ (the finite-dimensional convergence follows from the multivariate CLT). Combining Steps 1-2 gives $U_n(t,\hat{\bm\beta})\Rightarrow U_\infty(t,\bm\beta_0)$ in $\ell^\infty(\mathbb R)$.  Finally, $\mathbb E[\Sigma(t_1)-\Sigma(t_2)]^2\le C|t_1-t_2|$ for a finite constant $C$ under (A1)--(A3), so the limiting Gaussian process admits a continuous version by the Kolmogorov--Chentsov criterion (the Gaussian moments upgrade the second-moment bound to the required power bound).

\noindent{\bf Proof of Corollary \ref{coro.weakconv.stat}}
	Since $t\mapsto U_n(t,\hat{\bm\beta})$ is a right-continuous step function with jumps exactly at the sample points $t=\bm X_j^\top\hat{\bm\beta}$, the maximum in the definition of $KS_n$ coincides exactly with the supremum of $|U_n(t,\hat{\bm\beta})|$ over $t\in\mathbb R$ for every $n$: $KS_n$ is the value of the continuous functional $f\mapsto\sup_{t\in\mathbb R}|f(t)|$ at $U_n(\cdot,\hat{\bm\beta})$. The weak convergence of the empirical process and the continuous mapping theorem therefore give the convergence in distribution of $KS_n$ to $\sup_{t\in\mathbb R}|U_\infty(t,\bm\beta_0)|$, the supremum being finite almost surely by the tightness of the limiting process.

\noindent{\bf Proof of Theorem \ref{thm.consistency.stat}} Under $H_1$, $\hat{\bm\beta}\overset{p}{\to}\bm\beta^*$. Write 
$$\bar U_n(t)=\frac1{\sqrt n}U_n(t,\hat{\bm\beta})=\frac1n\sum_{i=1}^n u(\bm W_i,\hat{\bm\beta})I(\bm X_i^\top\hat{\bm\beta}\le t).$$ 
By the uniform law of large numbers for the Glivenko--Cantelli class $\{u(\bm W,\bm\beta)I(\bm X^\top\bm\beta\le t):\bm\beta\in\mathcal B,\ t\in\mathbb R\}$ of Lemma A.\ref{lemma.donsker},
\[\sup_{\bm\beta\in\mathcal B,\,t\in\mathbb R}\left|\frac1n\sum_{i=1}^n u(\bm W_i,\bm\beta)I(\bm X_i^\top\bm\beta\le t)-\mathbb E[u(\bm W,\bm\beta)I(\bm X^\top\bm\beta\le t)]\right|=o_p(1),\]
and, since $\hat{\bm\beta}\overset{p}{\to}\bm\beta^*$ and the function $\bm\beta\mapsto\mathbb E[u(\bm W,\bm\beta)I(\bm X^\top\bm\beta\le t)]$ is continuous, uniformly in $t$,
\begin{align*}
\frac{1}{n} \sum_{i=1}^n u(\bm{W}_i,\hat{\bm\beta}) I(\bm{X}_i^\top \hat{\bm\beta} \leq t) \overset{p}{\to} \mathbb E[u(\bm W,\bm\beta^*)I(\bm X^\top\bm\beta^*\le t)] = \mu(t,\bm\beta^*).
\end{align*}
No PEU correction term appears in this limit: since the estimator is $\sqrt n$-consistent and the mean map $\bm\beta\mapsto\mathbb E[u(\bm W,\bm\beta)I(\bm X^\top\bm\beta\le t)]$ is continuous in $\bm\beta$ uniformly in $t$ (its smooth part is Lipschitz and its indicator-jump part is H\"older of order $1/2$ by (A1)), while the centered part is stochastically equicontinuous,  we have $\sup_{t\in\mathbb R}|\bar U_n(t,\hat{\bm\beta})-\bar U_n(t,\bm\beta^*)|=o_p(1)$, so the estimation error is annihilated by the $n^{-1/2}$ normalization; equivalently, the term $-G^\top(t,\bm\beta^*)H^{-1}(\bm\beta^*)\mathbb E[u(\bm W,\bm\beta^*)\bm X]$  vanishes identically because the pseudo-true value satisfies $\mathbb E[u(\bm W,\bm\beta^*)\bm X]=0$. Since $KS_n/\sqrt n=\max_{1\le j\le n}|\bar U_n(\bm X_j^\top\hat{\bm\beta})|$, the uniform convergence of $\bar U_n$ to $\mu(\cdot,\bm\beta^*)$ and the weak convergence of the empirical distribution of $\{\bm X_i^\top\hat{\bm\beta}\}$ to the (continuous) law of $\bm X^\top\bm\beta^*$ give
\begin{align*}
\frac{KS_n}{\sqrt n} \overset{p}{\to} \sup_{t\in\mathbb R}|\mu(t,\bm\beta^*)|.
\end{align*}
Under $H_1$, $\mu(\cdot,\bm\beta^*)$ need not be nonzero: because $H_1$ is the negation of the full conditional moment restriction, an alternative may have $\mathbb E[u(\bm W,\bm\beta^*)|\bm X^\top\bm\beta^*]=0$ a.s.\ while $H_0$ fails at every parameter value; in that case the drift vanishes, and the test has no power. If, on the other hand, the index restriction fails at $\bm\beta^*$, that is, $\mu(\cdot,\bm\beta^*)\not\equiv0$, then $c:=\sup_{t\in\mathbb R}|\mu(t,\bm\beta^*)|>0$, and $KS_n\to\infty$ in probability. Consequently, for any fixed critical value $c_\tau$ (in particular, the quantiles of $\sup_{t\in\mathbb R}|U_\infty(t,\bm\beta_0)|$, which are finite by the tightness of the limiting process), $\lim_{n\to\infty}\mathbb P(KS_n>c_\tau)=1$, and the test is consistent against every fixed alternative in $H_1$ for which the index projection of the score moment at the pseudo-true value is nonzero.

\noindent{\bf Proof of Theorem \ref{thm.localpower.stat}} Under $H_{1n}$, define the conditionally centered mark 
$$\tilde u(\bm W,\bm\beta_0)=u(\bm W,\bm\beta_0)-\frac{a(\bm X^\top\bm\beta_0)}{\sqrt n}.$$
Then $\mathbb E[\tilde u(\bm W,\bm\beta_0)|\bm X]=0$ by the definition of $H_{1n}$. The decomposition of Lemma A.\ref{lemma.score.A2} remains valid under $H_{1n}$: the indicator-jump part is handled by the centered stochastic equicontinuity (applied to the conditionally centered mark $\tilde u$ introduced above) together with the fact that the remaining mean term $\mathbb E[a(\bm X^\top\bm\beta_0)(I_{\bm\beta_1,t_1}-I_{\bm\beta_2,t_2})]$ is $o(1)$ as the two indices coalesce, and $\hat{\bm\beta}-\bm\beta_0=O_p(n^{-1/2})$ for the maximum likelihood estimator under local alternatives, so uniformly in $t$,
\begin{align*}
U_n(t,\hat{\bm\beta}) = U_n(t,\bm\beta_0) - G^\top(t,\bm\beta_0)\sqrt n(\hat{\bm\beta}-\bm\beta_0) + o_p(1).
\end{align*}
Moreover,
\begin{align*}
U_n(t,\bm\beta_0)=\frac1{\sqrt n}\sum_{i=1}^n\tilde u(\bm W_i,\bm\beta_0)I(\bm X_i^\top\bm\beta_0\le t)+\frac1n\sum_{i=1}^n a(\bm X_i^\top\bm\beta_0)I(\bm X_i^\top\bm\beta_0\le t),
\end{align*}
and under $H_{1n}$ the influence expansion (A4) holds with
\begin{align*}
\sqrt n(\hat{\bm\beta}-\bm\beta_0) = H^{-1}(\bm\beta_0)\left[\frac1{\sqrt n}\sum_{i=1}^n\tilde u(\bm W_i,\bm\beta_0)\bm X_i + \mathbb E[a(\bm X^\top\bm\beta_0)\bm X]\right] + o_p(1),
\end{align*}
since $\frac1n\sum_i a(\bm X_i^\top\bm\beta_0)\bm X_i\overset{p}{\to}\mathbb E[a(\bm X^\top\bm\beta_0)\bm X]$. Substituting and collecting terms,
\begin{align*}
U_n(t,\hat{\bm\beta}) &= \frac1{\sqrt n}\sum_{i=1}^n \tilde u(\bm W_i,\bm\beta_0)\Bigl[I(\bm X_i^\top\bm\beta_0\le t)\\[2pt]
&\qquad\qquad - G^\top(t,\bm\beta_0)H^{-1}(\bm\beta_0)\bm X_i\Bigr] + D(t,\bm\beta_0) + o_p(1),
\end{align*}
uniformly in $t$, where
\begin{align*}
D(t,\bm\beta_0) &= \mathbb E[a(\bm X^\top\bm\beta_0)I(\bm X^\top\bm\beta_0\le t)]\\[2pt]
&\quad - G^\top(t,\bm\beta_0)H^{-1}(\bm\beta_0)\mathbb E[a(\bm X^\top\bm\beta_0)\bm X].
\end{align*}
The centered part converges weakly to $U_\infty(t,\bm\beta_0)$ by the Donsker argument of Theorem \ref{Thm.weakconvergence.stat} (with $u$ replaced by the conditionally centered mark $\tilde u$; since $H_{1n}$ is a contiguous perturbation of the null model, $\mathbb E[\tilde u^2(\bm W,\bm\beta_0)|\bm X]=\mathbb E[u^2(\bm W,\bm\beta_0)|\bm X]-n^{-1}a(\bm X^\top\bm\beta_0)^2$ differs from its null value by $O(n^{-1})$, so the covariance kernel converges to $\mathbb E[\Sigma(t_1)\Sigma(t_2)]$ by dominated convergence under $\mathbb E a^2<\infty$), and the drift converges uniformly by the law of large numbers for the class $\{a(\bm X^\top\bm\beta_0)I(\bm X^\top\bm\beta_0\le t)\}$, which is Glivenko--Cantelli by the continuity and envelope conditions imposed on $a$ below \eqref{eqn:local.alt}. We note that the mean-value coefficient of the decomposition, evaluated under $H_{1n}$, differs from $-G(t,\bm\beta_0)$ by $O(n^{-1/2})$: the local sequence is a contiguous $O(n^{-1/2})$-perturbation of the null model, so the conditional law of $(d,Y)$ given $\bm X$ (and hence the weight $\mathbb E[d|\bm X]$) is perturbed at most at this order; this discrepancy contributes $O_p(n^{-1/2})$ after multiplication by $\sqrt n(\hat{\bm\beta}-\bm\beta_0)=O_p(1)$ and therefore does not alter the drift $D(t,\bm\beta_0)$. Hence $U_n(t,\hat{\bm\beta})\Rightarrow U_\infty(t,\bm\beta_0)+D(t,\bm\beta_0)$ in $\ell^\infty(\mathbb R)$, and $KS_n\overset{\mathcal D}{\to}\sup_{t\in\mathbb R}|U_\infty(t,\bm\beta_0)+D(t,\bm\beta_0)|$ by the continuous mapping theorem.

\noindent{\bf Proof of Theorem \ref{thm.bootstrap.stat}}  Write
\begin{align*}
U_n^{\ast}(t,\hat{\bm\beta}) &= U_{n0}^{\ast}(t) + \mathcal R_n^*(t),\qquad
U_{n0}^{\ast}(t) = \frac{1}{\sqrt{n}}\sum_{i=1}^n V_i\, \Sigma_i(t),\\[2pt]
\Sigma_i(t) &= u(\bm W_i,\bm\beta_0) \left[I(\bm X_i^\top \bm\beta_0 \leq t) - G^\top(t,\bm\beta_0)H^{-1}(\bm\beta_0)\bm X_i\right],
\end{align*}
where $\mathcal R_n^*(t)$ collects the estimation and discretization errors.

{\em Step 1 ($\mathcal R_n^*$ is negligible, uniformly in $t$, conditionally on the data).} Conditional on the data, $\{V_i\}$ are i.i.d.\ with zero mean, so the multiplier sums below are $O_p(1)$ uniformly in $t$ by the conditional multiplier CLT \citep[Theorem 2.9.6]{Vaart1996Weak} applied to the Donsker classes of Lemma A.\ref{lemma.donsker}. We decompose $\mathcal R_n^*(t)$ as follows.
(a) {\em Correction coefficients:} the differences $G_n(t,\hat{\bm\beta})-G(t,\bm\beta_0)$ and $H_n^{-1}(\hat{\bm\beta})-H^{-1}(\bm\beta_0)$ are $o_p(1)$ uniformly in $t$ by Lemma A.\ref{lemma.score.A1}, and $\frac1{\sqrt n}\sum_i V_i u(\bm W_i,\bm\beta_0)\bm X_i=O_p(1)$; hence their contribution is $O_p(1)\cdot o_p(1)=o_p(1)$.
(b) {\em Smooth part in $\bm\beta$:} by the mean value theorem, $u(\bm W_i,\hat{\bm\beta})-u(\bm W_i,\bm\beta_0)=\partial_{\bm\beta}u(\bm W_i,\bar{\bm\beta}_i)^\top(\hat{\bm\beta}-\bm\beta_0)$ with mean-value points $\bar{\bm\beta}_i$ on the segments, and the corresponding multiplier sum (with the estimated correction coefficients) is $O_p(1)$ uniformly in $t$ (the class in Lemma A.\ref{lemma.donsker} with $u$ replaced by $\partial_{\bm\beta}u$), so this part is $O_p(1)\cdot O_p(n^{-1/2})=o_p(1)$.
(c) {\em Indicator jump:} $\frac1{\sqrt n}\sum_i V_i u(\bm W_i,\bm\beta_0)[I(\bm X_i^\top\hat{\bm\beta}\le t)-I(\bm X_i^\top\bm\beta_0\le t)]$ is $o_p(1)$ uniformly in $t$ by the conditional stochastic equicontinuity of the multiplier process on the Donsker class of Lemma A.\ref{lemma.donsker} (the argument of Lemma A.\ref{lemma.score.A2}, Step 2, applied conditionally on the data). Hence $\sup_t|\mathcal R_n^*(t)|=o_p(1)$ in conditional probability.

{\em Step 2 (conditional multiplier CLT).} Since $\{V_i\}$ satisfy (A7) and the class $\{\Sigma(t):t\in\mathbb R\}$ is $\mathbb P$-Donsker with square-integrable envelope (proof of Theorem \ref{Thm.weakconvergence.stat}, Step 2), the conditional multiplier CLT \citep[Theorem 2.9.6]{Vaart1996Weak} yields, conditional on the data, $U_{n0}^{\ast}\Rightarrow U_\infty(t,\bm\beta_0)$ in $\ell^\infty(\mathbb R)$, in probability.

{\em Step 3 (conclusion).} Combining Steps 1-2 and the continuous mapping theorem, $KS_n^{\ast}\overset{\mathcal D}{\to}\sup_{t\in\mathbb R}|U_\infty(t,\bm\beta_0)|$ conditional on the data, in probability, which is the limiting law of $KS_n$ by Corollary \ref{coro.weakconv.stat}; hence the bootstrap is valid. Under $H_{1n}$ the same argument applies verbatim: conditional on the data the bootstrap summands have zero conditional mean (as $\mathbb E[V_i]=0$), so the drift $D(t,\bm\beta_0)$ is not reproduced by the bootstrap, and the $n^{-1/2}a(\bm X_i^\top\bm\beta_0)$ component of $u(\bm W_i,\bm\beta_0)$ contributes only $n^{-1/2}\cdot O_p(1)=o_p(1)$; the bootstrap therefore still converges to the null limit $U_\infty(t,\bm\beta_0)$.

{\em Fixed alternative.} Under $H_1$, let $\bm\beta^*=\mathrm{plim}\,\hat{\bm\beta}$ and, for $t\in\mathbb R$, set $\Sigma_i^*(t)=u(\bm W_i,\bm\beta^*)[I(\bm X_i^\top\bm\beta^*\le t)-G^\top(t,\bm\beta^*)H^{-1}(\bm\beta^*)\bm X_i]$, with $G(\cdot,\bm\beta^*)$ and $H(\bm\beta^*)$ as in the statement. Replacing $(\bm\beta_0,G,H)$ by $(\bm\beta^*,G,H)$ in parts (a)-(c) of Step 1 gives
\[\sup_{t\in\mathbb R}\left|U_n^{\ast}(t,\hat{\bm\beta})-\frac1{\sqrt n}\sum_{i=1}^n V_i\,\Sigma_i^*(t)\right|=o_p(1)\]
in conditional probability. Indeed, $G_n(t,\hat{\bm\beta})\overset{p}{\to}G(t,\bm\beta^*)$ and $H_n(\hat{\bm\beta})\overset{p}{\to}H(\bm\beta^*)$ uniformly in $t$ by the Glivenko--Cantelli property of the class $\{[d-(1-d)\dot\lambda(-\bm X^\top\bm\beta)]\bm X I(\bm X^\top\bm\beta\le t)\}$ of Lemma A.\ref{lemma.donsker} together with $\hat{\bm\beta}\overset{p}{\to}\bm\beta^*$; here the limits are the random-weight forms, because the identity $\mathbb E[d|\bm X]=\Phi(\bm X^\top\bm\beta)$ used for the equivalent expressions in Lemma A.\ref{lemma.score.A1} holds only under  correct specification, and the remaining terms are handled as in parts (b)-(c) using $\hat{\bm\beta}-\bm\beta^*=O_p(n^{-1/2})$ and the Donsker property of the classes with $\partial_{\bm\beta}u$ and the indicators. Since the class $\{\Sigma^*(t;\cdot):t\in\mathbb R\}$ is Donsker by the bracketing argument of Lemma A.\ref{lemma.donsker} and
\[\frac1n\sum_{i=1}^n\Sigma_i^*(t_1)\Sigma_i^*(t_2)\overset{p}{\to}\mathbb E[\Sigma^*(t_1;\bm\beta^*)\Sigma^*(t_2;\bm\beta^*)],\]
the conditional multiplier CLT \citep[Theorem 2.9.6]{Vaart1996Weak} yields $U_n^{\ast}(t,\hat{\bm\beta})\Rightarrow U_\infty^{\ast}(t,\bm\beta^*)$ in $\ell^\infty(\mathbb R)$, in conditional probability. The continuous mapping theorem gives the limit of $KS_n^{\ast}$, and $KS_n\to\infty$ in probability under alternatives in $H_1$ with $\mu(\cdot,\bm\beta^*)\not\equiv0$ by Theorem \ref{thm.consistency.stat}; hence the bootstrap $p$-value converges to zero in probability and the test is consistent against every such alternative. $\square$

\noindent{\bf Proof of Theorem \ref{Thm.weakconvergence.twopar}} The proof is identical in structure to that of Theorem \ref{Thm.weakconvergence.stat} with $u(\bm W,\bm\beta)$, $\bm X$, and $(G,H)$ replaced by $m(\bm W,\bm\theta)$, the joint score $\bm s$, and $(\mathbb G,H)$. We spell out the two steps that differ. 

{\em Step 1 (uniform representation).} Writing 
$$U_n(t,\hat{\bm\theta})=\frac1{\sqrt n}\sum_i m(\bm W_i,\hat{\bm\theta})I(\bm X_i^\top\hat{\bm\beta}\le t),$$ 
the mean-value and equicontinuity arguments of Lemma A.\ref{lemma.score.A2} (applied to the mark $m$, which is Lipschitz in $\bm\theta$ with square-integrable constant under the assumed moment conditions and whose indicator-jump part is centered because $\mathbb E[m(\bm W,\bm\theta_0)|\bm X]=0$ under correct specification) give 
$$U_n(t,\hat{\bm\theta})=U_n(t,\bm\theta_0)-\mathbb G^\top(t,\bm\theta_0)\sqrt n(\hat{\bm\theta}-\bm\theta_0)+o_p(1)$$ uniformly in $t$. The score expansion 
$$\sqrt n(\hat{\bm\theta}-\bm\theta_0)=H^{-1}(\bm\theta_0)\frac1{\sqrt n}\sum_i\bm s(\bm W_i,\bm\theta_0)+o_p(1)$$ 
is the standard maximum-likelihood expansion for $(\bm\beta,\sigma)$ (with the information matrix $H$ of \eqref{eqn:two-param-H}). Combining the two expansions and canceling the PEU terms gives the first-order representation of the feasible process, uniformly in $t$:
\[U_n(t,\hat{\bm\theta}) = \frac1{\sqrt n}\sum_{i=1}^n\left[m(\bm W_i,\bm\theta_0)I(\bm X_i^\top\bm\beta_0\le t) - \mathbb G^\top(t,\bm\theta_0)H^{-1}(\bm\theta_0)\bm s(\bm W_i,\bm\theta_0)\right] + o_p(1).\]

{\em Step 2 (weak convergence).} Under correct specification $\mathbb E[\bm s(\bm W,\bm\theta_0)|\bm X]=0$ a.s., hence $\mathbb E[\Sigma(t)]=\mathbb E[mI]-\mathbb G^\top H^{-1}\mathbb E[\bm s]=0$. The class $\{\Sigma(t):t\in\mathbb R\}$ is $\mathbb P$-Donsker with square-integrable envelope by the bracketing argument of Lemma A.\ref{lemma.donsker} (the mark and the score are Lipschitz in $\bm\theta$, and the indicator class is VC). The Donsker theorem yields the weak convergence. $\square$

\noindent{\bf Proofs of Theorems \ref{thm.twopar.consistency}--\ref{thm.twopar.bootstrap}} These follow from the proofs of Theorems \ref{thm.consistency.stat}, \ref{thm.localpower.stat}, and \ref{thm.bootstrap.stat} with the substitutions above. Three points deserve mention. First, in the consistency theorem the drift is the mark moment $\mu(t,\bm\theta^*)=\mathbb E[m(\bm W,\bm\theta^*)I(\bm X^\top\bm\beta^*\le t)]$, exactly as in the proof of Theorem \ref{thm.consistency.stat}:  the estimation error is annihilated by the $n^{-1/2}$ normalization (by the same uniform-continuity and centered-equicontinuity argument as in the proof of Theorem \ref{thm.consistency.stat}), and no PEU correction survives; the uniform law of large numbers for the class $\{m(\bm W,\bm\theta)I(\bm X^\top\bm\beta\le t)\}$, together with $\hat{\bm\theta}\overset{p}{\to}\bm\theta^*$, gives the result. Second, the drift in the local-power theorem is obtained as follows. Under the local alternative write $m(\bm W_i,\bm\theta_0)=\tilde m_i+n^{-1/2}a(\bm X_i^\top\bm\beta_0)$ and $\bm s(\bm W_i,\bm\theta_0)=\tilde{\bm s}_i+n^{-1/2}\bm q(\bm X_i)$ with $\mathbb E[\tilde m_i|\bm X_i]=0$ and $\mathbb E[\tilde{\bm s}_i|\bm X_i]=0$. The influence expansion gives
\[\sqrt n(\hat{\bm\theta}-\bm\theta_0)=H^{-1}(\bm\theta_0)\Bigl[\frac1{\sqrt n}\sum_i\tilde{\bm s}_i+\mathbb E[\bm q(\bm X)]\Bigr]+o_p(1),\]
and substituting it into the uniform decomposition, the terms involving $n^{-1/2}\sum_i\tilde{\bm s}_i$ are absorbed by the first-order representation exactly as in the null case, leaving
\[U_n(t,\hat{\bm\theta})=\frac1{\sqrt n}\sum_i\Bigl[\tilde m_i I(\bm X_i^\top\bm\beta_0\le t)-\mathbb G^\top(t,\bm\theta_0)H^{-1}(\bm\theta_0)\tilde{\bm s}_i\Bigr]+D(t,\bm\theta_0)+o_p(1)\]
uniformly in $t$, where $D(t,\bm\theta_0)=\mathbb E[a(\bm X^\top\bm\beta_0)I(\bm X^\top\bm\beta_0\le t)]-\mathbb G^\top(t,\bm\theta_0)H^{-1}(\bm\theta_0)\mathbb E[\bm q(\bm X)]$. The centered part converges weakly to $U_\infty$, and the deterministic drift $D$ is not reproduced by the multiplier bootstrap (whose summands have conditional mean zero), exactly as in the one-parameter case. Third, the bootstrap proof of Theorem \ref{thm.bootstrap.stat} applies with the class of functions $m(\bm W,\bm\theta)I(\bm X^\top\bm\beta\le t)-\mathbb G^\top(t,\bm\theta)H^{-1}(\bm\theta)\bm s(\bm W,\bm\theta)$ in place of $\{\Sigma(t)\}$, which is Donsker by the bracketing argument of Lemma A.\ref{lemma.donsker} (the mark and the score are Lipschitz in $\bm\theta$ with square-integrable Lipschitz constants under the assumed moments). The fixed-alternative statement of the bootstrap theorem follows from the corresponding paragraph of the proof of Theorem \ref{thm.bootstrap.stat} with $(\bm\beta^*,G,H)$ replaced by $(\bm\theta^*,\mathbb G,H)$. As in the consistency case, under $H_1^{2p}$ the limits of the sample versions of $\mathbb G$ and $H$ are the random-weight forms evaluated at $\bm\theta^*$. $\square$



\end{spacing}

\end{document}